\newif\ifsubmission
 \submissionfalse
 
\newif\ifanon
  \anonfalse

\newif\ifnotes
\notesfalse

\def\showtableofcontents{1}

\ifsubmission
\documentclass[orivec,runningheads,envcountsect]{llncs}
\else
\documentclass[11pt]{article}
\usepackage{fullpage}
\fi

\usepackage[utf8]{inputenc}

\usepackage[shortlabels]{enumitem}
\usepackage{amsmath}
\usepackage{amssymb}
\usepackage{amsthm}
\usepackage[dvipsnames]{xcolor}
\usepackage{mathtools}
\usepackage[utf8]{inputenc}
\usepackage{amsfonts}
\usepackage{graphicx}
\usepackage{mathrsfs}
\usepackage{physics}
\usepackage{soul}
\usepackage{bm} 
\usepackage[ruled,linesnumbered]{algorithm2e}
\usepackage[T1]{fontenc}
\usepackage{quantikz} 

\ifsubmission
\else
\usepackage{beton}
\fi 

\definecolor{DarkBlue}{RGB}{0,0,150}
\definecolor{DarkRed}{RGB}{150,0,0}
\definecolor{DarkGreen}{RGB}{0,150,0}

\usepackage[colorlinks,linkcolor=Red,citecolor=DarkBlue]{hyperref}

\usepackage[capitalize]{cleveref}
\crefname{step}{Step}{Steps}
\usepackage[hyperpageref]{backref}

\newcommand{\authnote}[3]{\textcolor{#3}{[{\footnotesize {\bf #1:} { {#2}}}]}}

\newcommand{\thomas}[1]{\ifnotes \authnote{Thomas}{#1}{Emerald} \fi}

\ifsubmission
\else
\newtheorem{theorem}{Theorem}[section]

\newtheorem{claim}[theorem]{Claim}
\newtheorem{lemma}[theorem]{Lemma}

\newtheorem{corollary}[theorem]{Corollary}

\newtheorem{definition}[theorem]{Definition}

\newtheorem{remark}[theorem]{Remark}

\fi

\theoremstyle{definition}

\theoremstyle{remark}

\newcommand{\Gen}{\mathsf{Gen}}
\def\Eval{\mathsf{Eval}}
\newcommand{\Invert}{\mathsf{Invert}}
\def\Check{\mathsf{Check}}
\def\clawfree{\mathsf{ClawFree}}
\def\vecx{\mathbf{x}}
\def\vecy{\mathbf{y}}
\def\vecd{\mathbf{d}}
\def\inj{\mathsf{injective}}

\def\twotoone{\mathsf{2}\text{-}\mathsf{to}\text{-}\mathsf{1}}

\newcommand{\Enc}{\mathsf{Enc}}
\newcommand{\Dec}{\mathsf{Dec}}
\newcommand{\ct}{\mathsf{ct}}
\newcommand{\FHE}{\mathsf{FHE}}
\newcommand{\Ext}{\mathsf{Ext}}

\def\cA{{\cal A}}

\def\cC{{\cal C}}

\def\cF{{\cal F}}

\def\cL{{\cal L}}

\def\cR{{\cal R}}

\def\cV{{\cal V}}

\newcommand{\RegA}{\mathcal{A}}
\newcommand{\RegB}{\mathcal{B}}
\newcommand{\RegC}{\mathcal{C}}

\newcommand{\RegH}{\mathcal{H}}
\newcommand{\RegI}{\mathcal{I}}

\newcommand{\RegP}{\mathcal{P}}

\newcommand{\RegU}{\mathcal{U}}

\newcommand{\RegW}{\mathcal{W}}
\newcommand{\RegX}{\mathcal{X}}
\newcommand{\RegY}{\mathcal{Y}}
\newcommand{\RegZ}{\mathcal{Z}}

\newcommand{\textdef}[1]{\textnormal{\textsf{#1}}}

\newcommand{\QMA}{\mathsf{QMA}}
\newcommand{\NP}{\mathsf{NP}}

\newcommand{\secp}{\lambda}
\newcommand{\poly}{\mathsf{poly}}
\newcommand{\polylog}{\mathsf{polylog}}
\newcommand{\negl}{\mathsf{negl}}
\newcommand{\pk}{\mathsf{pk}}
\newcommand{\sk}{\mathsf{sk}}
\newcommand{\mode}{\mathsf{mode}}
\newcommand{\Setup}{\mathsf{Setup}}
\newcommand{\extpk}{\mathsf{ExtPk}}
\newcommand{\extsk}{\mathsf{ExtSk}}
\newcommand{\Program}{\mathsf{Program}}
\newcommand{\PK}{\mathsf{PK}}
\newcommand{\SK}{\mathsf{SK}}
\newcommand{\iO}{\mathsf{i}\mathcal{O}}
\newcommand{\Hyb}{\mathsf{Hyb}}

\newcommand{\MP}{\mathsf{MP}} 
\newcommand{\Test}{\mathsf{Test}}
\newcommand{\Out}{\mathsf{Out}}
\newcommand{\Commit}{\mathsf{Commit}}
\newcommand{\Open}{\mathsf{Open}}

\newcommand{\Id}{\mathsf{Id}}

\newcommand{\brho}{\bm{\rho}}
\newcommand{\bsigma}{\bm{\sigma}}
\newcommand{\btau}{\bm{\tau}}

\newcommand{\tensor}{\otimes}

\DeclareMathOperator*{\E}{\mathbb{E}}

\newcommand{\PRF}{\mathsf{PRF}}

\def\QMA{\textbf{QMA}}

\title{Succinct Classical Verification of Quantum Computation}

\ifsubmission
\author{Anonymous Submission, Full Version}
\else

\ifanon
\author{Anonymous Submission, Full Version}
\date{}
\else

\author{James Bartusek\thanks{UC Berkeley. Email:\texttt{ bartusek.james@gmail.com}.} \and Yael Tauman Kalai\thanks{Microsoft Research and MIT. Email: \texttt{yael@microsoft.com}.} \and Alex Lombardi\thanks{MIT. Email: \texttt{alexjl@mit.edu}.} \and Fermi Ma\thanks{Simons Institute and UC Berkeley. Email: \texttt{fermima@alum.mit.edu}.}\and Giulio Malavolta\thanks{Max Planck Institute for Security and Privacy. Email: \texttt{giulio.malavolta@hotmail.it}.} \and Vinod Vaikuntanathan\thanks{MIT. Email: \texttt{vinodv@mit.edu}.} \and Thomas Vidick\thanks{Caltech. Email: \texttt{vidick@caltech.edu}.} \and Lisa Yang\thanks{MIT. Email; \texttt{lisayang@mit.edu}.}}
\date{\today}
\fi
\fi

\begin{document}
\maketitle

\begin{abstract} 
We construct a classically verifiable \emph{succinct} interactive argument for quantum computation (BQP) with communication complexity and verifier runtime that are poly-logarithmic in the runtime of the BQP computation (and polynomial in the security parameter). Our protocol is secure assuming the post-quantum security of indistinguishability obfuscation (iO) and Learning with Errors (LWE). This is the first succinct argument for quantum computation \emph{in the plain model}; prior work (Chia-Chung-Yamakawa, TCC '20) requires both a long common reference string and non-black-box use of a hash function modeled as a random oracle.

At a technical level, we revisit the framework for constructing classically verifiable quantum computation (Mahadev, FOCS '18). We give a self-contained, modular proof of security for Mahadev's protocol, which we believe is of independent interest. Our proof readily generalizes to a setting in which the verifier's first message (which consists of many public keys) is \emph{compressed}. Next, we formalize this notion of compressed public keys; we view the object as a generalization of constrained/programmable PRFs and instantiate it based on indistinguishability obfuscation. 

Finally, we compile the above protocol into a fully succinct argument using a (sufficiently composable) succinct argument of knowledge for NP. Using our framework, we achieve several additional results, including

\begin{itemize}
    \item Succinct arguments for QMA (given multiple copies of the witness),
    \item Succinct \emph{non-interactive} arguments for BQP (or QMA) in the quantum random oracle model, and
    \item Succinct batch arguments for BQP (or QMA) assuming post-quantum LWE (without iO).
\end{itemize}

\end{abstract}

\ifsubmission
\else
\ifnum\showtableofcontents=1
{
\thispagestyle{empty}
\newpage
\pagenumbering{roman}
\setcounter{tocdepth}{2}
\tableofcontents
\newpage
\pagenumbering{arabic}
}
\fi

\section{Introduction}
\label{sec:introduction}

Efficient verification of computation is one of the most fundamental and intriguing concepts in computer science, and lies at the heart of the $\mathrm{P}$ vs.\ $\mathrm{NP}$ question. It has been studied in the classical setting for over three decades, giving rise to beautiful notions such as interactive proofs \cite{STOC:GolMicRac85}, multi-prover interactive proofs~\cite{STOC:BGKW88}, probabilistically checkable proofs \cite{FOCS:BabForLun90,FOCS:ALMSS92,FOCS:AroSaf92}, and culminating with the notion of a \emph{succinct} (interactive and non-interactive) \emph{argument}~\cite{STOC:Kilian92,Micali94}. Roughly speaking, a succinct  argument for a $T$-time computation enables a prover running in $\poly(T)$ time to convince a $\polylog(T)$-time verifier of the correctness of the computation using only $\polylog(T)$ bits of communication, with soundness against all polynomial-time cheating provers.

In a breakthrough result in 2018, Mahadev~\cite{FOCS:Mahadev18a} presented an interactive argument system that enables a classical verifier to check the correctness of an arbitrary \emph{quantum} computation. Mahadev's protocol represents a different kind of interactive argument --- unlike the traditional setting in which the prover simply has more \emph{computational resources} (i.e., running time) than the verifier, the prover in Mahadev's protocol works in a qualitatively more powerful \emph{computational model}. More precisely, for any $T$-time quantum computation, Mahadev's protocol enables a quantum prover running in time $\poly(T)$ to convince a classical $\poly(T)$-time verifier with $\poly(T)$ bits of classical communication. Soundness holds against all quantum polynomial-time cheating provers under the post-quantum hardness of the learning with errors (LWE) problem. 

A fundamental question is whether we can get the best of both worlds: can the prover have \emph{both} a more powerful computational model \emph{and} significantly greater computational resources? Namely, we want an interactive argument system for $T$-time quantum computation in which the quantum prover runs in $\poly(T)$ time and convinces a $\polylog(T)$-time classical verifier with $\polylog(T)$ bits of classical communication.

We answer this question affirmatively, both for $\poly(T)$-time quantum computations, corresponding to the complexity class $\mathbf{BQP}$, and also for the non-deterministic analog $\mathbf{QMA}$. 

\begin{theorem}[Succinct Arguments for $\mathbf{BQP}$]\label{thm:main1:informal}
  Let $\secp$ be a security parameter.
  Assuming the existence of a  post-quantum secure indistinguishability obfuscation scheme {\em (iO)} and the post-quantum hardness of the learning with errors problem {\em (LWE)}, there is an interactive argument system for any $T$-time quantum computation on input $x$,\footnote{A $T$-time quantum computation is a \emph{language} $L$ decidable by a bounded-error $T$-time quantum Turing machine \cite{BV97}. We leave it to future work to address more complex tasks such as \emph{sampling} problems (as in \cite{CLLW20}).} where 
  \begin{itemize}
      \item the prover is quantum and runs in time $\poly(T,\lambda)$,
      \item the verifier is classical and runs in time $\mathsf{poly}(\log T, \secp) + \tilde{O}(|x|)$,\footnote{As in the classical setting, some dependence on $|x|$ is necessary at least to read the input; as in \cite{STOC:Kilian92}, we achieve a fairly minimal $|x|$-dependence.} and
      \item the protocol uses $\mathsf{poly}(\log T, \secp)$ bits of classical communication.
  \end{itemize}
\end{theorem}

\begin{theorem}[Succinct Arguments for $\mathbf{QMA}$]\label{thm:main2:informal}
  Assuming the existence of a post-quantum secure indistinguishability obfuscation scheme {\em (iO)} and the post-quantum hardness of the learning with errors problem {\em (LWE)}, there is an interactive argument system for any $T$-time quantum computation on input $x$ and a $\poly(T)$-qubit witness, where 
  \begin{itemize}
      \item the prover is quantum and runs in time $\poly(T,\lambda)$, using polynomially many copies of the witness,\footnote{We inherit the need for polynomially-many copies of the witness from prior works. This is a feature common to all previous classical verification protocols, and even to the quantum verification protocol of \cite{FHM18}.}
      \item the verifier is classical and runs in time $\mathsf{poly}(\log T, \secp) + \tilde{O}(x)$, and
      \item the protocol uses $\mathsf{poly}(\log T, \secp)$ bits of classical communication.
  \end{itemize}
\end{theorem}

\paragraph{A New Proof of Security for the \cite{FOCS:Mahadev18a} Protocol.}  One might hope to prove~\cref{thm:main1:informal,thm:main2:informal} by treating the Mahadev result as a ``black box'' and showing that \emph{any} (classical) interactive argument for quantum computations can be compressed into a succinct protocol via a suitable cryptographic compiler. This is especially appealing given the extremely technical nature of Mahadev’s security proof. Unfortunately, for reasons that will become clear in the technical overview, this kind of generic compilation seems unlikely to be achievable in our setting. Even worse, there does not appear to be any easily formalized property of the Mahadev protocol that would enable such a compilation. 

Instead, our solution consists of two steps. 

\begin{enumerate}[(1)]
    \item We build a modified variant of the \cite{FOCS:Mahadev18a} protocol and give an entirely self-contained proof of security. This modified protocol satisfies a few technical conditions that the original \cite{FOCS:Mahadev18a} does not; most prominently, the \emph{first verifier message} of our modified protocol is already succinct.
    \item We give a generic compiler that converts the protocol from Step (1) into a succinct argument system.
\end{enumerate} 

Our Step (1) also results in a self-contained proof of security of the original \cite{FOCS:Mahadev18a} protocol that is more modular and amenable to further modification and generalization, which we believe will be useful for future work. Our analysis builds upon \cite{FOCS:Mahadev18a} itself as well as an alternative approach described in Vidick's (unpublished) lecture notes~\cite{Vid20-course}. A concrete consequence of our new proof is that one of the two ``hardcore bit'' security requirements of the main building block primitive (``extended noisy trapdoor claw-free functions'') in \cite{FOCS:Mahadev18a} is not necessary. %

\paragraph{Additional Results.} Beyond our main result of succinct arguments for $\mathbf{BQP}$ and $\mathbf{QMA}$, we explore a number of extensions and obtain various new protocols with additional properties.
\begin{itemize}
    \item {\em Non-Interactive}: Although our protocols are not public-coin, we show how to modify them in order to apply the Fiat-Shamir transformation and round-collapse our protocols. As a result, we obtain designated-verifier non-interactive arguments for $\mathbf{BQP}$ (and the non-deterministic analog $\mathbf{QMA}$) with security in the quantum random oracle model (QROM).
    \item {\em Zero-Knowledge}: We show how to lift both variants of our protocol (interactive and non-interactive) to achieve zero-knowledge. We show a generic transformation based on classical two-party computation for reactive functionalities that makes our protocols simulatable. This transformation does not add any new computational assumption to the starting protocol.
    \item {\em Batch Arguments from LWE}: For the case of batch arguments, i.e., where the parties engage in the parallel verification of $n$ statements, 
    we show a succinct protocol that only assumes the post-quantum hardness of LWE (without iO). In this context, succinctness requires that the verifier's complexity scales with the size of a \emph{single instance}, but is independent of $n$.
\end{itemize}

\paragraph{Prior Work.} As discussed above, Mahadev \cite{FOCS:Mahadev18a} constructs a \emph{non-succinct} argument system for \textbf{BQP}/\textbf{QMA} under LWE. The only prior work addressing \emph{succinct} classical arguments for quantum computation is the recent work of Chia, Chung and Yamakawa~\cite{TCC:ChiChuYam20}. \cite{TCC:ChiChuYam20} constructs a classically verifiable argument system for quantum computation in the following setting:

\begin{itemize}
    \item The prover and verifier share a $\poly(T)$-bits long, structured reference string (which requires a trusted setup to instantiate) along with a hash function $h$ (e.g. SHA-3).
    \item The ``online communication'' of the protocol is succinct ($\poly(\log T)$).
    \item Security is heuristic: it can be proved when $h$ is modeled as a random oracle, but the \emph{protocol description itself} explicitly requires the code of $h$ (i.e. uses $h$ in a non-black-box way).
\end{itemize}
We specifically note that when viewed in the \emph{plain model} (i.e., without setup), the verifier must send the structured reference string to the prover, resulting in a protocol that is \emph{not succinct}. We note that \cite{TCC:ChiChuYam20} was specifically optimizing for a \emph{two-message} protocol, but their approach seems incapable of achieving succinctness in the plain model even if further interaction is allowed. 

By contrast, our succinct interactive arguments are in the plain model and are secure based on well-formed cryptographic assumptions, and our succinct $2$-message arguments are proved secure in the QROM (and do not require a long common reference string).

Finally, we remark that our approach to achieving succinct arguments fundamentally (and likely necessarily) differs from \cite{TCC:ChiChuYam20} because we manipulate the ``inner workings'' of the \cite{FOCS:Mahadev18a} protocol; by contrast \cite{TCC:ChiChuYam20} makes ``black-box'' use of a specific soundness property of the \cite{FOCS:Mahadev18a} protocol (referred to as ``computational orthogonality'' by \cite{TCC:ACGH20}) and is otherwise agnostic to how the protocol is constructed. 

\paragraph{Acknowledgments.} AL is supported in part by a Charles M. Vest fellowship. GM is partially supported by the German Federal Ministry of Education and Research BMBF (grant 16K15K042, project 6GEM). TV is supported by AFOSR YIP award number FA9550-16-1-0495, a grant from the Simons
Foundation (828076, TV), MURI Grant FA9550-18-1-0161, the NSF QLCI program through grant number OMA-2016245 and the IQIM, an NSF Physics Frontiers Center (NSF Grant PHY-1125565) with support of the Gordon
and Betty Moore Foundation (GBMF-12500028). AL, VV, and LY are supported in part by DARPA under Agreement No. HR00112020023,
a grant from MIT-IBM Watson AI, a grant from Analog Devices, a Microsoft Trustworthy AI grant and the Thornton Family Faculty Research Innovation Fellowship. Any
opinions, findings and conclusions or recommendations expressed in this material are those of the author(s) and do
not necessarily reflect the views of the United States Government or DARPA. LY was supported in part by an NSF graduate research fellowship.

\section{Technical Overview} 
\label{sec:techoverview}

Our starting point is Mahadev's protocol for classical verification of quantum computation~\cite{FOCS:Mahadev18a}, the core ingredient of which is a {\em measurement protocol}.

\subsection{Recap: Mahadev's Measurement Protocol}

We begin by reviewing Mahadev's $N$-qubit measurement protocol. In Mahadev's protocol, a quantum prover holding an $N$-qubit quantum state $\brho$ interacts with a classical verifier, who wants to obtain the result of measuring $\brho$ according to measurement bases $h \in \{0,1\}^N$ ($h_i$ specifies a basis choice for the $i$th qubit, with $h_i = 1$ corresponding to the Hadamard basis and $h_i = 0$ corresponding to the standard basis).

\paragraph{Trapdoor Claw-Free Functions.} At the heart of the protocol is a cryptographic primitive known as an \emph{injective/claw-free trapdoor function} (a variant of lossy trapdoor functions~\cite{STOC:PeiWat08,C:PeiVaiWat08,STOC:GorVaiWic15}), which consists of two trapdoor function families $\mathsf{Inj}$ (for injective) and $\mathsf{Cf}$ (for claw-free), with the following syntactic requirements:\footnote{The actual syntactic requirements, described in Section~\ref{sec:mahadev-rtcf}, are somewhat more complex due to the fact that the functions in question are probabilistic.}
\begin{itemize}
    \item Each function in $\mathsf{Cf}\cup \mathsf{Inj}$ is indexed by a public-key $\pk$, where functions $f_{\pk} \in \mathsf{Inj}$ are injective and functions  $f_{\pk} \in \mathsf{Cf}$ are two-to-one. Moreover, $\pk$ can be sampled along with a secret key $\sk$ that enables computing $f^{-1}_{\pk}$ (i.e., $f^{-1}_{\pk}(y)$ consists of a single pre-image if $f_{\pk} \in \mathsf{Inj}$, and two pre-images if $f_{\pk} \in \mathsf{Cf}$).
    \item All functions in $\mathsf{Inj}$ and $\mathsf{Cf}$ have domain $\{0,1\}^{\ell+1}$ (for some $\ell$) and  the two pre-images of $y$ under $f_{\pk} \in \mathsf{Cf}$ are of the form $(0,x_0)$ and $(1,x_1)$ for some $x_0,x_1\in\{0,1\}^\ell$.
\end{itemize}
An injective/claw-free trapdoor function must satisfy the following security properties:\footnote{In fact, Mahadev's proof relies on two different hardcore bit properties, but we show in this work that only the adaptive hardcore bit property is needed.}
\begin{enumerate}
    \item \textbf{Claw-Free/Injective Indistinguishability.} A random function in $f_{\pk} \gets \mathsf{Cf}$ is computationally indistinguishable from a random function $f_{\pk} \gets \mathsf{Inj}$.
    \item \textbf{Adaptive Hardcore Bit.} Given $f_{\pk} \gets \mathsf{Cf}$, it is computationally infeasible to output both (1) a pair $(x,y)$ satisfying $f_{\pk}(x) = y$ and (2) a non-zero string $d \in \{0,1\}^{\ell+1}$ such that $d \cdot (1,x_0 \oplus x_1) = 0$, where $(0,x_0)$ and $(1,x_1)$ are the two preimages of $y$.\footnote{The full definition places a slightly stronger restriction on $d$ than simply being non-zero. However, this simplified version will suffice for this overview.}
\end{enumerate}

To build some intuition about the usefulness of such function families, notice that they can be used to commit to a single classical bit quite easily. The commitment key is a function $f_{\pk} \in \mathsf{Inj}$, and commitment to a bit $b$ is $y = f_{\pk}(b,x)$ for a random $x \in \{0,1\}^n$. It is not hard to verify that this is a statistically binding and computationally hiding commitment of $b$. On the other hand, if $f_{\pk} \in \mathsf{Cf}$, it is a statistically hiding and computationally binding commitment of $b$.\footnote{In particular, $f_{\pk} \in \mathsf{Cf}$ satisfies Unruh's definition of \emph{collapse-binding}~\cite{EC:Unruh16}.}

\paragraph{Protocol Description.} With this intuition in mind, we now describe a (slightly simplified version of) the Mahadev measurement protocol. Mahadev's protocol is a many-fold sequential repetition of the following one-bit challenge protocol.

  \begin{description}
  \item[$\mathbf{V} \rightarrow \mathbf{P}$:] The verifier samples $N$ pairs of functions and their trapdoors: for every $i\in[\ell]$, he samples $(\pk_i,\sk_i) \gets \mathsf{Inj}$ if $h_i =0$ and $(\pk_i,\sk_i) \gets \mathsf{Cf}$ if $h_i = 1$. He sends the $N$-tuple of public keys (function descriptions) \underline{$(\pk_1,\ldots,\pk_{N})$} to the prover.

  \item[$\mathbf{P} \rightarrow \mathbf{V}$:] The prover holds an $N$-qubit state $\brho$ on register $\RegB = \RegB_1 \otimes \cdots \RegB_N$. She initializes registers $\RegX = \RegX_1 \otimes \cdots \RegX_N$ so that each $\RegX_i$ is an $\ell$-qubit register in uniform superposition. Then she evaluates each $f_{\pk_i}$ on the $\ell+1$ qubit register $\RegB_i \otimes \RegX_i$ and measures the output to obtain $y_i$. Finally, she sends $\underline{(y_1,\ldots,y_{N})}$ to the verifier, keeping the registers $\RegB,\RegX$ for herself.
  
  \item[$\mathbf{V} \rightarrow \mathbf{P}$:] The verifier picks a random challenge bit \underline{$c\gets \{0,1\}$} and sends it to the prover. $c=0$ indicates a ``test round'' and $c=1$ indicates a ``measurement round''.
  \item[$\mathbf{P} \rightarrow \mathbf{V}$:] If $c=0$ (``test''), the prover measures each $\RegB_i \otimes \RegX_i$ in the standard basis, obtaining $(b_i,x_i) \in \{0,1\} \times \{0,1\}^\ell$, and sends \underline{$(b_i,x_i)_{i=1}^{N}$} to the verifier.
  
  If $c=1$ (``measure''), she measures each $\RegB_i \otimes \RegX_i$ in the Hadamard basis (i.e., she applies the Hadamard gate $H$ on each qubit and then measures in the standard basis), obtaining $d_i \in \{0,1\}^{\ell+1}$, and sends \underline{$(d_i)_{i=1}^{N}$} to the verifier. 
  \item[$\mathbf{V}$ (output):] If $c=0$ (``test''), $\mathbf{V}$ either accepts or rejects the interaction but outputs nothing. In particular, $\mathbf{V}$ checks that $f_{\pk_i}(b_i,x_i) = y_i$. If any of the checks fail, reject. 
  
  If $c=1$ (``measure''), proceed as follows. Let $R \subseteq [N]$ be the set of indices where $h_i = 1$, corresponding to Hadamard basis measurements, and let $S = [N] \setminus R$ be the set of indices corresponding to standard basis measurements.
  
  For each $i \in R$, $\mathbf{V}$ computes the two inverses $(0,x_{i,0})$ and $(1,x_{i,1})$ of $y_i$ (which are guaranteed to exist since $f_{\pk,i} \in \mathsf{Cf}$) using $\sk_i$. $\mathbf{V}$ sets $u_i := d \cdot (1,x_{i,0} \oplus x_{i,1})$ as the $i^{th}$ measurement outcome. For every $i \in S$, $\mathbf{V}$ ignores $d_i$, and sets $v_i$ to be the first bit of $f_{\pk_i}^{-1}(y_i)$, computed using the trapdoor $\sk_i$ (this is well-defined since $f_{\pk,i} \in \mathsf{Inj}$). Finally $\mathbf{V}$ outputs the $N$-bit string \underline{$(u,v) \in \{0,1\}^R \times \{0,1\}^S$}.
  \end{description}

Mahadev~\cite{FOCS:Mahadev18a} proves that if a malicious prover $\mathbf{P}^*$ passes the test round with probability $1$, then there exists an $N$-qubit quantum state $\brho^*$ --- \emph{independent} of the verifier's measurement basis $h$ --- such that the result of measuring $\brho^*$ according to $h$ is computationally indistinguishable from the verifier's $N$-bit output distribution in the measurement round.\footnote{This can be extended to provers that pass the test round with probability $1-\varepsilon$ by the gentle measurement lemma. In particular, an efficient distinguisher can only distinguish the verifier's output distribution from the result of measuring some $\brho^*$ with advantage $\poly(\varepsilon)$.} While her definition requires that such a $\brho^*$ \emph{exists}, Vidick and Zhang~\cite{EC:VidZha21} showed that Mahadev's proof steps implicitly define an extractor that efficiently produces $\brho^*$ using black-box access to $\mathbf{P}^*$.

\subsection{Defining a (Succinct) Measurement Protocol}

Our first (straightforward but helpful) step is to give an explicit definition of a \emph{commit-and-measure protocol} that abstracts the completeness and soundness properties of Mahadev's measurement protocol as established in \cite{FOCS:Mahadev18a,EC:VidZha21}. Roughly speaking, a commit-and-measure protocol is sound if, for any malicious prover $\mathbf{P}^*$ that passes the test round with probability $1$ and any basis choice $h$, there exists an efficient extractor that (without knowledge of $h$) interacts with prover and outputs an extracted state $\btau$ such that the following are indistinguishable:
\begin{itemize}
    \item the distribution of verifier outputs obtained in the measurement round from interacting with $\mathbf{P}^*$ using basis choice $h$, and
    \item the distribution of measurement outcomes obtained from measuring $\btau$ according to $h$.
\end{itemize}
This abstraction will be particularly helpful for reasoning about our eventual \emph{succinct} measurement protocols, which will necessitate modifying Mahadev's original protocol.

\paragraph{Can a Measurement Protocol be Succinct?}

Given the definition of a measurement protocol, an immediate concern arises with respect to obtaining succinct arguments: the verifier's \emph{input} to the measurement protocol -- the basis vector $h$ -- is inherently non-succinct. Since the number of qubits $N$ grows with the runtime of the BQP computation when used to obtain quantum verification \cite{FHM18}, this poses an immediate problem. 

Our solution to this problem is to only consider basis vectors $h$ that are \emph{succinct}; our formalization is that $h$ must be the truth table of an efficiently computable function $f: [\log N] \rightarrow \{0,1\}$. For any such $h$, we can represent the verifier's input as a circuit $C$ that computes $h$, removing the above obstacle. 

However, in order for there to be any hope of this idea working, it must be the case that measurement protocols for bases with succinct representations are still useful for constructing delegation for BQP. Fortunately, it has been shown \cite{TCC:ACGH20} that classically verifiable (non-succinct) arguments for BQP can be constructed by invoking Mahadev's measurement protocol (and, by inspection of the proof, any measurement protocol satisfying our definition) on a \emph{uniformly random basis string} $h \gets \{0,1\}^N$. Then, by computational indistinguishability, it is also possible to use a \emph{pseudorandom} string $h$ that has a succinct representation, i.e., $h = (\PRF_s(1), \hdots, \PRF_s(N))$ for some (post-quantum) pseudorandom function $\PRF$. 

Thus, we focus for the moment on constructing a succinct measurement protocol for $h$ with succinct representation, and return to the full delegation problem later. 

\subsection{Constructing a Verifier-Succinct Measurement Protocol}
\label{subsec:succ-measure-tech-overview}

Inspecting the description of the \cite{FOCS:Mahadev18a} protocol, there are three distinct reasons that the protocol is not succinct:

\begin{enumerate}
    \item The verifier's first message, which consists of $N$ TCF public keys, is non-succinct.
    \item The prover's two messages, consisting of the commitments $y_i$ and openings $z_i$ respectively, are non-succinct.
    \item The verifier's decision predicate, as it is a function of these commitments and openings, requires $\poly(N)$ time to evaluate.
\end{enumerate}

The latter two issues turn out to be not too difficult to resolve (although there is an important subtlety that we discuss later); for now, we focus on resolving (1), which is our main technical contribution. Concretely, we want to construct a measurement protocol for succinct bases $h$ where the verifier's first message is succinct.

\paragraph{Idea: Compress the Verifier's message with iO.} Given the problem formulation, a natural idea presents itself: instead of having $V$ send over $N$ i.i.d. public keys $\pk_i$, perhaps $V$ can send a succinct program $\PK$ that contains the description of $N$ public keys $\pk_i$ that are in some sense ``pseudoindependent!'' Using the machinery of obfuscation and the ``punctured programs'' technique \cite{STOC:SahWat14}, it is straightforward to write down a candidate program for this task: simply obfuscate the following code.

\begin{figure}[h!]
\centering
\begin{tabular}{|p{14cm}|}
\hline\ \\
{\bf Input}: index $i \leq N$\\
{\bf Hardwired Values}: Puncturable PRF seed $s$. Circuit $C$. \\
\begin{itemize}

\item Compute $\mode=C(i)$ and $r = \PRF_{s}(i)$.

\item Compute $(\pk_i, \sk_i) \gets \Gen(1^\secp, \mode; r)$. 

\item Output $\pk_i$. 
	
\end{itemize}
\ \\
\hline
\end{tabular}
\end{figure}
Here, $C$ is an efficient circuit with truth table $h$, and $\Gen(1^\secp, \mode)$ indicates sampling either from $\mathsf{Inj}$ or $\mathsf{Cf}$ depending on whether $h_i = C(i) = 0$ or $h_i = C(i) = 1$. 

Letting $\PK$ denote an obfuscation of the above program, $V$ could send $\PK$ to $P$ and allow the prover to compute each $\pk_i = \PK(i)$ on its own, and the protocol could essentially proceed as before, except that the verifier will have to expand its PRF seed $s$ into $(\sk_1, \hdots, \sk_N)$ in order to compute its final output. 

\paragraph{Problem: Proving Soundness.}

While it is not hard to describe this plausible modification to the \cite{FOCS:Mahadev18a} protocol that compresses the verifier's message, it is very unclear how to argue that the modified protocol is sound. The obfuscation literature has no shortage of proof techniques developed over the last 10 years, but since we have made a ``non-black-box'' modification of the \cite{FOCS:Mahadev18a} protocol, a deep understanding of the \cite{FOCS:Mahadev18a} proof of soundness is required in order to understand to what extent these techniques are compatible with the application at hand. 

We believe it \emph{should} be possible to incorporate punctured programming techniques into Mahadev's proof of soundness in \cite{FOCS:Mahadev18a} and conclude the desired soundness property of the new protocol. However, doing so would result in an extremely complex proof that would require the reader to verify the entirety of the \cite{FOCS:Mahadev18a} (already very complicated) original security proof with our modifications in mind.

\subsection{Proof of Soundness}\label{subsec:soundness-tech-overview}

Given the complicated nature of the \cite{FOCS:Mahadev18a} proof of soundness, we instead give a \emph{simpler} and \emph{more modular} proof of soundness for the \cite{FOCS:Mahadev18a} measurement protocol. Moreover, we give this proof for a generic variant of the \cite{FOCS:Mahadev18a} protocol where the prover is given an arbitrary representation $\PK$ of $N$ TCF public keys and show that precisely two properties of this representation $\PK$ are required in order for the proof to go through:

\begin{itemize}
    \item An appropriate generalization of the ``dual-mode'' property of individual TCFs must hold for $\PK$: for any two circuits $C_1, C_2$, it should be that $\PK_1$ generated from basis $C_1$ is computationally indistinguishable from $\PK_2$ generated from basis $C_2$. In fact, a stronger variant of this indistinguishability must hold: it should be the case that $\PK_1 \approx_c \PK_2$ even if the distinguisher is given all secret keys $\sk_j$ such that $C_1(j) = C_2(j)$.  
    \item For every $i$, the adaptive hardcore bit property of $f_{\pk_i}$ should hold \emph{even given} $\sk_j$ for all $j\neq i$.
\end{itemize}

Since these two properties are (essentially) all that is required for our proof to go through, in order to obtain a verifier-succinct protocol, it suffices to show that the obfuscated program $\PK$ above satisfies these two properties, which follows from standard techniques. 

Thus, we proceed by describing our new soundness proof for the \cite{FOCS:Mahadev18a} measurement protocol, which transparently generalizes to the verifier-succinct setting.

\paragraph{The ``Operational Qubits'' Approach.}

Let $P^*$ denote a prover that passes the test round (i.e., makes the verifier accept on the $0$ challenge) with probability $1$. Our goal is to show that the prover in some sense ``has an $N$-qubit state'' such that measuring this state in the $h$-bases produces the same (or an indistinguishable) distribution as the verifier's protocol output, which we will denote $D_{P^*, \mathrm{Out}}$. This $N$-qubit state should be efficiently computable from the prover's internal state $\ket{\psi}$; specifically, we use $\ket{\psi}$ to denote the prover's state after its first message $y$ has been sent.

In order to show this, taking inspiration from \cite{Vid20-course},\footnote{\cite{Vid20-course} gives a soundness proof for a variant of the \cite{FOCS:Mahadev18a} protocol, but in a qualitatively weaker setting. \cite{Vid20-course} only proves indistinguishability of $N$-qubit measurements that are either \emph{all} in the standard basis or \emph{all} in the Hadamard basis, and only proves indistinguishability with respect to \emph{linear} tests of the distribution (that is, \cite{Vid20-course} proves small-bias rather than full indistinguishability). Both of these relaxations are unacceptable in our setting, and achieving the latter specifically requires a different proof strategy.} we will proceed in two steps:

\begin{enumerate}
    \item Identify $N$ ``operational qubits'' within $\ket{\psi}$. That is, we will identify a set of $2N$ observables $Z_1, \hdots, Z_N, X_1, \hdots, \allowbreak X_N$ (analogous to the ``Pauli observables'' $\sigma_{z,1}, \hdots, \sigma_{z, N}, \sigma_{x, 1}, \hdots \sigma_{x, N}$) such that measuring $\ket{\psi}$ with these observables gives the outcome distribution $D_{P^*, \mathrm{Out}}$.
    
    \hspace{.7cm} Provided that these $2N$ observables roughly ``behave like'' Pauli observables with respect to $\ket{\psi}$ (e.g. satisfy the X/Z uncertainty principle), one could then hope to:
    \item  Extract a related state $\ket{\psi'}$ such that measuring $\ket{\psi'}$ in the \emph{actual} standard/Hadamard bases  matches the ``pseudo-Pauli'' $\{Z_j\}$, $\{X_i\}$,  measurements of $\ket{\psi}$ (and therefore $D_{P^*, \mathrm{Out}}$). 
\end{enumerate} 

\paragraph{Relating the Verifier's Output to Measuring $\ket{\psi}$.} Our current goal is to achieve Step (1) above. Let $\ket{\psi}$ denote $P^*$'s post-commitment state and let $U$ denote the unitary such that $P^*$'s opening is a measurement of $U \ket{\psi}$ in the Hadamard basis. 

Now, let us consider the verifier's output distribution. The $i$th bit of the verifier's output when $h_i = 1$ is defined to be $d\cdot (x_{0,i}\oplus x_{1,i})$ (where $d$ is the opening sent by the prover) of $U\ket{\psi}$ in the Hadamard basis. For each such $i$, we can define an observable $X_i$ characterizing this measurement, that \emph{roughly} takes the form
    \[ X_i \approx  U^\dagger (H_{\RegZ_i} \tensor \Id) \left(\sum_{d} (-1)^{ d \cdot (1, x_{0,i} \oplus x_{1,i}) } \ketbra{d}_{\RegZ_i} \tensor \Id_{\RegI, \{\RegZ_j\}_{j\neq i}} \right) (H_{\RegZ_i}\tensor \Id ) U.\]
    Here we have slightly simplified the expression for $X_i$ for the sake of presentation; the correct definition of $X_i$ (see \cref{subsec:XZ-definition}) must account for the case where $d$ is rejected by the verifier. To reiterate, the observable $X_i$ is a syntactic interpretation of the verifier's output $m_i$ as a function of $\ket{\psi}$. 
    
    On the other hand, when $h_i = 0$, the verifier's output $m_i$ is not \emph{a priori} a measurement of $\ket{\psi}$; indeed, the verifier ignores the prover's second message and just inverts $y_i$. However, under the assumption that the prover $P^*$ passes the test round with probability $1-\negl(\secp)$, making use of the fact that $f_{\pk_i}$ is injective, this $y_i$-inverse must be equal to what the prover \emph{would} have sent in the test round. This defines another observable on $\ket{\psi}$ that we call $Z_i$:
    \[ Z_i = \sum_{b, x} (-1)^b \ketbra{b, x}_{\RegZ_i} \tensor \Id_{\RegI, \{\RegZ_j\}_{j\neq i}}. \]
    
    Finally, note that the operator $Z_i$ syntactically makes sense even when $h_i = 1$. However, $X_i$ cannot even be \emph{defined} when $f_{\pk_i}$ is injective, corresponding to $h_i = 0$, since $X_i$ explicitly requires \emph{two} inverses of $y_i$. Therefore, from now on, we sample all $(\pk_i, \sk_i) \gets \mathsf{Cf}$ (forcing all TCFs to be 2-to-1). 
    
    This brings us to the punchline of this step: by invoking a computational assumption (the indistinguishability of $\mathsf{Cf}$ and $\mathsf{Inj}$), we can define observables $(X_i, Z_i)$ for all $i \in [N]$ such that for \emph{every} $i$ and \emph{every} basis choice $h$, the distribution resulting from measuring $\ket{\psi}$ with $X_i$ (resp. $Z_i$) matches the $i$th bit of the verifier's output distribution.
    
    With a little more work, one can actually show that the verifier's \emph{entire} output distribution in the $h$-basis is computationally indistinguishable from the following distribution $D_{P^*,\twotoone}$:
\begin{itemize}
    \item Sample keys $(\pk_i,\sk_i) \gets \mathsf{Cf}$. Run $P^*$ to obtain $y, \ket{\psi}$.
    \item For each $i$ such that $h_i = 0$, measure the first bit of the prover's $i$th response register in the standard basis to obtain (and output) a bit $b_i$.
    \item Measure $U\ket{\psi}$ in the Hadamard basis, obtaining strings $(d_1, \hdots, d_N)$.
    \item For each $i$ such that $h_i = 1$, compute (and output) $d_i \cdot (1,x_{0,i} \oplus x_{1,i})$. 
\end{itemize}
    
\paragraph{Aside: Why are these $Z_j$ and $X_i$ helpful?} As alluded to earlier, this approach is inspired by \emph{operational} definitions of ``having an $N$-qubit state,'' which consists of a state $\ket{\psi}$ and $2N$ ``pseudo-Pauli'' observables $Z_1, \hdots, Z_N, X_1, \hdots X_N$ that behave ``like Pauli observables'' on $\ket{\psi}$. For example, it is possible to prove that many of the ``Pauli group relations'' hold \emph{approximately} on these $X_i, Z_j$ with respect to $\ket{\psi}$, meaning that (for example)
\[ \bra{\psi} Z_i X_i Z_i + X_i \ket{\psi} = \negl(\secp)
\]
and 
\[ \bra{\psi} Z_j X_i Z_j - X_i \ket{\psi} = \negl(\secp)
\]
for $i\neq j$. In fact, these relations turn out to \emph{encode} the two basic properties of the TCF $f_{\pk_i}$: the adaptive hardcore bit property (encoded in the first relation) and that $f_{\pk_i}$ is indistinguishable from injective\footnote{Technically, the property encoded is the \emph{collapsing} of $f_{\pk_i}$, which is implied by (but not equivalent to) being indistinguishable from injective.} (encoded in the second relation)! We will not directly prove the relations here, but they are implicit in our full security proof and are the motivation for this proof strategy.

\paragraph{The Extracted State.}
Given these protocol observables $Z_1, \hdots, Z_N, X_1, \hdots, X_N$, it remains to implement Step (2) of our overall proof strategy: extracting a state $\ket{\psi'}$ whose standard/Hadamard measurement outcomes match $D_{P^*, \mathrm{Out}}$. At a high level, this is achieved by ``teleporting'' the state $\ket{\psi}$ onto a fresh $N$-qubit register in a way that \emph{transforms} the ``pseudo-Paulis" $\{X_i\}, \{Z_j\}$ into \emph{real} Pauli observables $\{\sigma_{x,i}\}, \{\sigma_{z,j}\}$. 

Fix a choice of $\{X_i,Z_i\},\ket{\psi} \gets \mathsf{Samp}$. For ease of notation, write $\RegH = \RegZ \otimes \RegI \otimes \RegU$ so that $\ket{\psi} \in \RegH$. We would like an efficient extraction procedure that takes as input $\ket{\psi} \in \RegH$ and generates an $N$-qubit state $\btau$ such that, roughly speaking, measuring $\ket{\psi}$ with $X/Z$ and measuring $\btau$ with $\sigma_X/\sigma_Z$ produce indistinguishable outcomes.

\paragraph{Intuition for the Extractor.} Before we describe our extractor, we first provide some underlying intuition. For an arbitrary $N$-qubit Hilbert space, let $\sigma_{x,i}$/$\sigma_{z,i}$ denote the  Pauli $\sigma_x$/$\sigma_z$ observable acting on the $i$th qubit. For each $r,s \in \{0,1\}^N$, define the $N$-qubit Pauli ``parity'' observables
\[\sigma_x(r) \coloneqq \prod_{i: r_i = 1} \sigma_{x,i} \;,\; \sigma_z(s) \coloneqq \prod_{i: r_i = 1} \sigma_{z,i}.\]

Suppose for a moment that $\ket{\psi} \in \RegH$ is \emph{already} an $N$-qubit state (i.e., $\RegH$ is an $N$-qubit Hilbert space) and moreover, that each $X_i$/$Z_i$ observable is simply the corresponding Pauli observable $\sigma_{x,i}$/$\sigma_{z,i}$. While these assumptions technically trivialize the task (the state already has the form we want from the extracted state), it will be instructive to \textbf{write down an extractor that ``teleports'' this state into another $N$-qubit external register. }

We can do this by initializing two $N$-qubit registers $\RegA_1 \otimes \RegA_2$ to $\ket{\phi^+}^{\otimes N}$ where $\ket{\phi^+}$ is the EPR state $(\ket{00} + \ket{11})/\sqrt{2}$ (the $i$th EPR pair lives on the $i$th qubit of $\RegA_1$ and $\RegA_2$). Now consider the following steps, which are inspired by the ($N$-qubit) quantum teleportation protocol
\begin{enumerate}
    \item Initialize a $2N$-qubit ancilla $\RegW$ to $\ket{0^{2N}}$, and apply $H^{\otimes 2N}$ to obtain the uniform superposition.
    \item Apply a ``controlled-Pauli'' unitary, which does the following for all $r,s \in \{0,1\}^N$ and all $\ket{\phi} \in \RegH \otimes \RegA_1$:
    \[\ket{r,s}_{\RegW} \ket{\phi}_{\RegH,\RegA_1} \rightarrow \ket{r,s}_{\RegW} (\sigma_x(r)\sigma_z(s)_{\RegH} \otimes \sigma_x(r)\sigma_z(s)_{\RegA_1})\ket{\phi}_{\RegH,\RegA_1}\] 
    \item Apply the unitary that XORs onto $\RegW$ the outcome of performing $N$ Bell-basis measurements\footnote{The Bell basis consists of the $4$ states $(\sigma_x^a \sigma_z^b \otimes \Id)\ket{\phi^+}$ for $a,b \in \{0,1\}$ on $2$ qubits.} on $\RegA_1 \otimes \RegA_2$ onto $\RegW$, i.e., for all $u,v,r,s \in \{0,1\}^N$: \[\ket{u,v}_{\RegW}(\sigma_x(r)\sigma_z(s) \otimes \Id)_{\RegA_1,\RegA_2}\ket{\phi^+}^{\otimes N}_{\RegA_1,\RegA_2} \rightarrow \ket{u \oplus r ,v \oplus s}_{\RegW}(\sigma_x(r)\sigma_z(s) \otimes \Id)_{\RegA_1,\RegA_2}\ket{\phi^+}^{\otimes N}_{\RegA_1,\RegA_2}.\]
    Finally, discard $\RegW$.
\end{enumerate}
One can show that the resulting state is
\begin{equation}
\label{eqn:trivial-extracted-state}
\frac{1}{2^N} \sum_{r,s \in \{0,1\}^N} (\sigma_x(r)\sigma_z(s) \otimes \sigma_x(r) \sigma_z(s) \otimes \Id)\ket{\psi}_{\RegH} \ket{\phi^+}_{\RegA_1,\RegA_2}=\ket{\phi^+}_{\RegH,\RegA_1} \ket{\psi}_{\RegA_2},
\end{equation}
where $\ket{\psi}$ is now ``teleported'' into the $\RegA_2$ register. 

\paragraph{The Full Extractor.} To generalize this idea to the setting where $\ket{\psi} \in \RegH$ is an arbitrary quantum state and $\{X_i,Z_i\}_i$ are an arbitrary collection of $2N$ observables, \emph{we simply replace each $\sigma_x(r)$ and $\sigma_z(s)$ acting on $\RegH$ above with the corresponding parity observables $X(r)$, $Z(s)$}, defined analogously (for $r,s \in \{0,1\}^N$ as
\[ Z(s) = \prod_{i=1}^N Z_i^{s_i}
\hspace{.1in} \mbox{and}
\hspace{.1in} X(r) = \prod_{i=1}^N X_i^{r_i}. 
\]
The rough intuition is that as long as the $\{X_i\}$ and $\{Z_i\}$ observables ``behave like'' Pauli observables with respect to $\ket{\psi}$, the resulting procedure will ``teleport'' $\ket{\psi}$ into the $N$-qubit register $\RegA_2$.

\paragraph{Relating Extracted State Measurements to Verifier Outputs.}

With the extracted state defined to be the state on $\RegA_2$ after performing the ``generalized teleportation'' described above, it remains to prove that the distribution $D_{P^*,\mathrm{Ext}}$ resulting from measuring the extracted state on $\RegA_2$ in the $h$-bases is indistinguishable from $D_{P^*,\twotoone}$.

One can show (by a calculation) that $D_{P^*,\mathrm{Ext}}$ is the following distribution (differences from $D_{P^*,\twotoone}$ in \textcolor{red}{red})
\begin{enumerate}
    \item Sample keys $(\pk_i,\sk_i) \gets \mathsf{Cf}$. Run $P^*$ to obtain $y, \ket{\psi}$.
    \item For each $i$ such that $h_i = 0$, measure the first bit of the prover's $i$th response register in the standard basis to obtain (and output) a bit $b_i$.
    \item \textcolor{red}{For each $i$ such that $h_i = 1$, flip a random bit $w_i$ and apply the unitary $Z_i^{w_i}$.}
    \item Measure $U\ket{\psi}$ in the Hadamard basis, obtaining strings $(d_1, \hdots, d_N)$.
    \item For each $i$ such that $h_i = 1$, compute (and output) $d_i \cdot (1,x_{0,i} \oplus x_{1,i}) \; \textcolor{red}{\oplus \; w_i}$.
\end{enumerate}

We prove indistinguishability between the $N$-bit distributions $D_{P^*,\mathrm{Ext}}$ and $D_{P^*,\twotoone}$ by considering $N$ hybrid distributions, where the difference between Hybrid $j-1$ and Hybrid $j$ is:
\begin{itemize}
    \item an additional application of the unitary $Z_j$ in Item 3, and
    \item an additional XOR of $e_j$ (the $j$th standard basis vector) in Item 5.
\end{itemize}

To conclude the soundness proof, we show that Hybrid $j-1$ and Hybrid $j$ in the following three steps.

\begin{itemize}
    \item First, we prove that the marginal distributions of Hybrid $(j-1)$ and Hybrid $j$ on $N \setminus \{j\}$ are indistinguishable due to the collapsing property of $f_{\pk_{j}}$. Intuitively this holds because the marginal distributions on $N \setminus \{j\}$ only differ by the application of $Z_j$, which is undetectable by collapsing.
    \item By invoking an elementary lemma about $N$-bit indistinguishability, the task reduces to proving a $1$-bit indistinguishability of the $j$th bit of Hybrid $(j-1)$ and Hybrid $j$, conditioned on an efficiently computable property of the marginal distributions on $N \setminus \{j\}$.
    \item Finally, we show that the indistinguishability of the $j$th bit holds due to the adaptive hardcore bit property of $f_{\pk_{j}}$. At a very high level, the above $j$th bit property involves a measurement of $X_j$, and the two hybrids differ in whether a random $Z_j^b$ is applied before $X_j$ is measured; in words, this exactly captures the adaptive hardcore bit security game.
    
    We refer the reader to \cref{subsec:measurement-indistinguishability} for a full proof of indistinguishability.

\end{itemize}

\subsection{From a Verifier-Succinct Measurement Protocol to Succinct Arguments for BQP}

Using \cref{subsec:succ-measure-tech-overview,subsec:soundness-tech-overview}, we have constructed a \emph{verifier-succinct} measurement protocol, for succinctly represented basis strings, with a single bit verifier challenge. What remains is to convert this into a (fully) succinct argument system for BQP (or QMA). This is accomplished via the following transformations:

\begin{itemize}
    \item Converting a measurement protocol into a quantum verification protocol. As described earlier, this is achieved by combining the \cite{FHM18} protocol for BQP verification with a limited quantum verifier (as modified by \cite{TCC:ACGH20}) with our measurement protocol, using a PRF to generate a pseudorandom basis choice instead of a uniformly random basis choice for the \cite{FHM18,TCC:ACGH20} verifier. This results in a verifier-succinct argument system for BQP/QMA with constant soundness error.
    \item Parallel repetition to reduce the soundness error. This follows from the ``computational orthogonal projectors'' property of the 1-bit challenge protocol and follows from \cite{TCC:ACGH20} (we give a somewhat more abstract formulation of their idea in \cref{appendix:cop}). This results in a verifier-succinct argument system for BQP/QMA with negligible soundness error.
    \item Converting a verifier-succinct argument system into a fully succinct argument system. We elaborate on this last transformation below, as a few difficulties come up in this step.
\end{itemize}

Assume that we are given a (for simplicity, 4-message) verifier-succinct argument system for BQP/QMA. Let $m_1, m_2, m_3, m_4$ denote the four messages in such an argument system. In order to obtain a fully succinct argument system, we must reduce (1) the prover communication complexity $|m_2| + |m_4|$, and (2) the runtime of the verifier's decision predicate.

The first idea that comes to mind is to ask the prover to send short (e.g. Merkle tree) commitments $\sigma_2$ and $\sigma_4$ of $m_2$ and $m_4$, respectively, instead of sending $m_2$ and $m_4$ directly. At the end of the interaction, the prover and verifier could then engage in a succinct interactive argument (of knowledge) for a (classical) NP statement that ``the verifier would have accepted the committed messages underlying $\sigma_2$ and $\sigma_4$''. One could potentially employ Kilian's succinct interactive argument of knowledge for NP which was recently shown to be post-quantum secure under the post-quantum LWE assumption~\cite{FOCS:CMSZ21}.

There are a few issues with this naive idea. First of all, the verifier's decision predicate is \emph{private} (it depends on the secret key $\SK$ in the measurement protocol and the PRF seed for its basis), so the NP statement above is not well-formed. One reasonable solution to this issue is to simply have the verifier send this secret information $\mathsf{st}$ after the verifier-succinct protocol emulation has occurred and before the NP-succinct argument has started. For certain applications (e.g. obtaining a non-interactive protocol in the QROM) we would like to have a \emph{public-coin} protocol; this can be achieved by using fully homomorphic encryption to encrypt this secret information in the \emph{first} round rather than sending it in the clear in a later round. For this overview, we focus on the private-coin variant of the protocol.

Now, we can indeed write down the appropriate NP relation\footnote{Note that the verifier also takes as input the QMA instance, but we suppress it here for clarity.} 
\begin{align*} \mathcal{R}_V = & \{ ((h,m_1,\sigma_2,m_3,\sigma_4,\mathsf{st}), (m_2,m_4)): \sigma_2 = h(m_2) \mbox{ and } \\ 
& \sigma_4 = h(m_4) \mbox{ and } V(\mathsf{st},m_1,m_2,c,m_4) = \mathsf{accept} \}
\end{align*}
and execute the aforementioned strategy.
However, this construction turns out not to work. Specifically, it does not seem possible to \emph{convert} a cheating prover $P^*$ in the above fully succinct protocol into a cheating prover $P^{**}$ for the verifier-succinct protocol; for example, $P^{**}$ needs to be able to produce a message $m_2$ given only $m_1$ from the verifier; meanwhile, the message $m_1$ can only be extracted from $P^*$ by repeatedly rewinding $P^*$'s \emph{last} message algorithm, which requires the verifier's secret information $\mathsf{st}$ as input! This does not correspond to a valid $P^{**}$, who does not have access to $\mathsf{st}$ when computing $m_2$.

Our refined compiler is to execute several arguments of knowledge: one right after the prover sends $\sigma_2$, proving knowledge of $m_2$; another one right after she sends $\sigma_4$, proving knowledge of $m_4$ (both before receiving the secret state $\mathsf{st}$ from the verifier); and a third one for the relation $\mathcal{R}_V$ described above. The first two arguments of knowledge are for the relation 
$$ \mathcal{R}_H = \{(h,\sigma),m): h(m) = \sigma\}$$
This allows for {\em immediate extraction} of $m_2$ and $m_3$ and appears to clear the way for a reduction between the verifier-succinct and fully succinct protocol soundness properties. 

However, there is one remaining problem: the argument-of-knowledge property of Kilian's protocol proved by \cite{FOCS:CMSZ21} is \emph{insufficiently composable} to be used in our compiler. They demonstrate an extractor for Kilian's protocol that takes any quantum cheating prover that convinces the verifier and extracts a witness from them. However, their post-quantum extractor might significantly disturb the prover's state, meaning that once we extract $m_2$ above, we may not be able to continue the prover execution in our reduction. 

Fortunately, a recent work \cite{LMS21} shows that a slight variant of Kilian's protocol is a succinct argument of knowledge for NP satisfying a composable extraction property called ``state-preservation.'' This security property is exactly what is required for our compiler to extract a valid cheating prover strategy $P^{**}$ for the verifier-succinct argument given a cheating prover $P^*$ for the compiled protocol. A full discussion of this is given in \cref{sec:final}.

This completes our construction of a succinct argument system for BQP (and QMA). We discuss additional results (2-message protocols, zero knowledge, batch arguments) in \cref{sec:additional}.

\ifsubmission
\else
\newcommand{\Good}{\mathsf{Good}}

\section{Preliminaries}
\label{sec:preliminaries}

\subsection{Quantum Information}

Let $\RegH$ be a finite-dimensional Hilbert space. A pure state is a unit vector $\ket{\psi} \in \RegH$.  Let $\mathrm{D}(\RegH)$ denote the set of all positive semidefinite operators on $\RegH$ with trace~$1$. A mixed state is an operator $\bm{\rho} \in \mathrm{D}(\RegH)$, and is often called a \emph{density matrix}. We sometimes divide $\RegH$ into named \emph{registers} written in uppercase calligraphic font, e.g., $\RegH = \RegA \otimes \RegB \otimes \RegC$. 

For a density matrix $\bm{\rho} \in\mathrm{D}(\RegH)$, where $\RegH\simeq (\mathbb{C}^2)^{\otimes \ell}$, we sometimes use the shortcut $M(h, \bm{\rho})$ to denote the distribution resulting from measuring each qubit of $\bm{\rho}$ (where the qubits are specified by the isomorphism $\RegH\simeq (\mathbb{C}^2)^{\otimes \ell}$) in the basis determined by $h\in\{0,1\}^\ell$. By convention, $h_i = 0$ corresponds to measuring the $i$-th register in the standard basis $\{\ket{0},\ket{1}\}$ and $h_i = 1$ corresponds to measuring the $i$-th register in the Hadamard basis $\{\ket{+},\ket{-}\}$.

An \emph{observable} is represented by a Hermitian operator $O$ on $\RegH$. In particular, any observable $O$ can be written in the form $\sum_i \lambda_i \Pi_i$ where $\{\lambda_i\}$ are real numbers and $\sum_i \Pi_i = \Id$. The measurement corresponding to an observable $O$ is the projective measurement $\{\Pi_i\}$ with corresponding outcomes $\{\lambda_i\}$. A \emph{binary observable} satisfies the additional requirement that $O^2 = \Id$. Notice that for any binary observable, $O$ is a unitary matrix with eigenvalues in $\{1,-1\}$. In this case we sometimes treat the outcomes as bits through the usual correspondence $1\to 0$, $-1\to 1$. 

Given a binary observable $O$, we define its corresponding \emph{projection operators} $O^+ = \frac 1 2 (\Id + O)$ and $O^- = \frac 1 2 (\Id - O)$. $O^+$ and $O^-$ correspond to projecting onto the $+1$ and $-1$ eigenspaces of $O$, respectively, and thus form a binary projective measurement.

\paragraph{The Class $\QMA$.} 
A language $\cL = (\cL_{\text{yes}}, \cL_{\text{no}})$ is in $\QMA$ if and only if there is a uniformly generated family of polynomial-size quantum circuits $\mathcal{V} = \{V_\lambda\}_{\lambda \in \mathbb{N}}$ such that for every $\lambda$, $V_\lambda$ takes as input a string $x \in \{0,1\}^\lambda$ and a quantum state $\ket{\phi}$ on $p(\lambda)$ qubits and returns a single bit and moreover the following conditions hold.
\begin{itemize}
    \item For all $x \in \cL_{\text{yes}}$ of length $\lambda$, there exists a quantum state $\ket{\psi}$ on at most $p(\lambda)$ qubits such that the probability that $V_\lambda$ accepts $(x, \ket{\phi})$ is at least $2/3$. We denote the (possibly infinite) set of quantum states (which we will also refer to as quantum witnesses) that make $V_\lambda$ accept $x$ by $\cR(x)$\thomas{this might not be used if we don't include any ``proof of knowledge'' aspects}.
    \item For all $x \in \cL_{\text{no}}$ of length $\lambda$, and all quantum states $\ket{\psi}$ on at most $p(\lambda)$ qubits, it holds that $V_\lambda$ accepts on input $(x,\ket{\psi})$ with probability at most $1/3$.
\end{itemize}

\subsection{Black-Box Access to Quantum Algorithms}
\label{subsec:black-box}

Let $A$ be a polynomial-time quantum algorithm with internal state $\brho \in \mathrm{D}(\RegI)$ that takes a classical input $r$ and produces a classical output $z$. Without loss of generality, the behavior of $A$ can be described as follows:
\begin{enumerate}
    \item Apply an efficient \emph{classical} algorithm to $r$ to generate the description of a unitary $U(r)$.
    \item Initialize registers $\RegZ \otimes \RegI$ to $\ketbra{0}_{\RegZ} \otimes \brho_{\RegI}$.
    \item Apply $U(r)$ to $\RegZ \otimes \RegI$, measure $\RegZ$ in the computational basis, and return the outcome $z$.
\end{enumerate}

A quantum oracle algorithm $S^A$ with \emph{black-box access} to $(A,\brho)$ does not have direct access to the adversary's \emph{internal registers} $\RegI$, and can only operate on the state $\brho \in \mathrm{D}(\RegI)$ by applying $U(r)$ or $U(r)^\dagger$ for any $r$. In more detail, black-box access to $(A,\brho)$ means the following:  
\begin{itemize}
    \item The registers $\RegZ \otimes \RegI$ are initialized to $\ketbra{0}_{\RegZ} \otimes \brho_{\RegI}$.
    \item Once the $\RegZ \otimes \RegI$ registers are initialized, the algorithm is permitted to perform arbitrary operations on the $\RegZ$ register, but can only act on the $\RegI$ registers by applying $U(r)$ or $U(r)^\dagger$ for any $r$. We explicitly permit the $U(r)$ and $U(r)^\dagger$ gates to be controlled on any external registers (i.e., any registers other than the registers $\RegZ \otimes \RegI$ to which $U(r)$ is applied). 
\end{itemize}

We note that this definition is consistent with the notions of interactive quantum machines and oracle access to an interactive quantum machine used in e.g.~\cite{unruh2012quantum} and other works on post-quantum zero-knowledge. 

\paragraph{The Binary Input Case.} Following~\cite{FOCS:Mahadev18a}, in the special case where $r \in \{0,1\}$, it will be convenient to re-define the internal state to be $\brho \coloneqq U(0)(\ketbra{0}_{\RegZ} \otimes \brho'_{\RegI}) U(0)^\dagger$ (where $\brho'_{\RegI} \in \mathrm{D}(\RegI)$ denotes the ``original'' internal state), so that the behavior of $A$ on $r = 0$ is to simply measure $\RegZ$ in the computational basis, and on $r = 1$ it applies the unitary $U \coloneqq U(1)U(0)^\dagger$ to its state and then measures the $\RegZ$ register. Notice that in this case, the internal state is technically on $\RegZ \otimes \RegI$ instead of just $\RegI$. Thus, black-box access to a quantum algorithm with binary input is formalized as follows:
\begin{itemize}
    \item The registers $\RegZ \otimes \RegI$ are initialized to $\brho \coloneqq U(0)(\ketbra{0}_{\RegZ} \otimes \brho'_{\RegI}) U(0)^\dagger$.
    \item Once the $\RegZ \otimes \RegI$ registers are initialized, the algorithm is permitted to perform arbitrary operations on the $\RegZ$ register, but can only act on the $\RegI$ registers by applying (possibly controlled) $U$ or $U^\dagger$ gates.
\end{itemize}
In this special case, an algorithm with black-box access to $A$ is denoted $S^{U,\brho}$. 

We remark that these definitions are tailored to the two-message challenge-response setting, whereas the protocols we consider in this paper have more rounds of interaction. However, our analysis will typically focus on a single back-and-forth round of interaction (e.g., the last two messages of the~\cite{FOCS:Mahadev18a} protocol), so $\brho$ will be the intermediate state of the interactive algorithm right before the next challenge is sent.\footnote{In the multi-round setting, ``re-defining'' the intermediate state to be $\brho = U(0)(\ketbra{0}_{\RegZ} \otimes \brho'_{\RegI})U(0)^\dagger$ can be implemented by replacing any unitary $W$ applied in the previous round with $U(0)W$; this follows the conventions used in~\cite{FOCS:Mahadev18a}.} Moreover, the unitaries $\{U(r)\}_r$ can be treated as independent of the (classical) protocol transcript before challenge $r$ is sent, since we can assume this transcript is saved in $\brho$. 

\subsection{Interactive Arguments}

In what follows we define the notion of an interactive argument for $\QMA$ languages.  We denote such arguments by $(P,V)$, and denote the output bit of the verifier by $\mathsf{Out}(P, V)$.
\begin{definition}\label{def:interactive-argument}
An interactive argument $(P,V)$ for a language $\cL=(\cL_\mathsf{yes},\cL_\mathsf{no})\in\QMA$ with relation $\cR(x)$ is a (classical) 2-party interactive protocol between a QPT prover $P$ and a p.p.t.\ verifier $V$, with the following completeness and soundness guarantees:

\paragraph{Completeness.}
For all $\lambda\in \mathbb{N}$, there exists a polynomial $k =k(\secp)$ such that for all $x \in \cL_{\text{yes}}$, and all $\ket{\phi}\in\cR(x)$, it holds that
\[
\Pr\left[\mathsf{Out}\left(P(\ket{\phi}^{\otimes k(\lambda)},x), V(x)\right) = 1\right] \geq 1 - \negl(\secp).
\]

\paragraph{Computational Soundness.}
For all $\lambda\in \mathbb{N}$, all $x \in \cL_{\mathsf{no}}$, and all non-uniform QPT provers $P^*$,
it holds that
\[
\Pr\left[\mathsf{Out}(P^*(x), V(x)) = 1\right] \leq \negl(\secp).
\]
\end{definition}

\paragraph{Batch Arguments.} We also consider a sub-class of interactive arguments where the prover simultaneously engages the verifier on $n$ sub-instances $(x_1, \dots, x_n)$, where each $x_i$ is supposed to be a Yes-instance of a fixed language $\mathcal L_i$. We require the following notion of (computational) soundness.

\begin{definition}[Soundness]\label{def:somewhere-soundness}
An interactive argument $(P, V)$ for a batch language $\cL = \cL_1 \times \hdots \times \cL_n \in\QMA$ with relation $\cR(x)$ is sound if for all $\lambda\in \mathbb{N}$, all polynomials $n=n(\secp)$, all indices $i\in[n]$, all statements $(x_1, \dots, x_n)$, where $x_i \in \cL_{\mathsf{no}}$, and all non-uniform QPT provers $P^*$, it holds that
\[
\Pr\left[\mathsf{Out}(P^*(x_1, \dots, x_n), V((x_1, \dots, x_n))) = 1\right] \leq \negl(\secp).
\]
\end{definition}

\subsection{Computational Indistinguishability}
Two classical distribution ensembles $\{(X^{(\secp)}, Y^{(\secp)})\}_\secp$ are said to be post-quantum computationally indistinguishable if for every non-uniform QPT algorithm $A = \{(A^{(\secp)}, \brho^{(\secp)})\}_{\secp}$ (that outputs a bit $b$), we have that \[\left| \mathbf \E\left[A^{(\secp)}(X^{(\secp)},\brho^{(\secp)})\right] - \mathbf \E\left[A^{(\secp)}(Y^{(\secp)},\brho^{(\secp)})\right]\right| = \negl(\secp).\]

Two quantum state ensembles $\{\brho_0^{(\secp)}, \brho_1^{(\secp)}\}_\secp$ are said to be \emph{computationally indistinguishable} if for every non-uniform QPT algorithm $A = \{A^{(\secp)}, \brho^{(\secp)} \}$ (that outputs a bit $b$), we have that 
\[\Big| \mathbf \E\left[A^{(\secp)}(\brho^{(\secp)},\brho^{(\secp)}_0)\right] - \mathbf \E\left[A^{(\secp)}(\brho^{(\secp)},\brho^{(\secp)}_1)\right]\Big| = \negl(\secp).
\]

Equivalently,  $\{\brho_0^{(\secp)}, \brho_1^{(\secp)}\}_\secp$ are computationally indistinguishable if for every efficiently computable non-uniform binary observable ($R, \bsigma)$, we have that 
\[\Big|\Tr(R (\brho_0 \tensor \bsigma) ) - \Tr(R (\brho_1 \tensor \bsigma)) \Big| = \negl(\secp).
\]
We will occasionally use the notation $\brho_0 \approx_c \brho_1$ to denote computational indistinguishability of $\{\brho_0^{(\secp)}, \brho_1^{(\secp)}\}_\secp$. 

More generally, we use $(T(\secp), \varepsilon(\secp))$-indistinguishability to denote computational indistiguishability as above where the distinguisher is allowed to run in time $T$ and the advantage is required to be at most $\varepsilon$.

\subsection{Mahadev Randomized TCFs}
\label{sec:mahadev-rtcf}

In this section, we define the cryptographic primitive used by Mahadev~\cite{FOCS:Mahadev18a} to obtain a (non-succinct) delegation scheme for $\QMA$ with classical verification. The primitive is closely related to Regev encryption \cite{STOC:Regev05} and LWE-based ``lossy'' trapdoor functions \cite{STOC:PeiWat08,C:PeiVaiWat08,STOC:GorVaiWic15}, but makes use of special-purpose structure relevant for quantum functionality. Most of this special-purpose structure, in particular, the ``adaptive hardcore bit'', was introduced in the work of Brakerski, Christiano, Mahadev, Vazirani and Vidick~\cite{FOCS:BCMVV18}, but \cite{FOCS:Mahadev18a} further requires ``dual-mode key generation'' in addition to the \cite{FOCS:BCMVV18} properties. Given the numerous special-purpose requirements, we refer to the primitive as ``Mahadev randomized trapdoor claw-free functions (rTCFs).''\footnote{Actually, \cite{FOCS:Mahadev18a} requires an extra (second) hardcore bit property (Property 2 of Definition 4.4 in \cite{FOCS:Mahadev18a}) that we drop from our definition, as our proof does not require it.} 
\begin{definition}
  \label{def:clawfree}
  A Mahadev randomized trapdoor claw-free function family (Mahadev rTCF) $\clawfree$ is described by a tuple of efficient classical algorithms $(\Gen,\Eval,\Invert, \Check, \Good)$ with the following syntax:
  \begin{itemize}
      \item $\Gen(1^\lambda,\mathsf{mode})$ is a dual-mode PPT key generation algorithm that takes as input a security parameter $\secp$ in unary, and a bit $\mathsf{mode}\in\{0,1\}$,  and it outputs a public key $\pk$ and a private key $\sk$. The description of the public key implicitly defines a domain of the form $\{0,1\}\times\mathcal{D}_{\pk}$ for the randomized function $f_{\pk}$. We view $\mathcal{D}_{\pk}$ as an explicit (efficiently verifiable and samplable) subset of $\{0,1\}^{\ell(\secp)}$, so that applying bit operations to elements of $\mathcal{D}_{\pk}$ is well-defined. 
      
      In our context, $\mode=0$ samples keys for an {\em injective} function and $\mode=1$ samples keys for  a {\em two-to-one} function.  For the sake of readability, we use a descriptive notation by which $\mode\in \{\inj,\twotoone\}$, where $\mode=\inj$ corresponds to $\mode=0$ and $\mode=\twotoone$ corresponds to $\mode=1$.
      \item $\Eval(\pk,b,\vecx)$ is a (possibly probabilistic) algorithm that takes as input a public key $\pk$, a bit $b\in\{0,1\}$ and an element $\vecx \in \mathcal{D}_{\pk}$, and outputs a string $\vecy$ with distribution $\chi$.
      \item $\Invert(\mode,\sk,\vecy)$ is a deterministic algorithm that takes as input $\mode\in\{\inj,\twotoone\}$, a secret key $\sk$, and an element $\vecy$ in the range.  If $\mode=\inj$ then it outputs a pair $(b,\vecx)\in\{0,1\}\times\mathcal{D}_\pk$ or $\bot$. If $\mode=\twotoone$ then it outputs two pairs $(0,\vecx_0)$ and $(1,\vecx_1)$ with $\vecx_0,\vecx_1\in \mathcal{D}_\pk$, or $\bot$. 
      \item $\Check(\pk,b,\vecx,\vecy)$ is a deterministic algorithm takes as input a public key $\pk$, a bit $b\in\{0,1\}$, an element $\vecx\in \mathcal{D}_\pk$, and an element $\vecy$ in the range, and it outputs a bit.
      \item $\Good(\vecx_0, \vecx_1, \vecd)$ is a deterministic poly-time algorithm that takes as input two domain elements $\vecx_0, \vecx_1 \in \mathcal{D}_{\pk}$ and a string $\vecd \in \{0,1\}^{\ell + 1}$. It outputs a bit that characterizes membership in a set that we call 
      \[\Good_{\vecx_0, \vecx_1} \coloneqq \{\vecd : \Good(\vecx_0,\vecx_1,\vecd) = 1\}.\]
      Moreover, we stipulate that $\Good(\vecx_0, \vecx_1, \vecd)$ ignores the first bit of $\vecd$.\footnote{We depart slightly from notation in prior work, which defines $\vecd$ to be an element of $\{0,1\}^\ell$ (corresponding to the last $\ell$ bits of our $\vecd$).} 
      
  \end{itemize}
  We require that the following properties are satisfied.
  \begin{enumerate}
      \item \textbf{\em Correctness:} 
      \begin{enumerate}
          \item 
      For \emph{all}
      ~$(\pk,\sk)$ in the support of $\Gen(\inj,1^\secp)$: For every $b\in\{0,1\}$, every  $\vecx\in\mathcal{D}_\pk$, and every $\vecy\in{\sf Supp}(\Eval(\pk,(b,\vecx)))$, 
      $$ \Invert(\inj,\sk,\vecy) = (b,\vecx).\footnote{Note that this implies that ${\sf Supp}(\Eval(\pk,(b_1,\vecx_1)))\cap {\sf Supp}(\Eval(\pk,(b_2,\vecx_2)))=\emptyset$  for every $(b_1,\vecx_1)\neq (b_2,\vecx_2)$. This can be enforced in the LWE-based instantiation by using \emph{truncated} discrete Gaussian errors}$$

      \item For \emph{all} $(\pk,\sk)$ in the support of $\Gen(\twotoone,1^\secp)$: For
      every $b\in\{0,1\}$, every  $\vecx\in\mathcal{D}_\pk$, and every $\vecy\in{\sf Supp}(\Eval(\pk,(b,\vecx)))$, 
      $$ \Invert(\twotoone,\sk,\vecy) = ((0,\vecx_0),(1,\vecx_1))$$ such that $\vecx_b=\vecx$ and  $\vecy\in{\sf Supp}(\Eval(\pk,(\beta,\vecx_\beta)))$ for every $\beta\in\{0,1\}$. 
      
      \item For every $(\pk,\sk)\in {\sf Supp}(\Gen(\twotoone,1^\secp)) \cup{\sf Supp}(\Gen(\inj,1^\secp))$, every $b\in\{0,1\}$ and every $\vecx\in\mathcal{D}$, 
      $$
      \Pr[\Check(\pk,(b,\vecx),\vecy)=1]=1
      $$
      if and only if $\vecy\in{\sf Supp}(\Eval(\pk,(b,\vecx)))$.
      \item For every $(\pk, \sk)$ in the support of $\Gen(\twotoone, 1^\secp)$ and every pair of domain elements $\vecx_0, \vecx_1$, the density of $\Good_{\vecx_0, \vecx_1}$ is $1-\negl(\secp)$.      \end{enumerate}
      
      \item \textbf{\em Key Indistinguishability:} 
      $$ \{ \pk: (\pk,\sk) \gets \Gen(\twotoone,1^\secp) \} \approx_c \{ \pk: (\pk,\sk) \gets \Gen(\inj,1^\secp) \}$$ 
      
      \item \textbf{\em Adaptive Hardcore Bit:}  For every ${\sf BQP}$ adversary ${\cal A}=(\cal{A}_0,\cal{A}_1)$ there exists a negligible function~$\mu$ such that for every $\secp\in\mathbb{N}$, the following difference of probabilities is equal to $\mu(\secp)$: 
      \begin{align*}
      & \Big| \Pr[{\cal A}_1(\pk,\vecy)=(\vecd,(b,\vecx)): \Check(\pk,b, \vecx, \vecy) = 1 ~~\wedge~~ \vecd\cdot (1, \vecx_0\oplus \vecx_1)=0 ~~\wedge~~ \vecd\in \Good_{\vecx_0,\vecx_1}] \\
      & -\Pr[{\cal A}_1(\pk,\vecy)=(\vecd,(b,\vecx)): \Check(\pk,b, \vecx, \vecy) = 1 ~~\wedge~~ \vecd\cdot (1, \vecx_0\oplus \vecx_1)=1 ~~\wedge~~ \vecd\in \Good_{\vecx_0,\vecx_1}]\Big|
      \end{align*}
      
       where the probabilities are over the experiment that generates $(\pk,\sk)\leftarrow \Gen(\twotoone,1^\secp)$,  $\vecy\leftarrow\mathcal{A}_0(\pk)$, and where $((0,\vecx_0),(1,\vecx_1))=\Invert(\twotoone,\sk,\vecy)$.

  \end{enumerate}
\end{definition}

\begin{lemma}[\cite{FOCS:BCMVV18,FOCS:Mahadev18a}]
  Assuming LWE, there is a collection of Mahadev randomized TCFs. 
\end{lemma}

\begin{remark}
For some of our applications (and for simplicity of proofs), we will actually require an rTCF that is \emph{perfectly correct}, which means that the correctness properties (a) and (b) hold with probability 1. That is, they hold for all $(\pk,\sk) \in \Gen(\inj,1^\secp)$ and $(\pk,\sk)\in \Gen(\twotoone,1^\secp)$ respectively. We briefly argue that this is possible. In the injective mode case, this is possible because the sampling procedure for injective keys given in \cite[Section 9.2]{FOCS:Mahadev18a} can determine whether the key it sampled is indeed injective and if not, output a fixed hard-coded injective key. In the 2-to-1 mode case, the sampling procedure given in \cite[Section 4.1]{FOCS:BCMVV18} is perfect except for when $s = 0^n$. Thus, we can again hard-code a fixed 2-to-1 key to output instead whenever $s = 0^n$.
\end{remark}

\subsection{Collapsing Hash Functions}

\paragraph{Collapsing Hash Functions.} Let $H = \{H_\lambda\}_{\lambda \in \mathbb{N}}$ be a hash function family where each $H_\lambda$ is a distribution over functions $h: \{0,1\}^{n(\lambda)}\rightarrow \{0,1\}^{\ell(\lambda)}$. 

Define the collapsing experiment $\texttt{CollapseExpt}_{H,\lambda,b}(D)$ on quantum distinguisher $D$ as follows.
\noindent $\texttt{CollapseExpt}_{H,\lambda,b}(D)$:
\begin{enumerate}
    \item The challenger samples $h \gets H_\lambda$ and sends $h$ to the distinguisher $D$.
    \item The distinguisher replies with a classical binary string $y \in \{0,1\}^{\ell(\lambda)}$ and an $n(\lambda)$-qubit quantum state on the register $\RegX$. Note that the requirement that $y$ be classical can be enforced by having the challenger immediately measure these registers upon receiving them.
    \item The challenger computes $h$ in superposition on the $n(\lambda)$-qubit quantum state, and measures the bit indicating whether the output of $h$ equals $y$. If the output does not equal $y$, the challenger aborts and outputs $\bot$.
    \item\label{item:collapsing-challenge} If $b = 0$, the challenger does nothing. If $b =1$, the challenger measures the $n(\lambda)$-qubit state in the standard basis.
    \item The challenger returns the contents of the $\RegX$ register to the distinguisher.
    \item The distinguisher outputs a bit $b'$.
\end{enumerate}

\begin{definition}[\cite{EC:Unruh16}] $H = \{H_\lambda\}_{\lambda}$ is \emph{collapsing} if for every security parameter $\lambda \in \mathbb{N}$ and any polynomial-size quantum distinguisher $\mathcal{D} = \{D_\lambda\}_{\lambda}$, there exists a negligible function $\mu$ such that
\[ \left|
\Pr[\texttt{\emph{CollapseExpt}}_{H,\lambda,0}(D_\lambda)=1]-\Pr[\texttt{\emph{CollapseExpt}}_{H,\lambda,1}(D_\lambda) = 1] \right| \leq  \mu(\lambda).\]
\end{definition}

Unruh~\cite{AC:Unruh16} constructs collapsing hash functions from lossy functions, which can be based on LWE~\cite{STOC:PeiWat08}. 

\begin{lemma}[\cite{STOC:PeiWat08,AC:Unruh16}]\label{lemma:lwe-collapsing}
  Assuming LWE, a family of collapsing hash functions $\{H_\lambda:\,\{0,1\}^*\to\{0,1\}^\lambda\}_{\lambda}$ exists.
\end{lemma}

\subsection{Fully Homomorphic Encryption}\label{sec:FHE-definition}

We define fully homomorphic encryption (FHE), which is used in \cref{sec:final}. A fully homomorphic encryption scheme $\FHE = (\FHE.\Gen, \FHE.\Enc, \allowbreak \FHE.\Dec, \FHE.\Eval)$ for (classical) polynomial-time computation is a tuple of four PPT algorithms.

\begin{itemize}
    \item $\Gen(1^\secp)$ takes as input the security parameter and outputs a key pair $(\pk, \sk)$.
    \item $\Enc(\pk, m)$ takes as input a message $m$ and outputs a ciphertext $\ct$. 
    \item $\Eval(f, \ct)$ takes as input a ciphertext $\ct$ corresponding to an $n$-bit plaintext as well as a function $f: \{0,1\}^n \rightarrow \{0,1\}$. It outputs a ciphertext $\ct_f$.
    \item $\Dec(\sk, \ct)$ takes as input the secret key and a ciphertext. It outputs a message. 
\end{itemize}

\noindent We require the following properties.

\begin{itemize}
    \item \textbf{Evaluation/Decryption Correctness:} for any (polynomial-size circuit) function $f: \{0,1\}^n \rightarrow \{0,1\}$ and any message $m\in \{0,1\}^n$, we have that
    \[ \Dec(\sk, \Eval(f, \Enc(\pk, m))) = f(m)
    \]
    with probability $1-\negl(\secp)$ over the parameter sampling. 
    \item \textbf{Compactness:} we require that $\FHE.\Eval(f, \Enc(\pk, m))$ has a fixed size $\poly(\secp)$ independent of $|f|, |m|$. 
    \item \textbf{Semantic Security:} For any pair of messages $(m_0, m_1)$, we have that $(\pk, \FHE.\Enc(\pk, m_0)) \approx_c (\pk, \FHE.\Enc(\pk, m_1))$. 
\end{itemize}

\begin{theorem}[\cite{STOC:Gentry09,FOCS:BraVai11,ITCS:BraGenVai12, ITCS:BraVai14}]
Under circular-secure variants of the Learning with Errors assumption, there exists a fully homomorphic encryption scheme for all polynomial-time computable functions. If the circular LWE variant is post-quantum, then so is the FHE scheme. 

Under the standard LWE assumption, there exists a FHE scheme for all polynomial-size circuits of depth $d(\secp)$, where the scheme has compactness $\poly(\secp, d)$. 
\end{theorem}

\subsection{Indistinguishability Obfuscation}\label{subsec:iO-def}

An indistinguishability obfuscator (iO) is an algorithm $\iO$ that takes as input a circuit $C$ and satisfies the following properties.
\begin{itemize}
    \item \textbf{Functional Equivalence:} for any (polynomial-size) circuit circuit $C: \{0,1\}^n \rightarrow \{0,1\}^m$ and any input $x\in \{0,1\}^n$, we have that
    $$\iO(C)(x) = C(x).$$
    \item \textbf{Security:} For any pair of functionally equivalent circuits $(C_0, C_1)$, we have that $$\iO(C_0) \approx_c \iO(C_1).$$ 
\end{itemize}
We mention that, while some recent candidates for iO (such as  \cite{STOC:JaiLinSah21}) can be broken using quantum algorithms, others, such as \cite{TCC:BGMZ18,C:CheVaiWee18,BDGM20,EC:WeeWic21,STOC:GayPas21,DQVWW21}, are plausibly post-quantum secure. Furthermore, it will be convenient for us to assume iO with perfect correctness to simplify our analysis (in particular the argument in Section~\ref{sec:zk}). We point out that this property is already satisfied by most candidates and can also be attained via generic transformations~\cite{EC:BitVai17}.

\newcommand{\Puncture}{\mathsf{Puncture}}
\newcommand{\PuncEval}{\mathsf{PuncEval}}
\newcommand{\punc}[2]{#1\{ #2 \}}
\subsection{Puncturable PRFs}
\begin{definition}[Puncturable PRF \cite{AC:BonWat13,PKC:BoyGolIva14,CCS:KPTZ13,STOC:SahWat14}]\label{def:pprf}
  A puncturable PRF family is a family of functions
  \[
    \cF = \left\{ F_{\secp,s} : \{0,1\}^{\nu(\secp)} \to \{0,1\}^{\mu(\secp)} \right \}_{\secp \in \mathbb N, s \in \{0,1\}^{\ell(\secp)}}
  \]
  with associated (deterministic) polynomial-time algorithms
  $(\cF.\Eval, \cF.\Puncture, \allowbreak \cF.\PuncEval)$ satisfying
  \begin{itemize}
  \item For all $x \in \{0,1\}^{\nu(\secp)}$ and all $s \in \{0,1\}^{\ell(\secp)}$, $\cF.\Eval(s, x) = F_{\secp,s}(x)$.
  \item For all distinct $x, x' \in \{0,1\}^{\nu(\secp)}$ and all
    $s \in \{0,1\}^{\ell(\secp)}$,
    $$\cF.\PuncEval(\cF.\Puncture(s, x), x') = \cF.\Eval(s, x')$$
  \end{itemize}
  For ease of notation, we write $F_s(x)$ and $\cF.\Eval(s, x)$ interchangeably,
  and we write $\punc{s}{x}$ to denote $\cF.\Puncture(s, x)$.

  $\cF$ is said to be $(s, \delta)$-secure if for every
  $\{x^{(\secp)} \in \{0,1\}^{\nu(\secp)}\}_{\secp \in \mathbb N}$, the following two distribution
  ensembles (indexed by $\secp$) are $\delta(\secp)$-indistinguishable to circuits of
  size $s(\secp)$:
  \[
    (\punc{S}{x^{(\secp)}}, F_S(x^{(\secp)})) \text{ where $S \gets \{0,1\}^{\ell(\secp)}$}
  \] and
  \[
    (\punc{S}{x^{(\secp)}}, U) \text{ where $S \gets \{0,1\}^{\ell(\secp)}$, $U \gets \{0,1\}^{\mu(\secp)}$}.
  \]
\end{definition}

\begin{theorem}[\cite{FOCS:GolGolMic84,CCS:KPTZ13,AC:BonWat13,PKC:BoyGolIva14,STOC:SahWat14}]
	If \{polynomially secure, subexponentially secure\} one-way functions exist, then for all functions $\mu: \mathbb N \to \mathbb N$ (with $1^{\mu(\nu)}$ polynomial-time computable from $1^{\nu}$), and all $\delta : \mathbb N \rightarrow [0,1]$ with $\delta(\nu) \geq 2^{-\poly(\nu)}$, there are polynomials $\ell(\secp), \nu(\secp)$ and a \{polynomially secure, $(\frac 1 {\delta(\nu(\secp))}, \delta(\nu(\secp)))$-secure\} puncturable PRF family \[\cF_{\mu} = \left\{F_{\secp, s} : \{0,1\}^{\nu(\secp)} \rightarrow \{0,1\}^{\mu(\nu(\secp))}\}_{\secp \in \mathbb N, s \in \{0,1\}^{\ell(\secp)}}\right\}.\]
\end{theorem}
\fi

\newcommand{\SimGen}{\mathsf{SimGen}}
\newcommand{\pparam}{\mathsf{pp}}
\newcommand{\sparam}{\mathsf{sp}}
\newcommand{\acc}{\mathsf{acc}}
\newcommand{\rej}{\mathsf{rej}}
\renewcommand{\st}{\mathsf{st}}
\newcommand{\Ver}{\mathsf{Ver}}
\newcommand{\TCF}{\mathsf{TCF}}
\newcommand{\pp}{\mathsf{pp}}
\newcommand{\mask}{\mathsf{mask}}
\newcommand{\Inv}{\mathsf{Inv}}
\newcommand{\td}{\mathsf{td}}
\newcommand{\ek}{\mathsf{ek}}

\section{Commit-and-Measure Protocols}\label{sec:measurement-protocol-definition}

\subsection{Defining Commit-and-Measure Protocols}\label{subsec:commit-and-measure}

In this section, we formalize the notion of a \emph{commit-and-measure protocol}, which was informally described in \cite{FOCS:Mahadev18a}. A commit-and-measure protocol enables a classical verifier to obtain the results of measuring, in the standard or Hadamard basis, each qubit of an $N$-qubit quantum state $\bsigma$ held by the prover. More precisely, the verifier encodes its choice of basis with a classical circuit $C: [N] = \{0,1\}^{\log N} \rightarrow \{0,1\}$, where $C(i) = b$ specifies the basis for the measurement of the $i$th qubit. We adopt the convention that $b = 0$ corresponds to the standard basis and $b = 1$ corresponds to the Hadamard basis. Note that in Mahadev's original protocol, $C$ is given as an explicit string $(C(0),C(1),\dots,C(N-1))$, but our eventual \emph{succinct} protocols will require circuits $C$ with size much smaller than $N$.

\begin{definition}[Commit-and-Measure Protocol Syntax]\label{def:measurement-protocol-syntax}
An $N$-qubit commit-and-measure protocol between a quantum polynomial-time prover $P = (\Commit,\Open)$ and a classical probabilistic polynomial-time verifier $V = (\Gen,\Test,\Out)$ has the following syntax.
\begin{enumerate}
    \item The verifier samples $(\pk,\sk) \gets  \Gen(1^\secp,C)$, where $C:[N] = \{0,1\}^{\log N} \to \{0,1\}$ represents a \emph{basis vector} $h \in \{0,1\}^N$, obtaining public parameters $\pk$ and secret parameters $\sk$. It sends the public parameters $\pk$ to the prover.
    \item The prover computes $(y,\brho) \gets \Commit(\pk,\bsigma)$, obtaining a classical ``commitment'' string $y$ and a private quantum state $\brho$. It sends $y$ to the verifier.
    \item The verifier samples a random challenge bit $c \gets \{0,1\}$ and sends $c$ to the prover; $c = 0$ corresponds to a ``test round'' and $c = 1$ corresponds to a ``measurement round''.
    \item The prover computes $z \gets \Open(\brho,c)$, obtaining a classical string $z$ that it sends to the verifier.
    \item If $c = 0$, the verifier computes $\{\acc,\rej\} \gets \Test(\pk,(y,z))$.
    
    If $c = 1$, the verifier computes $m \gets \Out(\sk,(y,z))$ to obtain a classical string $m \in \{0,1\}^{N}$ of measurement outcomes. 
\end{enumerate}
\end{definition}

The protocol is required to satisfy the following completeness (\cref{def:measurement-protocol-completeness}) and soundness (\cref{def:measurement-protocol-soundness}) properties. For the definitions below, we write $M(h,\bsigma)$ to denote the distribution of outcomes from measuring $\bsigma$ in the basis $h$.

\begin{definition}[Completeness]\label{def:measurement-protocol-completeness}
A commit-and-measure protocol is required to satisfy two completeness properties.
    
    \begin{enumerate}
        \item (Test Round Completeness) 
        For all $C: [N] \to \{0,1\}$ and $N$-qubit states $\bsigma$:
        \[\Pr\left[\acc \gets \Test(\pk,(y,z)): \begin{array}{r} (\pk,\sk) \gets \Gen(1^\secp,C) \\ (y,\brho) \gets \Commit(\pk,\bsigma) \\ z \gets \Open(\brho,0) \\ \end{array}\right] = 1-\negl(\secp).\]
        \item (Measurement Round Completeness) For all $C: [N] \to \{0,1\}$ and $N$-qubit states $\bsigma$:
        \[\left\{m \gets \Out(\sk,(y,z)): \begin{array}{r} (\pk,\sk) \gets \Gen(1^\secp,C) \\ (y,\brho) \gets \Commit(\pk,\bsigma) \\ z \gets \Open(\brho,1) \\ \end{array}\right\} \approx_c M(h,\bsigma),\]
        where $h \in \{0,1\}^N$ is such that $h_i = C(i)$ for all $i \in [N] = \{0,1\}^{\log N}$. 
    \end{enumerate}
\end{definition}

\begin{remark}
The~\cite{FOCS:Mahadev18a} protocol satisfies \emph{statistical} measurement round completeness, but our verifier-succinct commit-and-measure protocol will not.
\end{remark}

\begin{remark}
One of our applications will require a measurement protocol with \emph{perfect} completeness, which stipulates that the above completeness guarantees hold over \emph{all} $(\pk,\sk) \in \Gen(1^\secp,C)$ (and where the measurement round completeness is statistical rather than computational). This can be achieved by using an rTCF with perfect correctness, which we discuss in \cref{sec:mahadev-rtcf}, and, in the succinct case, an indistinguishability obfuscation scheme with perfect correctness (\cref{subsec:iO-def}).
\end{remark}

To state our soundness definition (\cref{def:measurement-protocol-soundness}), we first specify the registers that any non-uniform  cheating prover acts on:
\begin{itemize}
    \item $\RegP$ contains the public parameters $\pk$, 
    \item $\RegY$ contains the classical commitment string $y$,
    \item $\RegZ$ contains the classical opening string $z$,
    \item $\RegI$ contains the prover's initial state and its internal work registers.
\end{itemize}
In a protocol execution, $\RegP$ is initialized with $\ketbra{\pk}$. A non-uniform cheating prover $$P^* = (\brho_0,U_{\Commit^*},U_{\Open^*,0},U_{\Open^*,1})$$ is parameterized by:
\begin{itemize}
    \item An arbitrary quantum state $\brho_0 \in \mathrm{D}(\RegY\tensor \RegZ \tensor \RegI) $. In a protocol execution with $\widetilde P$, $\RegY \otimes \RegZ \otimes \RegI$ is initialized with $\brho_0$.
    \item An adversarial commitment unitary $U_{\Commit^*}$ on $\RegP \otimes \RegY \otimes \RegZ \otimes \RegI$ of the form
    \[ \sum_{\pk} \ketbra{\pk}_{\RegP} \otimes (U_{\Commit^*,\pk})_{\RegY,\RegZ,\RegI}.\]
    That is, $U_{\Commit^*}$ is classically controlled on $\RegP$. In particular, the adversarial prover's commitment on verifier message $\pk$ is obtained by measuring register $\RegY$ of
    \[ U_{\Commit^*}(\ketbra{\pk}_{\RegP} \otimes (\brho_0)_{\RegY, \RegZ, \RegI})\]
    in the computational basis to obtain $y$.
    \item An adversarial opening unitary $U_{\Open^*,0}$ on $\RegP \otimes \RegY \otimes \RegZ \otimes \RegI$ corresponding to the prover's behavior in the test round $(b = 0)$ of the form:
    \[ \sum_{\pk,y} \ketbra{\pk,y}_{\RegP, \RegY} \otimes (U_{\Open^*,0,\pk,y})_{\RegZ, \RegI}.\] 
    That is, $U_{\Open^*,0}$ is classically controlled on $\RegP$ and $\RegY$.
    In particular, given a commitment string $y$ and residual prover state $\brho \in \mathrm{D}(\RegZ\tensor \RegI)$, the prover's response on challenge $c=0$ is obtained by measuring register $\RegZ$ of \[U_{\Open^*, 0}(\ketbra{\pk,y}_{\RegP,\RegY} \otimes \brho_{\RegZ, \RegI})\]
    in the computational basis to obtain $z$.
    \item An adversarial opening unitary $U_{\Open^*,1}$ on $\RegP \otimes \RegY \otimes \RegZ \otimes \RegI$ corresponding to the prover's behavior in the measurement round $(b = 1)$ of the form:
    \[ \sum_{\pk,y} \ketbra{\pk,y}_{\RegP, \RegY} \otimes (U_{\Open^*,1,\pk,y})_{\RegZ, \RegI}.\] 
    In particular, given a commitment string $y$ and residual prover state $\brho \in \mathrm{D}(\RegZ\tensor \RegI)$, the prover's response on challenge $c=1$ is obtained by measuring register $\RegZ$ of \[U_{\Open^*, 1}(\ketbra{\pk,y}_{\RegP,\RegY} \otimes \brho_{\RegZ, \RegI})\]
    in the Hadamard basis to obtain $z$.
\end{itemize}

\noindent Following~\cite{FOCS:Mahadev18a}, we can assume without loss of generality that $U_{\Open^*,0}$ is the identity (refer to \cref{subsec:black-box} for additional details). We will therefore write $U$ to describe the prover's ``attack unitary'' for the measurement round $(c=1)$.

For defining soundness, we informally require that a prover $P^*$ that passes the test round with probability $1-\negl(\secp)$ (this could alternatively be enforced by applying a measurement in the security game) implicitly defines\footnote{In fact, we require that $\btau$ can be extracted efficiently from $P^*$.} an $N$-qubit state $\btau$ whose measurement outcome distribution matches the output distribution of $\Out(\cdot)$ (up to computational indistinguishability). 

\begin{definition}[Soundness]\label{def:measurement-protocol-soundness}
There exists an efficient classical algorithm $\SimGen(1^\secp)$ and an efficient quantum algorithm $\Ext^{U,\brho}(\pk,\sk,y)$ with black-box access to an attacker parameterized by a state $\brho$ and a unitary $U$ (see \cref{subsec:black-box} for more details on how we formalize quantum black-box access), that takes as input classical strings $(\pk,\sk,y)$, and satisfies the following properties:

\begin{itemize} 

\item Consider any non-uniform QPT cheating prover $P^* = (\brho_0,U_{\Commit^*},U)$ that passes the test round with probability $1-\negl(\secp)$ for all $h \in \{0,1\}^N$. 

Then,  for all $h \in \{0,1\}^N$ with circuit representation $C$, the following two distributions are computationally indistinguishable:

    $\texttt{Real}$:
    \begin{enumerate}
        \item Sample parameters $(\pk,\sk) \gets \Gen(1^\secp,C)$.
        \item Run the attacker $P^*$ on $\pk$ to obtain a classical commitment string $y$ (i.e., apply $U_{\Commit^*}$ and then measure the register containing $y$). Denote the post-measurement state as $\brho \in \mathrm{D}(\RegZ \otimes \RegI)$, where $\RegZ$ corresponds to the registers that will eventually be measured to obtain the prover's final message, and $\RegI$ contains all of the other internal registers of the prover.\footnote{We will also assume, without loss of generality, that the prover always copies $\pk$ and $y$ into its internal state registers $\RegI$.}
        
        \item Apply the prover's attack unitary $U$. This yields the state $\brho' \coloneqq U \brho_{\RegZ,\RegI}U^\dagger$.
        Measure the $\RegZ$ register of $\brho'$ in the Hadamard basis to obtain the prover's opening string $z$.
        \item Compute $m \gets \Out(\sk,(y,z))$ and output $m$. 
    \end{enumerate}
    
    $\texttt{Sim}$:
    \begin{enumerate}
        \item Sample parameters $(\pk,\sk) \gets \SimGen(1^\secp)$.
        \item Run the attacker $P^*$ on $\pk$ to obtain a classical commitment string $y$ (i.e., apply $U_{\Commit^*}$ and then measure the register containing $y$). Denote the post-measurement state as $\brho \in \mathrm{D}(\RegZ \otimes \RegI)$.
        
        \item Run $\Ext^{U,\brho}(\pk,\sk,y) \rightarrow \btau$ to obtain an $N$-qubit state $\btau$.
        \item Measure each qubit of $\btau$ according to the bases specified by $h \in \{0,1\}^N$ (i.e., qubit $i$ is measured in the Hadamard basis if $h_i = 1$ and the standard basis if $h_i = 0$) and output the result. 

    \end{enumerate}

\end{itemize}

\end{definition}

\section{A Measurement Protocol Template}\label{sec:meas-protocol}

In this section, we describe a generic construction of a $N$-qubit commit-and-measure protocol (\cref{sec:measurement-protocol-definition}) using two building blocks: (1) a family of Mahadev rTCFs (\cref{def:clawfree}), and (2) a ``batch key generation'' scheme (fully defined in \cref{sec:succinct-keygen}) whose syntax we describe below. We consider two different instantiations of this template:

\begin{itemize}
    \item Using a ``trivial'' batch key generation scheme in which the $N$ rTCF keys are sampled i.i.d., we recover Mahadev's original protocol \cite{FOCS:Mahadev18a}.
    \item Using a succinct key generation scheme (constructed in \cref{sec:succinct-keygen,sec:succinct-clawfree}, we obtain a measurement protocol in which the verifier's messages are succinct. We refer to this as a \textbf{verifier-succinct} measurement protocol.
\end{itemize}

\paragraph{Batch Key Generation.} For our construction, we make use of what we call a ``batch key generation scheme'' for the Mahadev rTCF. Let $\mathsf{TCF}.\Gen(1^\secp, \mode)$ denote the ``standard'' key generation algorithm for a Mahadev rTCF. Informally, a batch key generation scheme for $\mathsf{TCF}.\Gen(1^\secp, \mode)$ is a mechanism that produces a joint representation of $N$ TCF pairs $(\pk_i, \sk_i)$, from which any individual $\pk_i, \sk_i$ can be computed, such that the pairs $(\pk_i, \sk_i)$ are sufficiently ``independent'' of each other.

A full definition of a batch key generation scheme is given in \cref{def:succinct-key-gen}, but we formally state here the relevant syntax and security properties. Syntactically, a batch key generation scheme includes three algorithms $(\Gen, \extpk,\extsk)$, where:

\begin{itemize}
    \item $\Gen(1^\secp, C)$ takes as input a security parameter $\secp$ and a circuit $C: [N]\rightarrow \{0,1\}$ representing (through its truth table) an $N$-bit string. It outputs a master public key $\PK$ and master secret key $\SK$.
    \item $\extpk(\PK, i)$ is a deterministic algorithm that takes as input $\PK$ and an index $i\in N$, and outputs a public key $\pk_i$.
    \item $\extsk(\SK, i)$ is a deterministic algorithm that takes as input $\SK$ and an index $i\in N$, and outputs a secret key $\sk_i$.
\end{itemize}

When instantiated for a Mahadev rTCF family, we require the following properties to hold for such a procedure:

\begin{itemize}
    \item \textbf{Correctness}: for $(\PK, \SK)\gets \Gen(1^\secp, N, C)$ and $(\pk_i, \sk_i) = (\extpk(\PK, i), \extsk(\SK, i))$, we have that $(\pk_i, \sk_i)$ is in the range of $\TCF.\Gen(1^\secp, C(i))$ (i.e. they are a valid key pair in mode $C(i)$).
    \item \textbf{Key Indistinguishability}: if $C_1$ and $C_2$ represent functions that agree on a set $T$ of inputs, then $\PK$ output by $\Gen(1^\secp, C_1)$ is computationally indistinguishable from $\PK$ output by $\Gen(1^\secp, C_2)$, \emph{even} in the presence of all $\{\sk_i, i\in T\}$. 
    \item \textbf{Collapsing at a single index}: For any index $j$, the function $f_{\pk_j}$ is collapsing even given all secret keys $\sk_i$ for $i\neq j$.
    \item \textbf{Adaptive hardcore bit at a single index}: For any index $j$, the function $f_{\pk_j}$ satisfies the rTCF adaptive hardcore bit property even given all secret keys $\sk_i$ for $i\neq j$.
\end{itemize}

Our protocol is a variant of the Mahadev protocol~\cite{FOCS:Mahadev18a} in which the verifier's first message $(\pk_1, \hdots, \pk_N)$ is replaced by the output $\PK$ of a batch key generation procedure.

\subsection{Measurement Protocol Description}\label{subsec:meas-protocol}

Let $\mathsf{ClawFree} = (\mathsf{TCF.}\Gen, \Eval, \Invert, \Check, \Good)$ denote a family of \cite{FOCS:Mahadev18a} randomized TCFs (\cref{def:clawfree}). For simplicity of the analysis, we assume that on a fixed security parameter~$1^\secp$, $\Gen$ outputs keys defining functions from $\{0,1\} \times D$ to $R$, where $D$ is a fixed-size subset of $\{0,1\}^{\ell(\secp)}$ independent of $\pk$.\footnote{The LWE-based instantiation from \cite{FOCS:BCMVV18} satisfies this property. However, with more complicated expressions, our proof should go through without this simplifying assumption.} We denote the size of $D$ by $L(\secp)$. Let $(\Gen, \extpk, \extsk)$ denote a batch key generation algorithm for $\mathsf{TCF}.\Gen$ satisfying the requirements described above.

For simplicity, we assume the honest prover $P$ has an $N$-qubit pure state of the form
\[\ket{\psi}_{\RegW} = \sum_{w_1, \hdots, w_N}\alpha_{w_1, \hdots, w_N} \ket{w_1 \hdots w_N}
\]
on its internal $N$-qubit register $\RegW$; the protocol completeness immediately extends to mixed states by linearity. The verifier has a basis string $h\in\{0,1\}^N$, represented by a circuit $C: [N] \rightarrow \{0,1\}$, and wants to obtain a measurement of $\ket{\psi}$ in this basis (where $h_i = 0$ denotes a standard basis measurement and $h_i = 1$ denotes a Hadamard basis measurement). The protocol $\MP$ (for measurement protocol), with prover algorithm $P = (\MP.\Commit, \MP.\Open)$ and verifier algorithm $V = (\MP.\Gen, \MP.\Test, \MP.\Out)$, is executed as follows.

\begin{itemize}
    \item $\MP.\Gen$: On input $1^\secp$, $N$ and $C$, the verifier samples a representation of $N$ TCF key pairs by computing $(\PK, \SK) \gets \Gen(1^\secp, N, C)$, where $C$ is circuit computing $C(i) = h_i$. The verifier sends $\PK$ to the prover.
    \item $\MP.\Commit$: Given $\PK$ and its input state $\ket{\psi}_{\RegW}$, the prover computes public keys $\pk_i \gets \extpk(\PK, i)$ for $1\leq i\leq N$. From now on, $P$ operates directly on $(\pk_i)$ and ignores $\PK$. 

The prover coherently (with respect to $\ket{\psi}$) computes a ``range superposition''
    \[\frac 1 {\sqrt{L^N}} \sum_{\substack{x_1, \hdots, x_N \in D \\ y_1, \hdots, y_N \in R \\ w \in \{0,1\}^N}} \left( \alpha_{w} \prod_i \sqrt{p_{\pk_i}(w_i, x_i, y_i)} \ket{w}_{\RegW} \ket{x_1}_{\RegX_1} \hdots \ket{x_N}_{\RegX_N} \ket{y_1}_{\RegY_1} \hdots \ket{y_N}_{\RegY_N}\right)
    \]
    where each $\RegX_i$ is an $\ell(\secp)$-qubit register (where $D \subset \{0,1\}^{\ell})$, and each $\RegY_i$ 
    
    has basis $\{\ket{y} \}_{y \in R}$. Here, $p_{\pk}(b, x, y)$ denotes the probability density of $y$ in the distribution $f_{\pk}(b, x)$, where $p_{\pk}(b,x,y) \coloneqq 0$ for $x \in \{0,1\}^{\ell} \setminus D$. Following~\cite[Section 4.3]{FOCS:BCMVV18} the honest prover algorithm can efficiently prepare this state up to exponentially small trace distance. 
    
    After preparing this state, the prover measures $\RegY_1, \hdots, \RegY_N$ in the standard ($R$-)basis and sends the outcome $(y_1,\dots,y_N)$ to the verifier.

    \item The verifier sends a uniformly random challenge bit $c$. After receiving the prover response, the verifier computes each public key $\pk_i \gets \extpk(\PK, i)$ and secret key $\sk_i \gets \extsk(\SK, i)$ in order to evaluate either $\MP.\Test$ or $\MP.\Out$. 
    \item $\MP.\Open$: On challenge bit $c$, the prover operates as follows.
    \begin{itemize}
        \item If $c=0$, the prover measures $\RegW \tensor \RegX_1 \tensor \hdots \tensor \RegX_N$ in the standard basis and sends the outcome $(b_1, \hdots, b_N, x_1, \hdots, x_N)$ to the verifier. 
        \item If $c=1$, the prover instead measures $\RegW \tensor \RegX_1\tensor \hdots \tensor \RegX_N$ in the Hadamard basis, returning strings $d_1, \hdots d_N$.
    \end{itemize} 
    \item $\MP.\Test$: Given $(b_1, \hdots, b_N, x_1, \hdots, x_N)$, the verifier computes (for every $i$) $\mathsf{Check}(\pk_i, b_i, x_i, y_i)$ and rejects if any of these checks do not pass.
    \item $\MP.\Out$: Given $d_1, \hdots d_N$, the verifier outputs $N$ bits as follows. For each $i \in [N]$: 
    \begin{itemize}
        \item If $h_i = 0$, the verifier ignores $d_i$, computes $(b_i, x_i) = \Invert(\inj,\sk_i, y_i)$, and outputs $b_i$.
        \item If $h_i=1$, the verifier computes the two inverses $\{(0, x_{0,i}), (1, x_{1,i})\} \gets \Invert(\twotoone,\sk_i, y_i)$. For each $i$, the verifier checks whether $d_i\in \mathsf{Good}_{x_{0,i}, x_{1,i}}$ (corresponding to a valid equation in the $i$th slot), and if so, the verifier outputs $d_i \cdot (1, x_{0,i}\oplus x_{1,i})$. If $d_i \not\in \mathsf{Good}_{x_{0,i}, x_{1,i}}$, the verifier samples a uniformly random bit and outputs it.
    \end{itemize}  
\end{itemize}

Completeness of this protocol follows immediately from \cite{FOCS:Mahadev18a} and the correctness property of $\mathsf{Gen}$. Specifically, the correctness property of $\Gen$ implies that each $(\pk_i, \sk_i)$ in our protocol is in the range of $\TCF.\Gen(1^\secp, C(i))$, in which case (as shown in \cite{FOCS:Mahadev18a}) the verifier's output distribution is statistically close to $h = (C(0), \hdots, C(N))$-measurement outcome on $\ket{\psi}$. 

\def\semi{\mathsf{SemiSuccinct}}

\section{Soundness of Mahadev's Protocol}
\label{sec:succinct-measurement}

In this section, we prove that the measurement protocol from \cref{sec:meas-protocol} a computationally sound (\cref{def:measurement-protocol-soundness}) commit-and-measure protocol. As a consequence, we obtain a new, self-contained proof of soundness of the \cite{FOCS:Mahadev18a} protocol. Later (\cref{sec:succinct-keygen-construction}), we will instead instantiate our protocol with a \emph{succinct} key generation algorithm to obtain a verifier-succinct measurement protocol.

Our soundness proof is based in part on both \cite{FOCS:Mahadev18a} itself as well as a proof strategy suggested in \cite{Vid20-course}. 

\paragraph{Notation.} Throughout this section, we will fix the verifier's choice of basis $h \in \{0,1\}^N$. We write $R \coloneqq \{i \in [N] : h_i = 1\}$ and $S \coloneqq \{i \in [N] : h_i = 0\}$, where $R \subset [N]$ denotes the set of indices that the verifier wants to measure in the Hadamard basis, and $S \subset [N]$ denotes the set of indices the verifier wants to measure in the standard basis.


Finally, we will decompose the state space of the prover as $\RegZ \otimes \RegI$, where:
\begin{itemize}
    \item $\RegZ = \RegZ_1 \otimes \cdots \otimes \RegZ_N$. $1\leq i\leq N$, $\RegZ_i$ is an $(\ell(\secp)+1)$-qubit register that contains the classical opening string $z_i$. We will sometimes write $\RegZ$ as shorthand for $\RegZ_1 \otimes \hdots  \otimes\RegZ_N$.
    \item Each $\RegZ_i$ can be written as $\RegZ_i = \RegB_i \tensor \RegX_i$, where $\RegB_i$ is a one-bit register and $\RegX_i$ is an $\ell(\secp)$-bit register.
    \item $\RegI$ denotes any additional registers the prover uses.
\end{itemize}

\subsection{The Verifier's Output Distribution}

Our goal is to characterize the $N$-bit distribution $D_{P^*,\mathrm{Out}}$ corresponding to the verifier's output in the measurement protocol when interacting with a malicious prover $P^*$ using $h \in \{0,1\}^N$ as its choice of bases. In particular, we want to prove that if $P^*$ succeeds in passing the test round with probability $1-\negl(\secp)$, then $D_{P^*,\mathrm{Out}}$ is computationally indistinguishable from $D_{P^*,\mathrm{Ext}}$, a distribution obtained from (1) running an efficient extractor $\Ext^{P^*}$ to obtain an $N$-qubit quantum state $\btau$, and (2) measuring $\btau$ in the verifier's specified bases.

The distribution $D_{P^*,\mathrm{Out}}$ produces a sample according to the following steps:
\begin{enumerate}
    \item Sample keys $(\PK,\SK) \gets \MP.\Gen(1^\secp,N,C)$ (where $C$ specifies the choice of bases $R,S$).
    \item Run the malicious prover on $\PK$ to obtain a classical commitment string $y$. Let $\ket{\psi} \in \RegZ \otimes \RegI$ denote the prover's residual state. 
    
    \item For each $i \in S$, compute $(b_i,x_i) \gets \Invert(\inj,\sk_i,y_i)$. Let $v \in \{0,1\}^S$ be the vector whose $i$th entry is $b_i$.
    \item Next, apply the prover's attack unitary $U$ on $\RegZ \otimes \RegI$, and then measure $\RegZ$ in the Hadamard basis to obtain a response $z = (d_1,\dots,d_N)$. 
    \item For each $i \in R$, compute $(0,x_{0,i}),(1,x_{1,i}) \gets \Invert(\twotoone,\sk_i,y_i)$. If $d_i \in \Good_{x_{0,i},x_{1,i}}$, set $u_i = d_i \cdot (1,x_{0,i} \oplus x_{1,i})$. Otherwise, set $u_i$ to be a uniformly random bit. This results in a string $u \in \{0,1\}^R$
    \item Output $(u,v) \in \{0,1\}^R \times \{0,1\}^S$.
\end{enumerate}

Our first step is to show that $D_{P^*, \mathsf{Out}}$ is computationally indistinguishable from a distribution $D_{P^*,\twotoone}$ that does \emph{not} require running the $\Invert$ algorithm for any key pair $(\pk_i, \sk_i)$ in injective mode. Instead, this second distribution can be produced by directly measuring the register $\RegB_i$ (i.e., the first bit of $\RegZ_i$) of the prover's state $\ket{\psi}$). Since $\{\sk_i\}_{i\in S}$ will no longer be required at this point, we are also able to switch all key pairs $(\pk_i, \sk_i)$ to be in two-to-one mode by invoking key indistinguishability.


Formally, $D_{P^*,\twotoone}$ produces outcomes as follows (differences from $D_{P^*,\mathrm{Out}}$ highlighted in \textcolor{red}{red}):
\begin{enumerate}
    \item Sample keys $(\PK,\SK) \gets \MP.\Gen(1^\secp,N,\mathbf{1})$ \textcolor{red}{(where $\mathbf{1}$ denotes the constant $1$ function, corresponding to two-to-one mode)}
    \item Run the malicious prover on $\PK$ to obtain a classical commitment string $y$. Let $\ket{\psi} \in \RegZ \otimes \RegI$ denote the prover's residual state. 
    \item \textcolor{red}{For each $i \in S$, measure $\RegB_i$ to obtain a bit $v_i$; the result of this step is a string $v \in \{0,1\}^S$.}
    \item Next, apply the prover's attack unitary $U$ on $\RegZ \otimes \RegI$, and then measure $\RegZ$ in the Hadamard basis to obtain a response $z = (d_1,\dots,d_N)$. 
    \item For each $i \in R$, compute $(0,x_{0,i}),(1,x_{1,i}) \gets \Invert(\twotoone,\sk_i,y_i)$. If $d_i \in \Good_{x_{0,i},x_{1,i}}$, set $u_i = d_i \cdot (1,x_{0,i} \oplus x_{1,i})$. Otherwise, set $u_i$ to be a uniformly random bit. This results in a string $u \in \{0,1\}^R$
    \item Output $(u,v) \in \{0,1\}^R \times \{0,1\}^S$.
\end{enumerate}

\begin{lemma}\label{lemma:observables-to-output}
$D_{P^*,\twotoone}$ is computationally indistinguishable from $D_{P^*,\mathrm{Out}}$.
\end{lemma}

We prove \cref{lemma:observables-to-output} by first switching the keys sampled in $D_{P^*,\twotoone}$ to match the verifier's basis choice $h$. That is, we define the distribution $D_{P^*,h}$ to be the same distribution as $D_{P^*,\twotoone}$, except that the keys $(\pk_i, \sk_i)$ are sampled in mode $h_i$, i.e., Step 1 is replaced with:
\begin{enumerate}
    \item Sample keys $(\PK,\SK) \gets \MP.\Gen(1^\secp,N,C)$ \textcolor{red}{(where $C(i) = h_i$ for all $i$)}.
\end{enumerate}
This is well-defined because the $i$th bit of the output is still obtained by measuring $\RegB_i$, which can be done regardless of how $(\pk_i,\sk_i)$ is sampled.     
  \begin{claim}\label{lemma:extractor-key-indistinguishability}
     For every basis choice $h$, $D_{P^*,h}$ is computationally indistinguishable from $D_{P^*,\twotoone}$.
  \end{claim}
  
  \begin{proof}
     This follows by invoking the following key indistinguishability property of $\Gen$: 
     \[\left\{(\PK, \SK)\gets \Gen(1^\secp, 1): (\PK, \{\sk_i\}_{i\not\in S})\right\} \approx_c \left\{ (\PK, \SK)\gets \Gen(1^\secp, C): (\PK, \{\sk_i\}_{i\not\in S})\right\}.
     \]
     Since the distributions are sampled without use of $\sk_i$ for all $i\in S$, \cref{lemma:extractor-key-indistinguishability} follows from this indistinguishability. 
  \end{proof}
  
To conclude that $D_{P^*,\twotoone} \approx_c D_{P^*,\mathrm{Out}}$, we note:

\begin{claim}\label{lemma:verifier-output}
If $D_{P^*, \mathsf{Out}}$ is instantiated with basis choice $h$, then $D_{P^*, \mathsf{Out}}$ is statistically indistinguishable from $D_{P^*,h}$.
\end{claim}

\begin{proof}
\cref{lemma:verifier-output} follows from the injectivity of $f_{\pk_i}$ for each $i\in S$; by the correctness of $\Gen$, we have that each $\pk_i$ (for $i\in S$) is in the support of $\TCF.\Gen(1^\secp, \inj)$. Therefore, since $\ket{\psi}$ is guaranteed to pass the test round with probability $1-\negl(\secp)$, we have that with probability $1-\negl(\secp)$, measuring $\RegB_i$ gives the same result as computing the first bit of $\Invert(\sk_j, y_j)$ (which is the verifier's output). 
\end{proof}

\subsection{The Protocol Observables}
\label{subsec:XZ-definition}

\paragraph{Defining the Protocol Observables.}

In $D_{P^*,\twotoone}$, the entire $N$-bit output $(u,v)$ is the result of performing measurements on $\ket{\psi}$, the prover's residual state after it sends its commitment $y$.

We now define a collection of binary observables $\{X_i,Z_i\}_{i \in [N]}$, parameterized by $(\PK,\SK,y)$ and the malicious prover's attack unitary $U$, such that the following process is equivalent to sampling from $D_{P^*,\twotoone}$:

\begin{enumerate}
    \item Sample keys $(\PK,\SK) \gets \MP.\Gen(1^\secp,N,\mathbf{1})$.
    \item Run the malicious prover on $\PK$ to obtain a classical commitment string $y$. Let $\ket{\psi}$ denote the prover's residual state. 
    \item For each $i \in S$, measure $\ket{\psi}$ with the observable $Z_i$ to obtain a bit $v_i$.
    \item Next, for each $i \in R$, measure the observable $X_i$ to obtain a bit $u_i$.
    \item Output $(u,v) \in \{0,1\}^R \times \{0,1\}^S$.
\end{enumerate}

The definition of the $Z_i$ observable is straightforward: since each $v_i$ is obtained by measuring $\RegB_i$ in the standard basis, $Z_i$ is simply the Pauli-$Z$ observable $Z_i \coloneqq (\sigma_Z)_{\RegB_i}$.

Defining the $X_i$ observable requires more care. In $D_{P^*,\twotoone}$, the string $u \in \{0,1\}^R$ is obtained by applying the following steps (after $v \in \{0,1\}^S$ is measured)
\begin{enumerate}
    \item Apply the prover's attack unitary $U$ on $\RegZ \otimes \RegI$.
    \item For each $i \in R$:
    \begin{enumerate}
        \item Apply $H^{\otimes \ell+1}$ to the register $\RegZ_i$ containing the prover's response in the $i$th slot.
        \item Measure $\RegZ_i$ to obtain $d_i$. If $d_i\in \Good(x_{0,i}, x_{1,i})$, set $u_i = d_i \cdot (1, x_{0,i} \oplus x_{1,i})$. If $d_i\not\in \Good(x_{0,i}, x_{1,i})$, set $u_i$ to be a uniformly random bit.
    \end{enumerate}
\end{enumerate}

In order to output a uniformly random bit, we will prepare fresh one-qubit ancilla registers $\RegU_1,\dots,\RegU_N$, so that in the event that the prover returns an invalid $d_i$ in slot $i$, the verifier can generate a random bit by measuring $\RegU_i$ (initialized to $\ket{0}$) in the Hadamard basis. Note that the $\RegU = \RegU_1,\dots,\RegU_N$ register is not part of the malicious prover's state.

We therefore redefine $\ket{\psi} \coloneqq \ket{\psi}_{\RegZ,\RegI} \ket{0}_{\RegU}$ to denote the global state on $\RegZ \otimes \RegI \otimes \RegU$ including the ancilla $\RegU$ registers initialized to $\ket{0}_{\RegU}$.

Finally, the $X_i$ observable is defined as 

\[ X_i =  (U \tensor \Id_{\RegU})^\dagger (H^{\otimes \ell+1}_{\RegZ_i} \tensor \Id \tensor H_{\RegU_i}) X'_i (H^{\otimes \ell+1}_{\RegZ_i}\tensor \Id \tensor H_{\RegU_i}) (U \tensor \Id_{\RegU}) . 
\]
where
\begin{align*}
    X'_i = &\sum_{d \in \Good(x_{0,i}, x_{1,i})} (-1)^{ d \cdot (1, x_{0,i} \oplus x_{1,i}) } \ketbra{d}_{\RegZ_i} \tensor \Id_{\RegI, \{\RegZ_j\}_{j\neq i},\RegU} \\
    &+ \sum_{d \not\in \Good(x_{0,i}, x_{1,i}), u\in \{0,1\}} (-1)^{u} \ketbra{d,u}_{\RegZ_i,\RegU_i} \tensor \Id_{\RegI, \{\RegZ_j\}_{j\neq i},\{\RegU_j\}_{j\neq i}}.
\end{align*}

Note that the $X_i$ observables are defined so that each pair of $X_i,X_j$ commute. Moreover, one can verify that measuring $X_i$ for each $i \in R$ exactly corresponds to measuring $u \in \{0,1\}^R$ as described above.

The description of $X_i$ depends on $(y_i, \sk_i)$ because of the appearance of $x_{i,0}, x_{i,1}$ in $X_i'$. Moreover, note that each $X_i$ is efficiently computable given $\sk_i$.

For convenience, we define a procedure $\{X_i,Z_i\}_i,\ket{\psi}_{\RegZ,\RegI,\RegU} \gets \mathsf{Samp}$ that works as follows:
\begin{itemize}
    \item Sample keys $(\PK,\SK) \gets \MP.\Gen(1^\secp,N,\mathbf{1})$.
    \item Run the malicious prover on $\PK$ to obtain a classical commitment string $y$. Let $\ket{\psi'}$ denote the prover's residual state on $\RegZ \otimes \RegI$.
    \item Output the observables $\{X_i,Z_i\}$ parameterized by $(\PK,\SK,y)$ and malicious prover's unitary $U$, along with the state $\ket{\psi} \coloneqq \ket{\psi'} \otimes \ket{0}_{\RegU}$. 
\end{itemize}

For the remainder of this section, we will write $D_{P^*,\twotoone}$ as a two-step sampling process:
\begin{enumerate}
    \item Run $\{X_i,Z_i\}_i,\ket{\psi}_{\RegZ,\RegI,\RegU} \gets \mathsf{Samp}$.
    \item Starting with $\ket{\psi}$, measure each $Z_i$ for $i \in S$ to obtain $v \in \{0,1\}^S$. Then measure each $X_i$ for $i \in R$ to obtain $u \in \{0,1\}^R$. Output $(u,v) \in \{0,1\}^R \times \{0,1\}^S$.
\end{enumerate}

\subsection{The Extracted State}\label{subsec:extracted-state}

Recall that our definition of measurement protocol soundness (\cref{def:measurement-protocol-soundness}) requires us to give an extractor that:
\begin{enumerate}
    \item Generates keys $(\PK,\SK)$ according to an algorithm $\SimGen(1^\secp)$ (independently of the verifier's basis choice $h$).
    \item Runs the malicious prover $P^*$ on $\PK$ to obtain $y$; as usual, $\ket{\psi} \in \RegZ \otimes \RegI \otimes \RegU$ denotes the residual prover state with $\RegU$ initialized to $\ket{0}_{\RegU}$.\footnote{To match the syntax of our definition in~\cref{def:measurement-protocol-soundness}, the register $\RegU$ should be viewed as an internal register initialized by the extractor.} 
    \item Generates an extracted state $\btau \gets \Ext^{U,\ket{\psi}}(\PK,\SK,y)$ (the superscript denotes black-box access to a unitary $U$ and state $\ket{\psi}$, see \cref{subsec:black-box}).
\end{enumerate}

We define $\SimGen(1^\secp)$ to be $\MP.\Gen(1^\lambda,N,\mathbf{1})$, which exactly corresponds to how keys are sampled in $D_{P^*,\twotoone}$.

To establish soundness, it remains to (1) describe how to generate the extracted state $\btau$ given $(\PK,\SK,y),U$, and (2) prove that the distribution that arises from measuring $\btau$ with the Pauli-$X$ and Pauli-$Z$ observables in the verifier's chosen bases $h$ is computationally indistinguishable from $D_{P^*,\twotoone}$.

We handle (1) in~\cref{def:extracted-state}. We then describe the distribution $D_{P^*,\mathrm{Ext}}$ that arises from measuring our extracted state in~\cref{subsubsec:measuring-extracted} and prove that indistinguishability from $D_{P^*,\twotoone}$ in~\cref{subsec:measurement-indistinguishability}.

\subsubsection{A Teleportation-Inspired Extraction Procedure}
\label{def:extracted-state}

Fix a choice of $\{X_i,Z_i\},\ket{\psi} \gets \mathsf{Samp}$. For ease of notation, write $\RegH = \RegZ \otimes \RegI \otimes \RegU$ so that $\ket{\psi} \in \RegH$. We would like an efficient extraction procedure that takes as input $\ket{\psi} \in \RegH$ and generates an $N$-qubit state $\btau$ such that, roughly speaking, measuring $\ket{\psi}$ with $X/Z$ and measuring $\btau$ with $\sigma_X/\sigma_Z$ produce indistinguishable outcomes.

\paragraph{Intuition for the Extractor.} Before we describe our extractor, we first provide some underlying intuition. For an arbitrary $N$-qubit Hilbert space, let $\sigma_{x,i}$/$\sigma_{z,i}$ denote the  Pauli $\sigma_x$/$\sigma_z$ observable acting on the $i$th qubit. For each $r,s \in \{0,1\}^N$, define the $N$-qubit Pauli ``parity'' observables
\[\sigma_x(r) \coloneqq \prod_{i: r_i = 1} \sigma_{x,i} \;,\; \sigma_z(s) \coloneqq \prod_{i: s_i = 1} \sigma_{z,i}.\]

Suppose for a moment that $\ket{\psi} \in \RegH$ is \emph{already} an $N$-qubit state (i.e., $\RegH$ is an $N$-qubit Hilbert space) and moreover, that each $X_i$/$Z_i$ observable is simply the corresponding Pauli observable $\sigma_{x,i}$/$\sigma_{z,i}$. While these assumptions technically trivialize the task (the state already has the form we want from the extracted state), it will be instructive to \textbf{write down an extractor that ``teleports'' this state into another $N$-qubit external register. }

We can do this by initializing two $N$-qubit registers $\RegA_1 \otimes \RegA_2$ to $\ket{\phi^+}^{\otimes N}$ where $\ket{\phi^+}$ is the EPR state $(\ket{00} + \ket{11})/\sqrt{2}$ (the $i$th EPR pair lives on the $i$th qubit of $\RegA_1$ and $\RegA_2$). Now consider the following steps, which are inspired by the ($N$-qubit) quantum teleportation protocol
\begin{enumerate}
    \item Initialize a $2N$-qubit ancilla $\RegW$ to $\ket{0^{2N}}$, and apply $H^{\otimes 2N}$ to obtain the uniform superposition.
    \item Apply a ``controlled-Pauli'' unitary, which does the following for all $r,s \in \{0,1\}^N$ and all $\ket{\phi} \in \RegH \otimes \RegA_1$:
    \[\ket{r,s}_{\RegW} \ket{\phi}_{\RegH,\RegA_1} \rightarrow \ket{r,s}_{\RegW} (\sigma_x(r)\sigma_z(s)_{\RegH} \otimes \sigma_x(r)\sigma_z(s)_{\RegA_1})\ket{\phi}_{\RegH,\RegA_1}\] 
    \item Apply the unitary that XORs onto $\RegW$ the outcome of performing $N$ Bell-basis measurements\footnote{The Bell basis consists of the $4$ states $(\sigma_x^a \sigma_z^b \otimes \Id)\ket{\phi^+}$ for $a,b \in \{0,1\}$ on $2$ qubits.} on $\RegA_1 \otimes \RegA_2$ onto $\RegW$, i.e., for all $u,v,r,s \in \{0,1\}^N$: \[\ket{u,v}_{\RegW}(\sigma_x(r)\sigma_z(s) \otimes \Id)_{\RegA_1,\RegA_2}\ket{\phi^+}^{\otimes N}_{\RegA_1,\RegA_2} \rightarrow \ket{u \oplus r ,v \oplus s}_{\RegW}(\sigma_x(r)\sigma_z(s) \otimes \Id)_{\RegA_1,\RegA_2}\ket{\phi^+}^{\otimes N}_{\RegA_1,\RegA_2}.\]
    Finally, discard $\RegW$.
\end{enumerate}
One can show that the resulting state is
\begin{equation}
\label{eqn:trivial-extracted-state}
\frac{1}{2^N} \sum_{r,s \in \{0,1\}^N} (\sigma_x(r)\sigma_z(s) \otimes \sigma_x(r) \sigma_z(s) \otimes \Id)\ket{\psi}_{\RegH} \ket{\phi^+}_{\RegA_1,\RegA_2}=\ket{\phi^+}_{\RegH,\RegA_1} \ket{\psi}_{\RegA_2},
\end{equation}
where $\ket{\psi}$ is now ``teleported'' into the $\RegA_2$ register. 

To generalize this idea to the setting where $\ket{\psi} \in \RegH$ is an arbitrary quantum state and $\{X_i,Z_i\}_i$ are an arbitrary collection of $2N$ observables, \emph{we simply replace each $\sigma_x(r)$ and $\sigma_z(s)$ acting on $\RegH$ above with the corresponding parity observables for $\{X_i,Z_i\}$}. That is for each $r,s \in \{0,1\}^N$, define
\[ Z(s) = \prod_{i=1}^N Z_i^{s_i}
\hspace{.1in} \mbox{and}
\hspace{.1in} X(r) = \prod_{i=1}^N X_i^{r_i}. 
\]
The rough intuition is that as long as the $\{X_i\}$ and $\{Z_i\}$ observables ``behave like'' Pauli observables with respect to $\ket{\psi}$, the resulting procedure will ``teleport'' $\ket{\psi}$ into the $N$-qubit register $\RegA_2$.

\paragraph{The Full Extractor.} In more detail, we have the state $\ket{\psi}_\RegH = \ket{\psi}_{\RegZ,\RegI,\RegU}$, and we initialize two $N$-qubit registers $\RegA_1 \otimes \RegA_2$ to $\ket{\phi}^{\otimes N}$. We run the following steps (the changes from the above procedure are highlighted in \textcolor{red}{red}):
\begin{enumerate}
    \item Initialize a $2N$-qubit ancilla $\RegW$ to $\ket{0^{2N}}$, and apply $H^{\otimes {2N}}$.
    \item Apply a unitary that does the following for all $r,s \in \{0,1\}^N$:
    \[\ket{r,s}_{\RegW} \ket{\phi}_{\RegH,\RegA_1} \rightarrow \ket{r,s}_{\RegW} (\textcolor{red}{X(r) Z(s)_{\RegH}} \otimes \sigma_x(r)\sigma_z(s)_{\RegA_1})\ket{\phi}_{\RegH,\RegA_1}\]
    \item Apply the unitary that XORs onto $\RegW$ the outcome of performing $N$ Bell-basis measurements on $\RegA_1 \otimes \RegA_2$ onto $\RegW$, i.e., for all $u,v,r,s \in \{0,1\}^N$: \[\ket{u,v}_{\RegW}(\sigma_x(r)\sigma_z(s) \otimes \Id)\ket{\phi^+}^{\otimes N}_{\RegA_1,\RegA_2} \rightarrow \ket{u \oplus r ,v \oplus s}_{\RegW}(\sigma_x(r)\sigma_z(s) \otimes \Id)\ket{\phi^+}^{\otimes N}_{\RegA_1,\RegA_2}.\]
    Finally, discard $\RegW$.
\end{enumerate}

All of these steps can be efficiently implemented given black-box access to $\{X_i,Z_i\}_i$. The resulting state is
\begin{align*} 
\frac{1}{2^{N}}\sum_{r,s \in \{0,1\}^{N}} X(r)Z(s)\ket{\psi}_{\RegH} \otimes \sigma_x(r)\sigma_z(s)\ket{\phi^+}^{\otimes N}_{\RegA_1,\RegA_2},
\end{align*}
and we define the extracted state $\btau \coloneqq \mathsf{Ext}_{\{X_i\},\{Z_i\}}(\ket{\psi})$ to be the residual state on $\RegA_2$ after tracing out $\RegH$ and $\RegA_1$.\footnote{The same extracted state is defined in Vidick's lecture notes~\cite{Vid20-course}, although the notes do not give an explicit procedure for generating it.}

\subsubsection{Measuring the Extracted State}
\label{subsubsec:measuring-extracted}
We now consider the $N$-bit distribution of measurement outcomes that arise from measuring the extracted state $\btau$ using the Pauli observables $\sigma_x,\sigma_z$. In particular, we consider performing the measurements according to the verifier's basis choice, so that we measure $\sigma_{z,i}$ for each $i \in S$ and $\sigma_{x,i}$ for each $i \in R$. 

Formally, we define the distribution $D_{P^*,\mathrm{Ext}}$ on $\{0,1\}^N$ obtained by the following process:
\begin{itemize}
    \item Run $\{X_i,Z_i\},\ket{\psi} \gets \mathsf{Samp}$.
    \item Let $\btau = \mathsf{Ext}_{\{X_i\},\{Z_i\}}(\ket{\psi})$ be the $N$-qubit extracted state.
    \item Measure the Pauli-$Z$ observable $\sigma_{z,i}$ for all $i \in S$, obtaining $v \in \{0,1\}^S$.
    \item Measure the Pauli-$X$ observable $\sigma_{x,i}$ for all $i \in R$, obtaining $u \in \{0,1\}^R$.
    \item Output $(u,v) \in \{0,1\}^R \times \{0,1\}^S$.
\end{itemize}

It will be convenient to define the following projection operators. For each $u \in \{0,1\}^R$ and $v \in \{0,1\}^S$ let
 \begin{equation}\label{eq:f-i}
\Pi^{\sigma_x}_{u} = \E_{u'\in \{0,1\}^R} (-1)^{u\cdot u'}\sigma_x(u')\qquad\text{and}\qquad  \Pi^{\sigma_z}_{v} = \E_{v'\in \{0,1\}^S} (-1)^{v\cdot v'}\sigma_z(v')
\end{equation}
In words, $\Pi^{\sigma_x}_{u}$ is the projection that corresponds to measuring $\sigma_{x,i}$ for each $i \in R$ and obtaining the string of outcomes $u \in \{0,1\}^R$, and $\Pi^{\sigma_z}_{v}$ is the projection that corresponds to measuring $\sigma_{z,i}$ for each $i \in S$ and obtaining the string of outcomes $v \in \{0,1\}^S$.

Then the probability $D_{P^*,\mathrm{Ext}}$ outputs any $(u,v) \in \{0,1\}^R \times \{0,1\}^S$ can be written as
\[ D_{P^*,\mathrm{Ext}}(u,v) = \E_{\{X_i,Z_i\},\ket{\psi} \gets \mathsf{Samp}}[\Tr\big(\Pi^{\sigma_x}_{u} \Pi^{\sigma_z}_{v} \btau\big) : \btau = \mathsf{Ext}_{\{X_i\},\{Z_i\}}(\ket{\psi})].\]

We define a set of analogous projection operators for the $\{X_i\}$ and $\{Z_i\}$ observables. For each $u \in \{0,1\}^R$ and $v \in \{0,1\}^S$, let
\begin{equation}\label{eq:XZ-proj}
\Pi^{X}_{u} = \E_{u'\in \{0,1\}^R} (-1)^{u\cdot u'}X(u')\qquad\text{and}\qquad  \Pi^{Z}_{v} = \E_{v'\in \{0,1\}^S} (-1)^{v\cdot v'}Z(v')
\end{equation}
In words, $\Pi^{X}_{u}$ is the projection that corresponds to measuring $X_i$ for each $i \in R$ and obtaining the string of outcomes $u \in \{0,1\}^R$, and $\Pi^{Z}_{v}$ is the projection that corresponds to measuring $Z_i$ for each $i \in S$ and obtaining the string of outcomes $v \in \{0,1\}^S$.

With these definitions in mind, we state a claim that allows us to characterize the result of measuring the extracted state $\btau$ with the Pauli observables.

\begin{claim}\label{claim:extracted-state-distribution}
Fix any choice of $\{X_i,Z_i\}_{i \in [N]}$ and state $\ket{\psi}$, and let $\btau = \mathsf{Ext}_{\{X_i\},\{Z_i\}}(\ket{\psi})$. For all $(u,v)\in \{0,1\}^R\times \{0,1\}^S$ it holds that 
 \begin{equation}
\Tr\big(\Pi^{\sigma_x}_{u} \Pi^{\sigma_z}_{v} \btau\big) = \E_{u' \in \{0,1\}^{R}}  \bra{\psi} \Pi^{Z}_{v} Z(u') \Pi^{X}_{u' \oplus u} Z(u') \Pi^{Z}_{v} \ket{\psi}.
  \end{equation}
\end{claim}

The proof of~\cref{claim:extracted-state-distribution} is a straightforward (but slightly tedious) computation and is deferred to~\cref{sec:extracted-dist}.

Importantly, \cref{claim:extracted-state-distribution} gives a clear understanding of how $D_{P^*,\twotoone}$ and $D_{P^*,\mathrm{Ext}}$ relate to each other, since it allows us to view the distribution $D_{P^*,\mathrm{Ext}}$ (which arises from Pauli measurements on the extracted state $\btau$) as the result of performing certain protocol observable measurements $\{X_i,Z_i\}$ on $\ket{\psi}$.

Recall that the distribution $D_{P^*,\twotoone}$ is the following distribution:
\begin{enumerate}
    \item Run $\{X_i,Z_i\}_i,\ket{\psi}_{\RegZ,\RegI,\RegU} \gets \mathsf{Samp}$.
    \item Starting with $\ket{\psi}$, measure each $Z_i$ for $i \in S$ to obtain $v \in \{0,1\}^S$. Then measure each $X_i$ for $i \in R$ to obtain $u \in \{0,1\}^R$. Output $(u,v) \in \{0,1\}^R \times \{0,1\}^S$.
\end{enumerate}

By~\cref{claim:extracted-state-distribution}, we can write $D_{P^*,\mathrm{Ext}}$ as follows (differences from $D_{P^*,\twotoone}$ are in \textcolor{red}{red}): 
\begin{enumerate}
    \item Run $\{X_i,Z_i\},\ket{\psi} \gets \mathsf{Samp}$.
    \item Starting with $\ket{\psi}$ measure each $Z_i$ for $i \in S$ to obtain $v \in \{0,1\}^S$. \textcolor{red}{Then sample a uniformly random string $u' \gets \{0,1\}^R$ and apply the unitary $Z(u')$.} Finally, measure each $X_i$ for $i \in R$ \textcolor{red}{and XOR the output with $u'$} to obtain $u \in \{0,1\}^R$. Output $(u,v) \in \{0,1\}^R \times \{0,1\}^S$. 
\end{enumerate}

With this key difference in mind, it remains to prove indistinguishability of these two distributions.

\subsection{Indistinguishability of Measurement Outcomes} 
\label{subsec:measurement-indistinguishability}

In this subsection, we complete the proof that $D_{P^*,\twotoone}$ and $D_{P^*,\mathrm{Ext}}$ are computationally indistinguishable. We first write out their probability mass functions:
\begin{itemize}
    \item $D_{P^*,\twotoone}$ outputs $(u,v) \in \{0,1\}^R \times \{0,1\}^S$ with probability
    \[ D_{P^*,\twotoone}(u,v) = \underset{\{X_i,Z_i\},\ket{\psi} \gets \mathsf{Samp}}{\mathbb E}\left[ \bra{\psi} \Pi^Z_v \Pi^X_u \Pi^Z_v \ket{\psi} \right].\]
    \item $D_{P^*,\mathrm{Ext}}$ outputs $(u,v) \in \{0,1\}^R \times \{0,1\}^S$ with probability
    \[ D_{P^*,\mathrm{Ext}}(u,v) = \E_{\substack{\{X_i,Z_i\},\ket{\psi} \gets \mathsf{Samp}\\ u' \in \{0,1\}^R}}\left[ \bra{\psi} \Pi^{Z}_{v} Z(u') \Pi^{X}_{u' \oplus u} Z(u') \Pi^{Z}_{v} \ket{\psi} \right].\]
\end{itemize}

At this point, the reader may find it helpful to convince themselves that probability mass functions above exactly correspond to the descriptions of these distributions given at the end of~\cref{subsec:extracted-state}. The equivalence between these two representations will be a key component of the upcoming proofs.

For convenience, we will reorder the indices so that the indices in $R$ are labeled $1,2,\dots,|R|$. Let $u_{\leq j} \in \{0,1\}^R$ be the vector equal to $u$ on the first $j$ indices, and is $0$ on the remaining indices. For each $j \in \{0,1,\dots,|R|\}$, define hybrid $\Hyb_{j}$ to be the distribution that outputs $(u,v) \in \{0,1\}^R \times \{0,1\}^S$ with probability
\[ \Hyb_{j}(u,v) = \E_{\substack{\{X_i,Z_i\},\ket{\psi} \gets \mathsf{Samp}\\ u' \in \{0,1\}^R}}\left[ \bra{\psi} \Pi^{Z}_{v}Z(u'_{\leq j}) \Pi^{X}_{u \oplus u'_{\leq j}} Z(u'_{\leq j}) \Pi^{Z}_{v} \ket{\psi}\right].\]
Additionally, for each $j \in \{1,\dots,|R|\}$, and $b \in \{0,1\}$ define hybrid $\Hyb_{j,b}$ to be the distribution that outputs $(u,v) \in \{0,1\}^R \times \{0,1\}^S$ with probability
\[ \Hyb_{j,b}(u,v) = \E_{\substack{\{X_i,Z_i\},\ket{\psi} \gets \mathsf{Samp}\\ u' \in \{0,1\}^R}}\left[\bra{\psi} \Pi^{Z}_{v}Z(u'_{\leq j-1}) Z_j^b \Pi^{X}_{u \oplus u'_{\leq j-1} \oplus b\cdot e_{j}} Z_j^b Z(u'_{\leq j-1}) \Pi^{Z}_{v} \ket{\psi}\right],\]
where $e_j \in \{0,1\}^R$ denotes the $j$th standard basis vector.

\begin{claim}
\label{claim:indistinguishable-u}
For all $j \in \{1,\dots,|R|\}$, the distributions $\Hyb_{j,0}$ and $\Hyb_{j,1}$ are computationally indistinguishable.
\end{claim}

Observe that for $j \in \{1,2,\dots,|R|\}$, $\Hyb_{j,0} = \Hyb_{j-1}$, and that $\Hyb_{j,0}$ is the uniform mixture of $\Hyb_{j,0}$ and $\Hyb_{j,1}$. Since $\Hyb_0 = D_{P^*,\twotoone}$ and $\Hyb_{|R|} = D_{P^*,\mathrm{Ext}}$,~\cref{claim:indistinguishable-u} implies that $D_{P^*,\twotoone}$ and $D_{P^*,\mathrm{Ext}}$ are computationally indistinguishable.

We now prove~\cref{claim:indistinguishable-u}, which will complete the proof of measurement protocol soundness. Our proof involves the following steps:
\begin{itemize}
    \item First, we prove~\cref{claim:collapsing}, which states that the marginal distributions of $\Hyb_{j,0}$ and $\Hyb_{j,1}$ on $N \setminus \{j\}$ are indistinguishable due to the collapsing property of $f_{\pk_{j}}$.
    \item We then state~\cref{claim:d1}, which (together with~\cref{claim:collapsing}) shows that if $\Hyb_{j,0}$ and $\Hyb_{j,1}$ are efficiently distinguishable, then they can be distinguished as follows:
    \begin{enumerate}
        \item Given a sample $x$ (from either $\Hyb_{j,0}$ or $\Hyb_{j,1}$) run an efficient algorithm $A$ on $x_{\setminus \{j\}}$ ($x$ without the $j$th bit). 
        \item If $A$ outputs $0$, guess a random bit $b$. If $A$ outputs $1$, guess $b = x_{j}$.
    \end{enumerate}
    Roughly speaking, this reduces the task to arguing about the indistinguishability of the single bit $x_{j}$ (conditioned on $A$ outputting $1$).
    \item Finally, we show that the $1$-bit conditional distributions must be indistinguishable by appealing to the adaptive hardcore bit property of $f_{\pk_{j}}$.
\end{itemize}

\begin{claim}
\label{claim:collapsing}
Let $R' = R \setminus \{j\}$ and let $(\Hyb_{j,0})_{[N] \backslash \{j\}}$ and $(\Hyb_{j,1})_{[N] \backslash \{j\}}$ be the marginal distributions of $\Hyb_{j,0}$ and $\Hyb_{j,1}$ on $[N]\backslash \{j\} = R' \cup S$. Then $(\Hyb_{j,0})_{[N] \backslash \{j\}}$ and $(\Hyb_{j,1})_{[N] \backslash \{j\}}$ are computationally indistinguishable.
\end{claim}

\begin{proof}
Any quantum algorithm for distinguishing $(\Hyb_{j,0})_{[N] \backslash \{j\}}$ and $(\Hyb_{j,1})_{[N] \backslash \{j\}}$ can be represented as an $N-1$ qubit binary POVM $(A,\Id-A)$, where the distinguisher outputs $1$ on $x$ with probability $\bra{x} A \ket{x}$. We show that this contradicts the collapsing property of $f_{\pk_{j}}$ (given $\PK, \{\sk_i\}_{i\neq j}$). 

Consider the following adversary for the $f_{\pk_j}$ collapsing security game:
\begin{itemize}
    \item Given $\PK, \{\sk_i\}_{i\neq j}$, the adversary runs the prover $P^*$ on $\PK$ to obtain $(y, \ket{\psi})$. Recall that $\ket{\psi}$ is guaranteed to contain a valid pre-image in register $\RegZ_{j}$. The adversary submits $y$ to the collapsing game challenger. 
    \item The challenger flips a random bit and either applies $Z_{j}$ or does nothing.\footnote{This version of the collapsing game is equivalent to the standard formulation in which the challenger either does/does not perform a measurement. This follows from the fact that measuring a qubit in the computational basis (and discarding the outcome) is \emph{equivalent} to applying $Z^b$ for a random $b \gets \{0,1\}$. Thus, the challenger's measurement (in the $b = 1$ experiment) is equivalent to applying $Z$ with probability $1/2$; for simplicity, our formulation has the challenger (in the $b =1$ experiment) apply $Z$ with probability $1$, which increases the adversary's distinguishing advantage by a factor of $2$.}
    \item Then the adversary performs the following steps:
    \begin{enumerate}
        \item Measure $Z_i$ for every $i \in S$ obtaining outcomes $v \in \{0,1\}^S$.
        \item Sample a random string $u' \gets \{0,1\}^R$ and apply the unitary $Z(u'_{\leq j-1})$.
        \item Measure $X_{i}$ for every $i \in R'$, and XOR the outcomes with $u'_{\leq j-1}$ to obtain an output string $u \in \{0,1\}^{R'}$.
        \item Finally, measure $\ket{u,v}$ with the POVM $\{A,\Id-A\}$, and output $1$ if and only if the measurement outcome is $A$.
    \end{enumerate}
\end{itemize}
All of the adversary's steps can be efficiently performed given $(\PK, \{\sk_i\}_{i\neq j})$. Moreover, the above adversary's advantage in the collapsing game is polynomially related to the advantage the POVM $(A,\Id-A)$ attains in distinguishing $(\Hyb_{j,0})_{[N] \backslash \{j\}}$ and $(\Hyb_{j,1})_{[N] \backslash \{j\}}$.
\end{proof}

Given that the marginal distributions of $\Hyb_{j,0}$ and $\Hyb_{j,1}$ on $[N] \backslash \{j\}$ are computationally indistinguishable (\cref{claim:collapsing}), we next invoke a general property of $N$-bit distributions implying that a distinguisher between $\Hyb_{j,0}$ and $\Hyb_{j,1}$ must be distinguishing some (efficiently computable) property of the $j$th bit of $\Hyb_{j,0}$ and $\Hyb_{j,1}$ \emph{conditioned} on an efficiently computable property of the $[N] \backslash \{j\}$-marginal distributions.

\begin{claim}\label{claim:d1}
Let $k=k(\lambda)$ be a positive integer-valued function of a security parameter $\lambda$. 
Let $\{D_{0,\lambda}\}_{\lambda\geq 1}$ and $\{D_{1,\lambda}\}_{\lambda \geq 1}$ be families of distributions on $\{0,1\}^{k+1}$ such that the marginal distributions $D_{0,\lambda}'$ and $D_{1,\lambda}'$ of $D_{0,\lambda}$ and $D_{1,\lambda}$ respectively on the first $k$ bits are computationally indistinguishable. Suppose that $D_{0,\lambda}$ and $D_{1,\lambda}$ are computationally \emph{distinguishable}. Then there is an efficiently computable binary-outcome POVM $\{M,\Id-M\}$ acting on $k$ qubits such that
\[\Big| \E_{x\sim D_{0,\lambda}} (-1)^{x_{k+1}} \bra{x_{\leq k}}M \ket{x_{\leq k}} 
- \E_{x\sim D_{1,\lambda}} (-1)^{x_{k+1}} \bra{x_{\leq k}}M\ket{x_{\leq k}} \Big| >  \frac{1}{\poly(\lambda)}.\]
\end{claim}

We defer the proof to~\cref{sec:distinguish-marginals}.

Finally, we show that the $j$th bit distinguisher of $\Hyb_{j,0}$ and $\Hyb_{j,1}$ discussed by \cref{claim:d1} cannot exist by the adaptive hardcore bit property of $f_{\pk_j}$ (given $\PK, \{\sk_i\}_{i\neq j}\}$).

\begin{claim}
\label{claim:hardcore-bit}
For any efficiently computable binary outcome POVM $\{M, \Id-M\}$,
\begin{align}
\label{eqn:adaptive-hc-bit}
    \Big| \E_{(u,v) \sim \Hyb_{j,0}} (-1)^{u_{j}} \bra{u_{\setminus \{j\}},v}M\ket{u_{\setminus \{j\}},v} - \E_{(u,v) \sim \Hyb_{j,1}} (-1)^{u_{j}} \bra{u_{\setminus \{j\}},v}M\ket{u_{\setminus \{j\}},v} \Big| = \negl(\lambda).
\end{align}
\end{claim}

\begin{proof}
For the reader's convenience, we write out the probability mass functions of $\Hyb_{j,0}$ and $\Hyb_{j,1}$ explicitly, with the differences highlighted in \textcolor{red}{red}
\begin{align*} \Hyb_{j,0}(u,v) &= \E_{\substack{\{X_i,Z_i\},\ket{\psi} \gets \mathsf{Samp}\\ u' \in \{0,1\}^R}}\left[ \bra{\psi} \Pi^{Z}_{v}Z(u'_{\leq j-1}) \Pi^{X}_{u \oplus u'_{\leq j-1}} Z(u'_{\leq j-1}) \Pi^{Z}_{v} \ket{\psi}\right]\\
\Hyb_{j,1}(u,v) &= \E_{\substack{\{X_i,Z_i\},\ket{\psi} \gets \mathsf{Samp}\\ u' \in \{0,1\}^R}}\left[\bra{\psi} \Pi^{Z}_{v}Z(u'_{\leq j-1}) \textcolor{red}{Z_j} \Pi^{X}_{u \oplus u'_{\leq j-1} \textcolor{red}{\oplus e_{j}}} \textcolor{red}{Z_j} Z(u'_{\leq j-1}) \Pi^{Z}_{v} \ket{\psi}\right].
\end{align*}

We define one more distribution (with the difference relative to $\Hyb_{j,0}$ highlighted in \textcolor{red}{red})
\begin{align*}
\overline{\Hyb}_{j,1}(u,v) &= \E_{\substack{\{X_i,Z_i\},\ket{\psi} \gets \mathsf{Samp}\\ u' \in \{0,1\}^R}}\left[\bra{\psi} \Pi^{Z}_{v}Z(u'_{\leq j-1}) \textcolor{red}{Z_j} \Pi^{X}_{u \oplus u'_{\leq j-1}} \textcolor{red}{Z_j} Z(u'_{\leq j-1}) \Pi^{Z}_{v} \ket{\psi}\right].
\end{align*}
We now rewrite the left-hand-side of~\cref{eqn:adaptive-hc-bit}, where in the second expectation we sample from $\overline{\Hyb}_{j,1}$ instead of $\Hyb_{j,1}$. Note that these distributions are identical except that $u_j$ is flipped, so we have
\begin{align*}
    \Big| & \E_{(u,v) \sim \Hyb_{j,0}} (-1)^{u_{j}} \bra{u_{\setminus \{j\}},v}M\ket{u_{\setminus \{j\}},v} - \E_{(u,v) \sim \Hyb_{j,1}} (-1)^{u_{j}} \bra{u_{\setminus \{j\}},v}M\ket{u_{\setminus \{j\}},v} \Big| \\
    &= \Big| \E_{(u,v) \sim \Hyb_{j,0}} (-1)^{u_{j}} \bra{u_{\setminus \{j\}},v}M\ket{u_{\setminus \{j\}},v} + \E_{(u,v) \sim \overline{\Hyb}_{j,1}} (-1)^{u_{j}} \bra{u_{\setminus \{j\}},v}M\ket{u_{\setminus \{j\}},v} \Big|.
\end{align*}
Dividing the right-hand-side by $2$ gives an expression equal to the (absolute value of) the expectation of the output in the following process:
\begin{itemize}
    \item Prepare $\{X_i,Z_i\},\ket{\psi} \gets \mathsf{Samp}$.
    \item Sample $b \gets \{0,1\}$ and prepare $Z_j^b \ket{\psi}$ (the $b = 0$ case corresponds to $\Hyb_{j,0}$ and the $b = 1$ case corresponds to $\overline{\Hyb}_{j,1}$).
    \item Then measure $Z_i$ for all $i \in S$ to obtain $v \in \{0,1\}^S$. Sample a random $u' \gets \{0,1\}^R$ and apply $Z(u'_{\leq j-1})$, and finally measure $X_i$ for all $i \in R$ and XOR the result with $u'_{\leq j-1}$ to obtain $u \in \{0,1\}^R$. 
    \item Prepare the state $ \ket{u_{\setminus \{j\}},v}$ and measure it with the POVM $\{M,\Id-M\}$. If the output is $\Id-M$, stop at this point and output $0$.
    \item Otherwise, if the output is $M$, output $(-1)^{u_j}$.
\end{itemize}
Notice that the second step is equivalent to measuring $\ket{\psi}$ with $Z_j$, since (writing $Z_j = Z_j^+ - Z_j^-$, where $Z_j^+$ is the projection onto the $1$ eigenstate of $Z_j$ and $Z_j^- = \Id - Z_j^+$ is the projection onto the $-1$ eigenstate of $Z_j$):
\[ \frac{1}{2}(Z_j \ketbra{\psi} Z_j + \ketbra{\psi}) = Z_j^+ \ketbra{\psi} Z_j^+ + Z_j^- \ketbra{\psi} Z_j^-.\]

It follows that
\begin{align*}
    \Big| \E_{(u,v) \sim \Hyb_{j,0}} (-1)^{u_{j}} \bra{u_{\setminus \{j\}},v}M\ket{u_{\setminus \{j\}},v} + \E_{(u,v) \sim \overline{\Hyb}_{j,1}} (-1)^{u_{j}} \bra{u_{\setminus \{j\}},v}M\ket{u_{\setminus \{j\}},v} \Big|/2
\end{align*}
is polynomially-related to the advantage of the following adversary for the adaptive hardcore bit game:
\begin{itemize}
    \item Given $\PK, \{\sk_i\}_{i\neq j}$, the adversary runs the prover $P^*$ on $\PK$ to obtain $(y, \ket{\psi})$. Recall that $\ket{\psi}$ is guaranteed to contain a valid pre-image in register $\RegZ_{j}$ with probability $1-\negl(\secp)$.
    \item The adversary measures the register $\RegZ_{j}$ of $\ket{\psi}$ in the standard basis, obtaining a string $(b_{j},x_{j})$. By the assumption that $\ket{\psi}$ contains valid pre-images and the fact that $\pk_j$ is in the range of $\mathsf{TCF}.\Gen(1^\secp, \twotoone)$, this is equivalent to measuring the observable $Z_{j}$ (which just measures $b_{j}$). 
    \item Next, the adversary measures $Z_{i}$ for all $i \in S$, obtaining a string of outcomes $v \in \{0,1\}^S$.
    \item Then the adversary samples random $u' \gets \{0,1\}^{R}$ and applies the unitary $Z(u'_{\leq j-1})$ to its state.
    \item The adversary measures $X_{i}$ for all $i \in R'$ and XORs the outcome with $u'_{\leq j-1}$, obtaining a string $u \in \{0,1\}^{R'}$. 
    \item The adversary prepares the state $\ket{u,v}$ and measures it with $\{M,\Id-M\}$. Depending on the outcome, it does the following:
    \begin{itemize}
        \item If the measurement outcome is $\Id-M$, it samples a uniformly random string $d_{j} \gets \{0,1\}^{\ell+1}$ and sends $(b_{j},x_{j},d_{j})$ to the challenger (in this case obtaining $\negl(\secp)$ advantage).
        \item If the measurement outcome is $M$, it applies $U$ to its state, followed by $H^{\otimes \ell+1}$ to $\RegZ_{j}$. It then measures $\RegZ_{j}$ to obtain a string $d_{j} \in \{0,1\}^{\ell+1}$ and sends $(b_{j},x_{j},d_{j})$ to the challenger. Note that the challenger's output bit (i.e., whether the adversary wins or loses) exactly corresponds to the bit $u_j$.
    \end{itemize}
\end{itemize}
By assumption, this adversary outputs a valid pre-image $(b_{j},x_{j})$ with probability $1-\negl(\lambda)$. Since all of the adversary's steps are efficient given $(\PK,\{\sk_i\}_{i \neq j})$, the claim follows from the adaptive hardcore bit property of $f_{\pk_{j}}$.
\end{proof}

This completes the proof of \cref{claim:indistinguishable-u}, which in turn implies the soundness of the measurement protocol.

\section{Succinct Key Generation from iO}\label{sec:succinct-keygen}

In this section, we construct a cryptographic primitive that provides a succinct representation of $N$ key pairs. We call this primitive a ``succinct batch key generation algorithm,'' and provide definitions and a construction based on iO in \cref{sec:succinct-keygen-construction}. In \cref{sec:succinct-clawfree}, we compose our succinct key generation primitive with Mahadev randomized TCFs~\cite{FOCS:Mahadev18a} and prove that the composition satisfies the hypotheses stated in \cref{sec:meas-protocol}, while also having \emph{succinct} keys $(\PK, \SK)$.

\subsection{Batch Key Generation: Definition and Construction }\label{sec:succinct-keygen-construction}
A batch key generation algorithm is an algorithm that outputs a description of many $(\pk, \sk)$-pairs; a \emph{succinct} batch key generation algorithm produces a short such description. Formally, we will define this primitive relative to any dual-mode key generation algorithm. 
\begin{definition}
An algorithm $\Gen$ is said to be a dual-mode key generation algorithm if it takes as input a security parameter $1^\secp$ and a bit $\mode\in\{0,1\}$, and it outputs a pair of keys $(\pk,\sk)$. Moreover, we require \emph{key indistinguishability}:  public keys sampled using $\Gen(1^\secp, 0)$ are computationally indistinguishable from public keys sampled using $\Gen(1^\secp, 1)$.
\end{definition}

\begin{definition}\label{def:succinct-key-gen}
Let $(\pk,\sk) \gets \Gen(1^\secp, \mode)$ denote a dual-mode key generation algorithm. A \textdef{(succinct) batch key generation algorithm} $\mathsf{BatchGen}$ for $\Gen$ is a tuple of p.p.t.\ algorithms $(\Setup, \extpk,\allowbreak \extsk, \Program)$ with the following syntax.
  
\begin{itemize}
      \item $\Setup(1^\secp, N, f)$ takes as input a security parameter $\secp$ in unary; the number of indices $N$ in binary;  and the  description of a circuit 
      $f: [N]\rightarrow \{0, 1\}$.
      It outputs a master public key $\PK$ and a master secret key $\SK$. 
      \item $\extpk(\PK, i)$ is a deterministic algorithm that takes as input a master public key $\PK$ and an index $i\in [N]$. It outputs a public key $\pk_i$. 
      \item $\extsk(\SK, i)$ is a deterministic algorithm takes as input a master secret key $\SK$ and an index $i\in [N]$. It outputs a secret key $\sk_i$. 
      \item $\Program(1^\secp, N, f, i, \pk)$ takes as input $(1^\secp, N, f)$ just as $\Setup$ does, along with two additional inputs: an index $i \in [N]$ and a public key $\pk$. It outputs a master public key $\PK$ and (an implicitly restricted) master secret key $\SK$. 
\end{itemize}

\noindent 
We require that the following three properties are satisfied. Informally, we require that (0) $\Setup(1^\secp, N, f)$ always outputs a representation of valid key pairs, (1) $\Program(1^\secp, \allowbreak N, f, i, \pk)$ successfully programs $\pk$ into the $i$th ``slot'' of $\PK$, (2) if $(\pk, \sk)\gets \Gen(1^\secp, \mode = f(i))$ this programming is undetectable (even given all secret keys), and (3) mode indistinguishability continues to hold for batched keys, even in the presence of ``irrelevant secret keys.''
      
\begin{enumerate}
\item \textbf{Setup Correctness.} For any  $\secp,N\in\mathbb{N}$, any circuit $f:[N]\rightarrow\{0,1\}$, any index $i\in [N]$, we have that for $(\PK, \SK) \gets \Setup(1^\secp, N, f)$ and $(\pk_i, \sk_i) = (\extpk(\PK, i), \extsk(\SK, i))$, $(\pk_i, \sk_i)$ is in the range of $\Gen(1^\secp, \mode = f(i))$. 
\item \textbf{\em Programming Correctness.} For any  $\secp,N\in\mathbb{N}$, any circuit $f:[N]\rightarrow\{0,1\}$, any index $i\in [N]$, and any bit $\mode$, we have the following guarantee: for $(\pk, \sk)\gets \Gen(1^\secp, \mode)$ and $(\PK, \SK)\gets \Program(1^\secp, N, f, i, \pk)$, 
          \[ \extpk(\PK, i) = \pk. 
          \]
with probability $1$.

\item \textbf{\em Programming Indistinguishability.} For any $N=N(\secp)$, any circuit $f:[N]\rightarrow\{0,1\}$ and any index $i\in [N]$, the following distributions are $(\poly(\secp, N), \negl(\secp, N))$-indistinguishable:  
\[ \left\{(\PK, \SK) \gets \Setup(1^\secp, N, f), \sk_j \gets \extsk(\SK, j): \textcolor{blue}{ (\PK, \sk_1, \hdots, \sk_N)} \right\}_{\secp\in\mathbb{N}}
\] 
\[ \approx_c \Big\{(\pk, \sk) \gets \Gen(1^\secp, \mode = f(i)), (\PK, \SK) \gets \Program(1^\secp, N, f, \pk, i), 
\]
\[ \sk_i = \sk \mbox{ and } \forall j \neq i,  \sk_j \gets \extsk(\SK, j): \textcolor{blue}{ (\PK, \sk_1, \hdots, \sk_N)} \Big\}_{\secp\in\mathbb{N}}
\]
While we let the circuit be arbitrary in this definition, we note that it will be instantiated with an efficient circuit of size $\mathsf{poly}(\log N, \secp)$ in our eventual constructions.

\item \textbf{\em Key Indistinguishability:} For any $N = N(\secp)$, for any subset $S\subset [N]$, and for any two circuits $f_0, f_1: [N]\rightarrow \{0,1\}$ such that $f_0(i) = f_1(i)$ for all $i\in S$, for $(\PK_b, \SK_b) \gets \Setup(1^\secp, N, f_b)$, the distributions of keys
    \[ \Big\{ \PK_b, \Big(\sk_i \gets \extsk(\SK_b, i)\Big)_{i\in S}\Big\}_{\secp\in\mathbb{N}}
    \]
    are computationally $(\poly(\secp, N), \negl(\secp, N))$-indistinguishable.
      \end{enumerate}
\end{definition}

We now construct succinct key generation from iO and puncturable PRFs using standard puncturing techniques.
\begin{theorem}
  For any $N(\secp)$, assuming a $(\poly(\secp, N), \negl(\secp, N))$-secure iO scheme and a $(\poly(\secp, N), \negl(\secp, N))$-secure puncturable PRF, there exists a \emph{succinct} batch key generation algorithm $$(\Setup, \extpk, \extsk, \Program)$$ where $\Setup$ supports batch sizes up to $N(\secp)$ and runs in time $\poly(\secp, \log N)$. 
  
  In particular, when $N(\secp) = 2^\secp$ we rely on the sub-exponential hardness of $\iO$ and puncturable PRFs, while for any $N(\secp) = \poly(\secp)$ we rely on polynomial hardness.
\end{theorem}

\begin{proof}

Given a dual-mode key generation algorithm $\Gen$,  an iO scheme $\iO$, and a puncturable PRF family $\PRF$, we define our batch key generation procedure $\mathsf{SuccGen} = (\Setup, \extpk, \extsk, \Program)$ as follows. 

\begin{itemize}
    \item $\Setup(1^\secp, N, f)$ samples a PRF seed $s$ and outputs (as the public key) an obfuscated program $\widetilde{P}=\iO(P_{s, f})$, where $P$ is defined in \cref{fig:succinct-keys}, and (as the secret key) the PRF seed $s$ and the function $f$.   
    \item $\extpk(\PK, i)$ computes and outputs $\pk_i = \widetilde P(i)$ (for $\widetilde P= \PK$). 
    \item $\extsk(\SK, i)$ computes $r = \PRF_{s}(i)$ and $\mode = f(i)$. It then computes $(\pk_i, \sk_i) \gets \Gen(1^\secp, \mode; r)$ and outputs $\sk_i$. 
    \item $\Program(1^\secp, N, f, \pk, i^*)$ samples a PRF seed $s$ and outputs (as the public key) an obfuscated program $\iO(P_{\pk, i^*, s, f})$, where $P_{\pk, i^*, s, f}$ is defined in \cref{fig:program-key}, and (as the secret key) the PRF seed $s$ and the function $f$. 
\end{itemize}

\begin{figure}[h!]
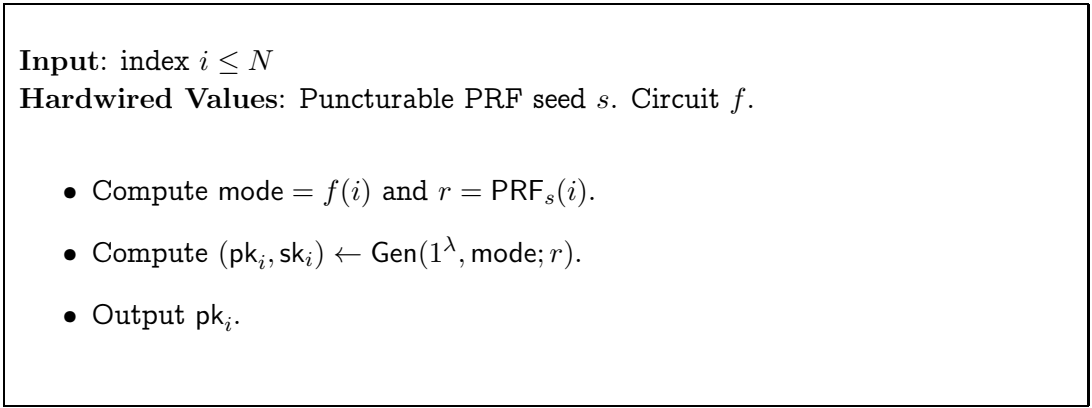

\centering
\begin{tabular}{|p{14cm}|}
\hline\ \\
{\bf Input}: index $i \leq N$\\
{\bf Hardwired Values}: Puncturable PRF seed $s$. Circuit $f$. \\
\begin{itemize}

\item Compute $\mode=f(i)$ and $r = \PRF_{s}(i)$.

\item Compute $(\pk_i, \sk_i) \gets \Gen(1^\secp, \mode; r)$. 

\item Output $\pk_i$. 
	
\end{itemize}
\ \\
\hline
\end{tabular}
\caption{The program $P$.}	
\label{fig:succinct-keys}
\end{figure}

\begin{figure}[h!]
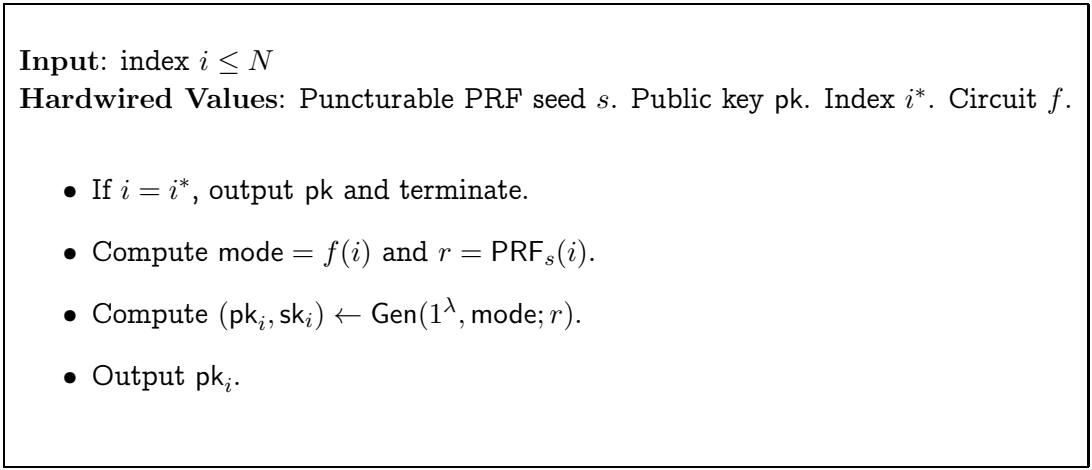

\centering
\begin{tabular}{|p{14cm}|}
\hline\ \\
{\bf Input}: index $i \leq N$\\
{\bf Hardwired Values}: Puncturable PRF seed $s$. Public key $\pk$. Index $i^*$.  Circuit $f$. \\
\begin{itemize}

\item If $i = i^*$, output $\pk$ and terminate. 

\item Compute $\mode=f(i)$ and $r = \PRF_{s}(i)$.

\item Compute $(\pk_i, \sk_i) \gets \Gen(1^\secp, \mode; r)$. 

\item Output $\pk_i$. 
	
\end{itemize}
\ \\
\hline
\end{tabular}
\caption{The program $P_{\pk, i^*, s, f}$. 
}
\label{fig:program-key}
\end{figure}

Succinctness, setup correctness and programming correctness are immediate from the definitions. We now prove programming indistinguishability. 

\begin{claim}\label{claim:programming-indistinguishability}
    For any circuit $f$ and any index $i\in [N]$, the following distributions are $(\poly(\secp, N), \allowbreak \negl(\secp, N))$- computationally indistinguishable: 
          \[ \left\{(\PK, \SK) \gets \Setup(1^\secp, N, f), \sk_j \gets \extsk(\SK, j): \textcolor{blue}{ (\PK, \sk_1, \hdots, \sk_N)} \right\}
          \]
          \[ \approx_c \Big\{(\pk, \sk) \gets \Gen(1^\secp, \mode = f(i)), (\PK, \SK) \gets \Program(1^\secp, N, f, \pk, i), 
          \]
          \[ \sk_j \gets \extsk(\SK, j) (j\neq i), \sk_i = \sk: \textcolor{blue}{ (\PK, \sk_1, \hdots, \sk_N)} \Big\}
          \] 
\end{claim}

\begin{proof}
    We know that $(\iO(P_{s, f}), s) \approx_c (\iO(P_{\pk, i^*, s, f}), s)$ for $(\pk, \sk) \gets \Gen(1^\secp, \mode = f(i^*); \PRF_s(i^*))$ by iO security because these two circuits $P_{s, f}, P_{\pk, i^*, s, f}$ are functionally equivalent.
    
    Moreover, $(\iO(P_{\pk, i^*, s, f}), \{\sk_i\}_{1\leq i\leq N})$ for pseudorandom $(\pk,\sk_{i^*})$ is computationally indistinguishable from $(\iO(P_{\pk, i^*, s, f}), \{\sk_i\}_{1\leq i\leq N})$ for truly random $(\pk, \sk_{i^*})$ by puncturing $s$ at $i^*$ (invoking iO security to do so) and then invoking PRF security.
\end{proof}

Finally, we prove key indistinguishability.

\begin{claim}
For any $N = N(\secp)$, for any subset $S\subset [N(\secp)]$, and for any two circuits $f_0, f_1: [N]\rightarrow \{0,1\}$ such that $f_0(i) = f_1(i)$ for all $i\in S$, for $(\PK_b, \SK_b) \gets \Setup(1^\secp, N, f_b)$, the distributions of keys
    \[ \Big\{ \PK_b, \Big(\sk_i \gets \extsk(\SK_b, i)\Big)_{i\in S}\Big\}_{\secp\in\mathbb{N}}
    \]
    are computationally $(\poly(\secp, N), \negl(\secp, N))$- indistinguishable.
\end{claim}

\begin{proof}
Consider the following hybrid circuits $f'_j$ for $0\leq j\leq N$:
   
   \[ f'_j(i) = f_0(i) \text{ if } i\geq j \text{ and } f'_j(i) = f_1(i) \text{ if } i > j. 
   \]
   Note that $f'_0 = f_0$ and $f'_N = f_1$. Now, we consider the $N+1$ distributions
    \[ \Hyb_j = \Big\{ \PK, \Big(\sk_i \gets \extsk(\SK, i)\Big)_{i\in S}\Big\}_{\secp\in\mathbb{N}}
    \]
   for $(\PK, \SK) \gets \Setup(1^\secp, N, f'_j)$. The claim holds as long as $\Hyb_{j-1} \approx_c \Hyb_{j}$ for all $j\geq 1$. To see that this indistinguishability holds, it suffices to consider two further hybrid distributions:
   \[ \Hyb_{j, 1} = \Big\{ \PK, \Big(\sk_i \gets \extsk(\SK, i)\Big)_{i\neq j\in S}, \sk_{j} \text{ (included if } j\in S)\Big\}_{\secp\in\mathbb{N}}
   \]
   for $(\pk_j, \sk_j) \gets \Gen(1^\secp, \mode = f_0(j))$ and $(\PK, \SK) \gets \Program(1^\secp, N, f'_{j-1}, j, \pk_j)$, and
   \[ \Hyb_{j, 2} = \Big\{ \PK, \Big(\sk_i \gets \extsk(\SK, i)\Big)_{i\neq j\in S}, \sk_{j} \text{ (included if } j\in S)\Big\}_{\secp\in\mathbb{N}}
   \]
   for $(\pk_j, \sk_j) \gets \Gen(1^\secp, \mode = f_1(j))$ and $(\PK, \SK) \gets \Program(1^\secp, N, f'_{j}, j, \pk_j)$.
   
   We have that $\Hyb_{j-1} \approx_c \Hyb_{j, 1}$ by programming indistinguishability (\cref{claim:programming-indistinguishability}). We have that $\Hyb_{j,1}\approx_c \Hyb_{j, 2}$ by considering two cases: if $f_0(j) = f_1(j)$ then $(\pk_j, \sk_j)$ are sampled from identical distributions in the hybrid and $f'_{j-1} = f'_j$, so indistinguishability follows from a single invocation of $\iO$ security. If $f_0(j) \neq f_1(j)$, then $\sk_j$ is not included in the hybrid distributions; moreover, $\pk_j$ in $\Hyb_{j,1}$ is computationally indistinguishable from $\pk_j$ in $\Hyb_{j, 2}$ by the key indistinguishability of $\Gen$. Finally, note that for a fixed $\pk_j$, the programs $P_{\pk_j, j, s, f'_{j-1}}$ and $P_{\pk_j, j, s, f'_{j}}$ are functionally equivalent (as index $j$ is being programmed to $\pk_j$ in both cases), so the claimed indistinguishability now follows from $\iO$ security. 
   
   Finally, we have $\Hyb_{j, 2}\approx_c \Hyb_j$ by programming indistinguishability. This completes the proof of the claim. 
   \end{proof}

\noindent This completes the proof that $(\Setup, \extpk, \extsk,\Program)$ is a succinct batch key generation algorithm for $\Gen$. ~\qedhere

\end{proof}

\subsection{Combining Succinct Key Generation with Mahadev rTCFs}\label{sec:succinct-clawfree}

In our protocols, we compose a batch key generation algorithm (\cref{def:succinct-key-gen}) with a family of Mahadev randomized TCFs (\cref{def:clawfree}). The composition is simple: use a batch key generation procedure $\mathsf{BatchGen} = (\Setup, \extpk, \extsk, \Program)$ to batch the procedure $\mathsf{TCF}.\Gen(1^\secp, \mode)$ for many Mahadev rTCFs $f_{\pk_1}, \hdots, f_{\pk_N}$. The composition has the following syntax:

\begin{itemize}
    \item $\Setup(1^\secp, N, C)$ takes as input the security parameter $\secp$, the batch size $N$ (in binary), and a circuit $C$ computing a function mapping $[N]\rightarrow \{0,1\}$. It outputs a public key $\PK$ and secret key $\SK$.
    \item $\extpk(\PK, i)$ then outputs a public key $\pk_i$ that can be used to evaluate a randomized TCF $f_{\pk_i}$. 
    \item $\extsk(\SK, i)$ outputs a secret key $\sk_i$ that can be used to invert a TCF evaluation $y_i$. 
    \item $\Program$, as defined above, can be used to program a fresh $(\pk_i, \sk_i) \gets \Gen(1^\secp, \mode = C(i))$ into a succinct program generated using circuit $C$. $\Program$ is an auxiliary algorithm used only for analysis. 
\end{itemize}

We now establish that all of the necessary properties listed in \cref{sec:meas-protocol} are satisfied by this composition.

Correctness of the composition (i.e., that key pairs $(\pk_i, \sk_i)$ are in the range of $\mathsf{TCF}.\Gen(1^\secp, \mode = f(i))$) follows immediately from the correctness of $\Setup$. Key indistinguishability of the composition is also inherited directly from the key indistinguishability of $\mathsf{BatchGen}$. 

We next prove that collapsing of $f_{\pk_j}$ holds in the presence of $\PK$ and all $\{\sk_i\}_{i\neq j}$. 

\begin{lemma}\label{lemma:io-collapsing}[Collapsing]
For any circuit $C$, any index $j$, and $(\PK, \SK) \gets \Setup(1^\secp, N, C)$, the TCF $f_{\pk_j}$ is collapsing, even to an adversary given $\PK$ \emph{along with} all secret keys $\{\sk_i\}_{i\neq j}$ besides $\sk_j$.

Formally, a computationally bounded adversary cannot win the following distinguishing game with non-negligible advantage:
\begin{enumerate}
    \item The adversary chooses an index $j\in [N]$ and a circuit $C: [N]\rightarrow \{0,1\}$.
    \item The challenger samples $(\PK, \SK) \gets \Setup(1^\secp, N, C)$.
    \item The challenger sends $(\PK, \{\sk_i\}_{i\neq j})$ to the adversary.
    \item The adversary prepares a quantum state $\ket{\psi}$ on registers $\RegB, \RegX$ along with a string $y$ and sends both to the challenger.
    \item The challenger computes, in superposition, whether $\Check(\pk_j, b, x, y) = 1$. 
    \begin{itemize}
        \item If $\Check$ fails, the challenger samples a random bit $c$ and stops. 
        \item If $\Check$ passes, the challenger samples a random bit $c$; if $c=1$, the challenger measures $\RegB$.
    \end{itemize}
    \item The adversary, given access to the modified $(\RegB, \RegX)$, outputs a bit $c'$ and wins if $c' = c$. 
\end{enumerate}
\end{lemma}

\begin{proof}
   We consider the following hybrid experiments.
   \begin{itemize}
       \item $\Hyb_0$: this is the actual security game.
       \item $\Hyb_1$: In step (2), challenger samples $(\pk, \sk) \gets \Gen(1^\secp, C(j))$ and samples $(\PK, \SK) \gets \Program(1^\secp, N, C, \pk, j)$.
       
       $\Hyb_0$ and $\Hyb_1$ are computationally indistinguishable by the programming indistinguishability of $\mathsf{SuccGen}$. 
       \item $\Hyb_2$: In step (2), the challenger instead samples $(\pk, \sk) \gets \Gen(1^\secp, \inj)$. 
       
       $\Hyb_1$ and $\Hyb_2$ are computationally indistinguishable by the key indistinguishability of the injective/claw-free trapdoor functions.
   \end{itemize}
   
   Finally, in $\Hyb_3$, even a computationally unbounded adversary cannot guess the challenge bit $c$, as with all but negligible probability, $\pk_j = \extpk(\PK, j)$ defines an injective function (by \cref{def:clawfree}), so after verifying that $\Check(\pk_j, b, x, y) = 1$, the register $\RegB$ is already a standard basis state. This completes the proof of \cref{lemma:io-collapsing}. 
\end{proof}

Finally, we prove that the adaptive hardcore bit property of $f_{\pk_j}$ holds given $\PK$ and all $\{\sk_i\}_{i\neq j}$

\begin{lemma}\label{lemma:io-hardcore-bit}[Adaptive Hardcore Bit]
    For any $j$ and any circuit $C$ such that $C(j) = 1 = \twotoone$, for $(\PK, \SK) \gets \Setup(1^\secp, N, C)$, the adaptive hardcore bit property (see \cref{def:clawfree}) holds for the function $f_{\pk_j}$ (with associated secret key $\sk_j$), even if the adversary is given $(\PK, \{\sk_i\}_{i\neq j})$.
\end{lemma}

\begin{proof}
   We consider the following hybrid experiments.
   
   \begin{itemize}
       \item $\Hyb_0$: this is the adaptive hardcore bit security game for $(\pk_j, \sk_j)$ as sampled above.
       \item $\Hyb_1$: this is the adaptive hardcore bit security game for $(\pk, \sk) \gets \Gen(1^\secp, \twotoone)$, $(\PK, \SK) \gets \Program(1^\secp, N, C, \pk, j)$. The adversary is additionally given $\PK$ and $\sk_i = \extsk(\SK, i)$ for all $i\neq j$. 
   \end{itemize}
   
   $\Hyb_0$ and $\Hyb_1$ are computationally indistinguishable by the programming indistinguishability property. Moreover, the adversary's advantage in $\Hyb_1$ is negligible by the adaptive hardcore bit property of the freshly generated key pair $(\pk, \sk) \gets \Gen(1^\secp, \twotoone)$, as a reduction given $\pk$ can simulate $\Hyb_1$ by sampling all other parameters given to the $\Hyb_1$ adversary itself.
   
   This completes the proof of \cref{lemma:io-hardcore-bit}. 
\end{proof}

\section{A Verifier-Succinct Protocol}\label{section:semi-succinct-delegation}

In this section, we present a delegation protocol for $\QMA$ with succinct verifier messages. First, in \cref{subsec:cop}, we describe results due to \cite{TCC:ACGH20} about the parallel repetition of certain commit-challenge-response protocols with a quantum prover. Our treatment is somewhat more abstract than \cite{TCC:ACGH20}, so for completeness we provide proofs of all claims (based on proofs appearing in \cite{TCC:ACGH20}). Next, in \cref{subsec:FHM}, we describe the syntax of a non-interactive, information-theoretic $\QMA$ verification protocol (with quantum verifier) that we will use, due to \cite{FHM18}. In \cref{subsec:semi-succinct-description}, we describe a \emph{verifier-succinct} protocol for $\QMA$ delegation, where the verifier messages (but not the prover messages) are succinct.

\subsection{Quantum commit-challenge-response protocols}\label{subsec:cop}

Consider any \emph{commit-challenge-response} protocol between a quantum prover $P$ and a classical verifier $V$, with the following three phases.

\begin{itemize}
    \item Commit: $P(1^\secp)$ and $V(1^\secp;r)$ engage in a (potentially interactive) commitment protocol, where $r$ are the random coins used by $V$.
    \item Challenge: $V$ samples a random bit $b \gets \{0,1\}$ and sends it to $P$.
    \item Response: $P$ computes a (classical) response $z$ and sends it to $V$. 
\end{itemize}

After receiving the response, $V$ decides to accept or reject the execution.

Consider any non-uniform QPT prover $P^*$, and let $\ket{\psi^{P^*}_{\secp,r}}_{\RegA,\RegC}$ be the (purified) state of the prover after interacting with $V(1^\secp;r)$ in the commit phase, where $\RegC$ holds the (classical) prover messages output during this phase, and $\RegA$ holds the remaining state.

The remaining strategy of the prover can be described by family of unitaries $\left\{U^{P^*}_{\secp,0},U^{P^*}_{\secp,1}\right\}_{\secp \in \mathbb{N}^+}$, where $U^{P^*}_{\secp,0}$ is applied to $\ket{\psi^{P^*}_{\secp,r}}$ on challenge 0 (followed by a measurement of $z$), and $U^{P^*}_{\secp,1}$ is applied to $\ket{\psi^{P^*}_{\secp,r}}$ on challenge 1 (followed by a measurement of $z$).

Let $V_{\secp,r,0}$ denote the accept/reject predicate applied by the verifier to the prover messages when $b=0$, written as a projection to be applied to the registers holding the prover messages, and define $V_{\secp,r,1}$ analogously. Then define the following projectors on $\cA \otimes \cC:$

\[\Pi^{P^*}_{\secp,r,0} \coloneqq {U^{P^*}_{\secp,0}}^\dagger V_{\secp,r,0}U^{P^*}_{\secp,0}, \ \ \ \Pi^{P^*}_{\secp,r,1} \coloneqq {U^{P^*}_{\secp,1}}^\dagger V_{\secp,r,1}U^{P^*}_{\secp,1}.\]

\begin{definition}\label{def:computationally-orthogonal-projectors}
A \emph{commit-challenge-response} protocol has \emph{computationally orthogonal projectors} if for any QPT prover $P^*$, \[\E_{r}\left[\bra{\psi^{P^*}_{\secp,r}}\Pi^{P^*}_{\secp,r,0} \Pi^{P^*}_{\secp,r,1} \Pi^{P^*}_{\secp,r,0} \ket{\psi^{P^*}_{\secp,r}}\right] = \negl(\secp).\]
\end{definition}

Proofs of the following are given in \cref{appendix:cop}.

\begin{lemma}[\cite{TCC:ACGH20}]\label{lemma:two-conditions}
Consider a commit-challenge-response protocol with the following properties.
\begin{enumerate}
    \item $V_{\secp,r,0}$ does not depend on $r$ (that is, it is publicly computable given the transcript).
    \item For any $P^*$, if $\E_{r}\left[\bra{\psi^{P^*}_{\secp,r}}\Pi^{P^*}_{\secp,r,0} \ket{\psi^{P^*}_{\secp,r}}\right] = 1-\negl(\secp),$ then $\E_{r}\left[\bra{\psi^{P^*}_{\secp,r}}\Pi^{P^*}_{\secp,r,1} \ket{\psi^{P^*}_{\secp,r}}\right] = \negl(\secp).$
\end{enumerate}
Then, the protocol has computationally orthogonal projectors.
\end{lemma}

\begin{theorem}[\cite{TCC:ACGH20}]\label{thm:parallel-repetition}
Consider the $\secp$-fold parallel repetition of any commit-challenge-response protocol with computationally orthogonal projectors. The probability that the verifier accepts all $\secp$ parallel repetitions of the protocol is $\negl(\secp)$.
\end{theorem}

\subsection{Non-Interactive Post Hoc Verification of $\QMA$}\label{subsec:FHM}

We recall a useful information-theoretic QMA verification protocol of Fitzsimons, Hajdušek, and Morimae~\cite{FHM18}. In fact, we will use an ``instance-independent'' version due to~\cite{TCC:ACGH20}.
\begin{lemma}[\cite{FHM18,TCC:ACGH20}]\label{lmm:FHM}
For all languages $\cL = (\cL_{\mathsf{yes}}, \cL_{\mathsf{no}}) \in \QMA$ there exists a polynomial $k(\secp)$, a function $\ell(\secp)$ that is polynomial in the time $T(\secp)$ required to verify instances of size $\secp$, a QPT algorithm $P_\mathsf{FHM}$, and a PPT algorithm $V_\mathsf{FHM}$ such that the following holds.
\begin{itemize}
    \item $P_\mathsf{FHM}(x,\ket{\psi}) \to \ket{\pi}$: on input an instance $x \in \{0,1\}^\secp$ and a quantum state $\ket{\psi}$, $P_\mathsf{FHM}$ outputs an $\ell(\secp)$-qubit state $\ket{\pi}$.
    \item \textbf{Completeness.} For all $x\in\cL_{\mathsf{yes}}$ and $\ket{\phi}\in\cR_\cL(x)$ it holds that
    $$
    \Pr[V_\mathsf{FHM}(x, M(h, \ket{\pi})) = \acc : \ket{\pi} \gets P_\mathsf{FHM}\left(x, \ket{\phi}^{\otimes k(\secp)}\right)] \geq 1 - \negl(\secp)
    $$
    where $h \gets \{0,1\}^{\ell(\secp)}$.
    \item \textbf{Soundness.} For all $x\in\cL_{\mathsf{no}}$ and all $\ell$-qubit states $\ket{\pi^*}$ it holds that
    $$
    \Pr[V_\mathsf{FHM}(x, M(h, \ket{\pi^*})) = \acc] \leq \negl(\secp)
    $$
    where $h \gets \{0,1\}^{\ell(\secp)}$.
\end{itemize}
Moreover, when $\mathcal L \in \mathbf{BQP}$, the honest prover algorithm $P_{\mathsf{FHM}}$ is also a $\mathbf{BQP}$ algorithm. 
\end{lemma}
While the result was originally stated in~\cite{FHM18,TCC:ACGH20} to have an inverse polynomial soundness gap, we have driven the soundness gap to negligible by standard QMA amplification. Finally, we remark that although the algorithm $V_\mathsf{FHM}$ is completely classical, the entire verification procedure is quantum since it involves measuring the quantum state sent by the prover.

\subsection{Semi-Succinct Delegation for QMA}\label{subsec:semi-succinct-description}

We describe a protocol for verifying any QMA language $\cL$.

\paragraph{Ingredients:}

\begin{itemize}
    \item Let $(P_\mathsf{FHM},V_\mathsf{FHM})$ be the non-interactive protocol described in \cref{lmm:FHM} for language $\cL$ with associated polynomials $k(\secp),\ell(\secp)$. 
    \item Let $\PRF: \{0,1\}^{\secp} \times \{0,1\}^{\log \ell(\secp)} \to \{0,1\}$ be a pseudo-random function.
    \item Let $P_\mathsf{Meas} = (\Commit,\Open)$ and $V_\mathsf{Meas} = (\Gen,\Test,\Out)$ be the prover and verifier algorithms for an $\ell(\secp)$-qubit (verifier succinct) commit-and-measure protocol, defined in \cref{sec:measurement-protocol-definition} and constructed in \cref{subsec:meas-protocol}.
\end{itemize}

\paragraph{The Protocol:}

\begin{itemize}
    \item The verifier is initialized with an instance $x \in \{0,1\}^\secp$ and the prover is initialized with $x$ and $k(\secp)$ copies of a witness $\ket{\phi} \in \cR_\cL(x)$.
    \item The verifier samples $s \gets \{0,1\}^\secp$, defines $C$ so that $C(i) = \PRF_s(i)$, and computes $(\pk, \sk) \gets \Gen(1^\secp, C)$. It sends $\pk$ to the prover.
    \item The prover first computes $\ket{\psi} \gets P_\mathsf{FHM}\left(x,\ket{\phi}^{\otimes k(\secp)}\right)$, and then computes $(y,\ket{\st}) \gets \Commit(\pk,\ket{\psi})$. It sends $y$ to the verifier.
    \item The verifier samples a random challenge $c \gets \{0,1\}$ and sends $c$ to the prover.
    \item The prover computes $z \gets \Open(\ket{\st}, c)$ and sends $z$ to the verifier.
    \item If $c = 0$, the verifier checks whether $\Test(\pk,(y,z)) = \acc$ and rejects if the test fails. If $c = 1$, the verifier computes $m \gets \Out(\sk, (y, z))$ and checks whether $V_\mathsf{FHM}(x, m) = \acc$. The verifier accepts if and only the verification is successful.
\end{itemize}

\begin{theorem}\label{thm:semi-succinct}
Let $(P_{\mathsf{SS}},V_{\mathsf{SS}})$ be the $\secp$-fold parallel repetition of the above protocol. Then, $(P_{\mathsf{SS}},V_{\mathsf{SS}})$ satisfies completeness and soundness as defined in \cref{def:interactive-argument}. Moreover, for an instance $x$ with QMA verification time $T$, the total size of verifier messages is $\poly(\secp,\log T)$.
\end{theorem}

\begin{proof}
First, the verifier message size guarantee follows from the fact that the verifier initializes the commit-and-measure protocol with a circuit $C$ that succinctly encodes $\ell(\secp)$ bits using a PRF with input size $\log(\ell(\secp))$, where $\ell(\secp)$ is polynomially related to the QMA verification time $T$. 

Next, we argue completeness. For any $x \in \cL$, we show that a single repetition of the protocol accepts with $1-\negl(\secp)$ probability, and so completeness of the $\secp$-fold parallel repetition then follows by a union bound. In the case of a test round, this follows from the test round completeness of the commit-and-measure protocol (\cref{def:measurement-protocol-completeness}). In the case of a measurement round, we first see that measurement round completeness of the commit-and-measure protocol (\cref{def:measurement-protocol-completeness}) implies that the probability that the verifier outputs 1 is $\negl(\secp)$-close to the probability that the $\mathsf{FHM}$ protocol $(P_\mathsf{FHM},V_\mathsf{FHM})$ accepts when run with an honest prover, but where $h = (\PRF_s(1),\dots,\PRF_s(\ell))$ for $s \gets \{0,1\}^\secp$. Then, the security of the $\PRF$ implies that this probability is $\negl(\secp)$-close to the probability that the $\mathsf{FHM}$ protocol accepts when $h \gets \{0,1\}^\ell$. Finally, this is $\negl(\secp)$-close to 1 by the completeness of the $\mathsf{FHM}$ protocol (\cref{lmm:FHM}).

Finally, we argue soundness. Consider any $x \notin \cL$. By \cref{thm:parallel-repetition}, it suffices to show that the single repetition of the protocol satisfies the conditions of \cref{lemma:two-conditions}. We define $V_{\secp,r,0}$ to be the verifier's accept projection on a test round, and $V_{\secp,r,1}$ to be the verifier's accept projection on a measurement round. Condition 1 of \cref{lemma:two-conditions} follows immediately from the structure of the commit-and-measure protocol. Now, consider any prover $P^*$ such that the first expectation in condition 2 is $1-\negl(\secp)$, meaning that $P^*$ passes the test round with $1-\negl(\secp)$ probability. By the soundness of the commit-and-measure protocol, there exists a state $\rho$ such that the probability that the verifier accepts on a measurement round is $\negl(\secp)$-close to the probability that $V_{\mathsf{FHM}}$ accepts given $x,h$, and $M(h,\rho)$, where $h = (\PRF_s(1),\dots,\PRF_s(\ell))$ for $s \gets \{0,1\}^\secp$. By security of the $\PRF$, this probability is $\negl(\secp)$-close to the probability that $V_{\mathsf{FHM}}$ accepts given $x,h$, and $M(h,\rho)$, where $h \gets \{0,1\}$. Since $x \notin \cL$, soundness of the $\mathsf{FHM}$ protocol implies that this is $\negl(\secp)$. Thus, the second expectation in condition 2 of \cref{lemma:two-conditions} is $\negl(\secp)$, which establishes that this condition is satisfied, and completes the proof.

\end{proof}

\def\kilian{\mathsf{AoK}}

\section{The Fully Succinct Protocol} \label{sec:final}

In this section, we compile the verifier-succinct delegation scheme from \cref{section:semi-succinct-delegation} into a full-fledged delegation scheme for $\QMA$. Formally, we assume the existence of a delegation scheme for $\QMA$ satisfying the following properties:

\begin{enumerate}
    \item All verifier messages can be computed in time $\poly(\secp, \log N)$. (This is the definition of verifier succinctness.)
    \item Moreover, the verifier messages can be computed \emph{obliviously} to the $\QMA$ instance and the prover messages (this holds for the \cref{section:semi-succinct-delegation} protocol). 
\end{enumerate}
We present two compilers enabling this:

\begin{enumerate}
    \item The first (and simpler) compiler only additionally assumes the existence of a collapsing hash function, which is implied by LWE (\cref{lemma:lwe-collapsing}). It converts a $2r$-round verifier-succinct protocol into a $4(r+1)$ round fully succinct protocol. In particular, the protocol from  \cref{section:semi-succinct-delegation} is compiled into a 12-round succinct argument for $\QMA$. 
    \item The second compiler additionally assumes collapsing hash function \emph{and} a (classical, post-quantum) fully homomorphic encryption (FHE) scheme. It converts a $2r$ round verifier-succinct protocol to a $4r$ round fully succinct protocol. \emph{Moreover}, if the verifier-succinct protocol is public-coin (except for the first message), then so is the fully succinct protocol. This results in an $8$ round succinct argument system for $\QMA$ that is public-coin except for the first message. 
\end{enumerate}

For simplicity, we write down the compiled protocols in the case $r=2$, corresponding to the protocols from \cref{section:semi-succinct-delegation}. 

Our main tool for these compilers are post-quantum succinct arguments of knowledge for $\NP$ \cite{FOCS:CMSZ21,LMS21}. Specifically, the security guarantees proved in \cite{FOCS:CMSZ21} are insufficient for the compilers, because the post-quantum extraction algorithm from \cite{FOCS:CMSZ21} is not sufficiently \emph{composable} since their extractor might significantly disturb the prover's state. Instead, we make use of a composable variant of the \cite{FOCS:CMSZ21} extractor due to \cite{LMS21} called ``state-preserving succinct arguments of knowledge,''\footnote{We only require a weak variant of what was constructed in \cite{LMS21}, where state preservation is allowed an inverse polynomial $\epsilon$ error.} which we now define. 

\subsection{State-Preserving Succinct Arguments of Knowledge}
\label{sec:state-preserving}

\begin{definition}
  A publicly verifiable argument system $\Pi$ for an $\mathsf{NP}$ language $L$ (with witness relation $R$) is an $\epsilon$-state-preserving succinct argument-of-knowledge if it satisfies the following properties.
  
  \begin{itemize}
      \item Succinctness: when invoked on a security parameter $\secp$ and instance size $n$ and a relation decidable in time $T$, the communication complexity of the protocol is $\poly(\secp, \log T)$. The verifier computational complexity is $\poly(\secp, \log T) + \tilde O(n)$. 
      \item $\epsilon$-State-Preserving Extraction. There exists an extractor $E^{(\cdot)}(x, \epsilon)$ with the following properties
      \begin{itemize}
          \item Efficiency: $E^{(\cdot)}(x, \epsilon)$ runs in time $\poly(n, \secp, 1/\epsilon)$ as a quantum oracle algorithm (with the ability to apply controlled $U$-gates given an oracle $U(\cdot)$), outputting a classical transcript $\tilde \tau$ and a classical string $w$.
          \item State-preserving: Let $\ket{\psi} \in \RegA \otimes \RegI$ be any $\poly(\secp)$-qubit pure state and let $\brho = \Tr_{\RegA}(\ket{\psi}) \in \mathrm{D}(\RegI)$.\footnote{In general, the prover's input state on $\RegI$ may be entangled with some external register $\RegA$, and we ask that computational indistinguishability holds even given $\RegA$. Our definition is stated this way for maximal generality, though we remark that the applications in this section do not require indistinguishability in the presence of an entangled external register.} Consider the following two games:
          \begin{itemize}
              \item Game 0 (Real) Generate a transcript $\tau$ by running $P^*(\brho_{\RegI},x)$ with the honest verifier $V$. Output $\tau$ along with the residual state on $\RegA \otimes \RegI$.
              \item Game 1 (Simulated) Generate a transcript-witness pair $(\tilde \tau,w) \gets E^{P^*(\brho_{\RegI},x)}$. Output $\tilde \tau$ and the residual state on $\RegA \otimes \RegI$.
          \end{itemize}
          Then, we have that the output distributions of Game 0 and Game 1 are computationally $\varepsilon$-indistinguishable to any quantum distinguisher.
          
          \item Extraction correctness: for any $P^*$ as above, the probability that $\tilde \tau$ is an accepting transcript but $w$ is \emph{not} in $R_x$ is at most $\epsilon + \negl(\secp)$. 
          
      \end{itemize}
  \end{itemize}
\end{definition}

\begin{theorem}[\cite{LMS21}]\label{thm:kilian}
Assuming the post-quantum $\poly(\secp, 1/\epsilon)$ hardness of learning with errors, there exists a (4-message, public coin) $\epsilon$-state preserving succinct argument of knowledge for $\NP$. 
\end{theorem}

\subsection{The $\QMA$ Protocol, Version 1}\label{subsec:first-compiler}

Let $\semi$ denote a verifier-succinct $\QMA$ delegation scheme additionally satisfying verifier obliviousness.  For simplicity, we assume that $\semi$ is a four-round protocol. We formalize the execution of $\semi$ on a $\QMA$ instance $x$ as follows:

\begin{itemize}
    \item The verifier computes and sends $\pk \gets  \semi.V_1(1^\secp, 1^{|x|}; r)$ (obliviously to the instance $x$) with randomness $r\gets \{0,1\}^\secp$.
    \item The prover, on initial state $\ket{\psi}$, computes $(y, \brho) \gets \semi.P(1^\secp, \pk, \ket{\psi})$, which results in a message $y$ and residual state $\brho$. 
    \item The verifier computes and sends $\beta = \semi.V_2(r)$, obliviously to the instance $x$ and the prover message $y$.
    \item The prover computes and sends $z \gets \semi.P(\brho, \beta)$. 
    \item The verifier computes and outputs a (potentially expensive) predicate $V(x, y, z, r)$. 
\end{itemize}

Finally, let $\kilian$ denote the state-preserving succinct argument of knowledge of \cref{thm:kilian}, and let $H$ denote a collapsing hash function family mapping $\{0,1\}^*$ to $\{0,1\}^\secp$. Our succinct $\QMA$ delegation protocol $\mathsf{QMArg}$ is defined as follows. 

\begin{enumerate}
    \item The verifier computes and sends $\pk = \semi.V_1(1^\secp, 1^{|x|}; r)$ with randomness $r\gets \{0,1\}^\secp$, along with a hash function $h \gets H_\secp$. 
    \item The prover computes $(y, \brho) \gets \semi.P(1^\secp, \pk, \ket{\psi})$ and sends $\hat y = h(y)$. 
    \item The prover and verifier execute $\kilian$ on the statement ``$\exists w \text{ such that } \hat y = h(w)$.''
    \item The verifier computes and sends $\beta = \semi.V_2(r)$. Note that $\semi.V_2$ is oblivious to the prover message and so can be computed without it. 
    \item The prover computes $z \gets \semi.P(\brho, \beta)$ and sends $\hat z = h(z)$. 
    \item The prover and verifier execute  $\kilian$ on the statement ``$\exists w \text{ such that } \hat z = h(w)$.'' 
    \item The verifier sends $r$.
    \item The prover and verifier execute  $\kilian$ on the statement $\exists w_1, w_2$ such that $\hat y = h(w_1), \hat z = h(w_2)$, and $V(x, w_1, w_2, r) = 1$. 
\end{enumerate}

\noindent Completeness of the protocol follows directly from the completeness of $\semi$ and $\kilian$. Moreover, succinctness follows directly from the compression of $H$, the verifier succinctness of $\semi$, and the succinctness of $\kilian$.

Since $\kilian$ has a round complexity of $4$ and the first message can be re-used (indeed, $h$ can be used as the first message for $\kilian$), the round complexity of $\mathsf{QMArg}$ is 12.

\subsubsection{Proof of Soundness}

\begin{theorem}\label{thm:first-compiler}
   Assume that $\semi$ is (post-quantum) computationally sound, $H$ is collapsing,\footnote{Collision-resistance of $H$ suffices.} and that $\kilian$ is an $\epsilon$-state-preserving argument of knowledge. Then, $\mathsf{QMArg}$ is (post-quantum) computationally sound.  
\end{theorem}

\begin{proof}
  Let $x\not\in L$ and suppose that a QPT $P^*(\brho, x)$ breaks the soundness of $\mathsf{QMArg}$ with probability $\epsilon^*$. We use $P^*$, together with the soundness guarantees of $\kilian$ and the collision resistance property of $H$, to break the soundness of $\semi$.  In particular, consider the following attack on the soundness of $\semi$:
  
  \begin{itemize}
      \item Set an accuracy parameter $\epsilon = \frac{\epsilon^*}{10}$. Whenever we call the $\kilian$ extractor $E$, we will use accuracy parameter $\epsilon$.  
      \item Given a verifier message $\pk$, we feed $(h, \pk)$ to $P^*(\brho, x)$ and obtain a hash value $\hat y$. Then, we run the $\kilian$ extractor $E$ on $P^*$'s execution of step (3) (the first execution of $\kilian$), outputting a triple $(\tilde \tau_1, \tilde \brho_1, y)$. We send $y$ to the verifier.
      \item Given the verifier challenge $\beta$, we run $P^*(\tilde \brho_1, \pk, h, \tilde \tau_1)$ to obtain a message $\hat z$. Then, we run the $\kilian$ extractor $E$ on $P^*$'s execution of step (6), obtaining a triple $(\tilde \tau_2, \tilde \brho_2, z)$. We send $z$ to the verifier. 
  \end{itemize}
  
  Finally, to analyze the behavior of this attack, we consider the following additional step (this is only a mental experiment).
  
  \begin{itemize}
      \item Given the secret verifier randomness $r$, run the $\kilian$ extractor $E$ on $P^*$'s execution of step (8) to obtain a triple $(\tilde \tau_3, \tilde \brho_3, y', z')$. 
  \end{itemize}

\begin{claim}
   With probability at least $\epsilon$ over the attack experiment, we have that $\semi.V(x, \allowbreak \mathbf{y}, \mathbf{z}, r) = 1$. 
\end{claim}

Note that this claim contradicts the soundness of $\semi$. 

\begin{proof}
  The equation $\semi.V(x, \mathbf{y}, \mathbf{z}, r) = 1$ follows from the following properties of an execution of the mental experiment:
  
  \begin{itemize}
      \item $h(y) = \hat y$
      \item $h(z) = \hat z$
      \item $h(y') = \hat y$
      \item $h(z') = \hat z$
      \item $\semi.V(x, y', z', r) = 1$. 
  \end{itemize}
  
  The above suffices because it implies that $(y, z) = (y', z')$ except with negligible probability by the collapsing (or just collision-resistance) of $H$, and so the last equation implies that $\semi.V(x, y, z, r)=1$ (except with negligible probability). 
  
  Finally, we note that all five of the above conditions simultaneously hold with probability at least $\epsilon$ by the state-preservation and correctness of $E$. More specifically,
  
  \begin{itemize}
      \item The transcript $(x, \pk, h, \hat y, \tilde \tau_1, \beta, \hat z, \tilde \tau_2, r, \tilde \tau_3)$ is accepting (according to the $\mathsf{QMArg}$ verifier) with probability at least $\epsilon^* - 3\epsilon$. This follows by a hybrid argument invoking the state preservation of $E^*$ on the three executions of $\kilian$ (first w.r.t. $\tilde \tau_1$, then $\tilde \tau_2$, then $\tilde \tau_3$). 
      \item Then, the correctness property of $E$ implies that all five conditions hold simultaneously with probability at least $\epsilon^* - 6\epsilon - \negl(\secp)$. \qedhere
  \end{itemize}
\end{proof}
\noindent This completes the proof of soundness of $\mathsf{QMArg}$. \qedhere

\end{proof}

\subsection{The $\QMA$ Protocol, Version 2}\label{subsec:public-coin}
We now describe a public-coin variant of the \cref{subsec:first-compiler} transformation that additionally uses a Fully Homomorphic Encryption (FHE) scheme $\mathsf{FHE} = (\mathsf{FHE}.\Gen, \mathsf{FHE}.\Enc, \mathsf{FHE}.\Dec, \mathsf{FHE}.\Eval)$. 

Let $\semi$ and $\kilian$ denote the argument systems from \cref{subsec:first-compiler}. Then, our second succinct argument system $\mathsf{QMArg}_2$  is defined as follows. 

\begin{enumerate}
    \item The verifier computes and sends $\pk = \semi.V_1(1^\secp, 1^{|x|}; r)$ with randomness $r\gets \{0,1\}^\secp$, along with a hash function $h \gets H_\secp$.
    \item \textcolor{red}{The verifier also samples $(\FHE.\pk, \FHE.\sk) \gets \FHE.\Gen(1^\secp)$ and computes FHE ciphertext $\ct_V = \mathsf{FHE}.\Enc(\FHE.\pk, r)$. The verifier sends $\FHE.\pk, \ct_V$ to the prover.}
    \item The prover computes $(y, \brho) \gets \semi.P(1^\secp, \pk, \ket{\psi})$ and sends $\hat y = h(y)$. 
    \item The prover and verifier execute $\kilian$ on the statement ``$\exists w \text{ such that } \hat y = h(w)$.''
    \item The verifier computes and sends $\beta = \semi.V_2(r)$. Note that $\semi.V_2$ is oblivious to the prover message and so can be computed without it. 
    \item The prover computes $z \gets \semi.P(\brho, \beta)$ and sends $\hat z = h(z)$. \textcolor{red}{The prover also computes $\ct_P = \mathsf{FHE}.\mathsf{Eval}(V(x, y, z, \cdot), \ct_V)$ and sends $\ct_P$ to the verifier. }
    \item The prover and verifier execute  $\kilian$ on the statement $\exists w_1, w_2$ such that $\hat y = h(w_1), \hat z = h(w_2)$, \textcolor{red}{and $\ct_P =\mathsf{FHE}.\mathsf{Eval}(V(x, y, z, \cdot), \ct_V)$.} 
    \item \textcolor{red}{The verifier checks that $\mathsf{FHE}.\Dec(\ct_P) = 1$. }
\end{enumerate}

\noindent As before, completeness and succinctness follow immediately from the definitions. Additionally, we note that the round complexity has been reduced to $8$ because $\kilian$ is only invoked twice. Finally, we note that as long as $\semi$ and $\kilian$ are public-coin (except for the first message of $\semi$), then $\mathsf{QMArg}_2$ is also public-coin (except for the first verifier message). 

\subsubsection{Proof of Soundness}\label{thm:second-compiler}
\begin{theorem}
   Assume that $\semi$ is (post-quantum) computationally sound, $H$ is collapsing, $\FHE$ is semantically secure, and that $\kilian$ is an $\epsilon$-state-preserving argument of knowledge. Then, $\mathsf{QMArg}_2$ is (post-quantum) computationally sound.  
\end{theorem}

\begin{proof}
  Let $x\not\in L$ and suppose that a QPT $P^*(\brho, x)$ breaks the soundness of $\mathsf{QMArg}_2$ with probability $\epsilon^*$. We use $P^*$, together with the soundness guarantees of $\kilian$, the semantic security of $\FHE$, and the collision resistance property of $H$, to break the soundness of $\semi$.  In particular, consider the following attack on the soundness of $\semi$:
  
  \begin{itemize}
      \item Set an accuracy parameter $\epsilon = \frac{\epsilon^*}{10}$. Whenever we call the $\kilian$ extractor $E$, we will use accuracy parameter $\epsilon$.  
      \item Given a verifier message $\pk$, we sample $h, \FHE.\pk$ ourselves and feed $(h, \pk, \FHE.\pk, \ct_V = \FHE.\Enc(\FHE.\pk, 0))$ to $P^*(\brho, x)$ and obtain a hash value $\hat y$. Then, we run the $\kilian$ extractor $E$ on $P^*$'s execution of step (3) (the first execution of $\kilian$), outputting a triple $(\tilde \tau_1, \tilde \brho_1, y)$. We send $y$ to the verifier.
      \item Given the verifier challenge $\beta$, we run $P^*(\tilde \brho_1, \pk, h, \tilde \tau_1)$ to obtain a message $\hat z, \ct_P$. Then, we run the $\kilian$ extractor $E$ on $P^*$'s execution of step (6), obtaining a triple $(\tilde \tau_2, \tilde \brho_2, y', z)$. We send $z$ to the verifier. 
  \end{itemize}
  
  We claim that this attack breaks the soundness of $\semi$ -- meaning that $V(x, y, z, r) = 1$ -- with probability at least $\epsilon^* - 4\epsilon - \negl(\secp)$. To prove this, by $\FHE$ semantic security, it suffices to show the same thing when $\ct_V$ is instead sampled as \textcolor{red}{$\FHE.\Enc(\FHE.\pk, r)$}. 
  
  From here, the proof proceeds similarly to the proof of \cref{thm:first-compiler}. In particular, the equation $\semi.V(x, y, z, r) = 1$ follows from the following properties of the hybrid attack execution:
  
  \begin{itemize}
      \item $h(y) = \hat y$
      \item $h(y') = \hat y$
      \item $h(z) = \hat z$
      \item $\FHE.\Eval(V(x, y', z, \cdot), \ct_V) = \ct_P$
      \item $\FHE.\Dec(\FHE.\sk, \ct_P) = 1$.
  \end{itemize}
  
  \noindent This suffices due to the collapsing of $H$ and the correctness of $\FHE.\Eval$. By the same argument as in the proof of \cref{thm:first-compiler}, these properties simultaneously hold with probabiltiy at least $\epsilon^* - 4\epsilon - \negl(\secp)$ by the state-preserving extraction properties of $\kilian$.
  
  This completes the proof of soundness of $\mathsf{QMArg}_2$.

 \end{proof}

\ifsubmission
\else
\section{Additional Results}\label{sec:additional}

In this section, we describe a number of additional new results that follow from our template for building succinct arguments for QMA. First, we show how to compile the protocol from \cref{subsec:public-coin} into a \emph{two-message} succinct argument for QMA in the quantum random oracle model. We sometimes refer to such argument systems as \textdef{designated-verifier SNARGs} ($\mathsf{dvSNARGs}$) in the QROM. Next, we show how to obtain batch arguments for QMA (where the communication size only depends on a \emph{single} instance size) from only the quantum hardness of learning with errors (i.e. without indistinguishability obfuscation). Finally, we describe how to add \emph{zero-knowledge} to our succinct argument in the plain model and to our $\mathsf{dvSNARG}$ in the QROM.

\subsection{Succinct Non-interactive Arguments in the QROM}\label{subsec:two-message-QROM}

Consider any constant-round protocol $(P,V)$ for language $\cL$ that is public-coin except for the first message. That is, the verifier is defined by two circuits $(V_0,V_1)$. Given instance $x$, the verifier first samples random coins $r$ and computes a first message $s_0 = V_0(x,r)$. Then, the subsequent verifier message are uniformly random strings $s_1,\dots,s_c$ of at least $\secp$ bits. Finally, the verifier computes a circuit $V_1(x,r,s_0,t_0,s_1,t_1,\dots,s_c,t_c)$ that determines whether it accepts or rejects, where $t_0,\dots,t_c$ are the prover messages. Let $H: ((I \times S) \cup S) \times  ([c] \times T) \to S$ be a random oracle, where $I$ is the space of instances, $S$ is the space of verifier messages, and $T$ is the space of prover messages. Let $(P_{\mathsf{FS}},V_{\mathsf{FS}})$ be the following protocol.

\begin{itemize}
    \item Given $x$, $V_{\mathsf{FS}}$ samples $r$ and outputs $s_0 = V_0(x,r)$.
    \item $P_{\mathsf{FS}}$ runs $P$ on $(x,s_0)$ to obtain $t_0$. Then it computes $s_1 = H((x,s_0),(0,t_0))$ and continues to run $P$ on $s_1$ to obtain $t_1$. Then for $i \in [c]$, it computes $s_i = H(s_{i-1},(i,t_{i-1}))$ and continues to run $P$ on $s_i$ to obtain $t_i$. Finally, it sends $(t_0,\dots,t_c)$.
    \item $V_{\mathsf{FS}}$ checks that $s_1 = H((x,s_0),(0,t_0))$ and that for each $i \in [2,\dots,c]$, $s_i = H(s_{i-1},(i-1,t_{i-1}))$. If so, it outputs $V(x,r,s_0,t_0,\dots,s_c,t_c)$.
\end{itemize}

\begin{theorem}[Multi-input measure-and-reprogram \cite{C:DonFehMaj20}]\label{thm:measure-and-reprogram}
   Let $c$ be an integer, and $W,X,Y$ be finite sets. There exists a polynomial-time quantum algorithm $S$ such that the following holds. Let $A$ be an arbitrary quantum oracle algorithm that makes $q$ queries to a uniformly random $H: (W \cup Y) \times X \to Y$ and outputs a tuple $(x_0,\dots,x_c)$. Then for any $\widehat{x} \in X^{c+1}$ without duplicate entries, any predicate $V$, and any $w \in W$,
   
   \begin{align*}
       &\Pr_{y_1,\dots,y_c}\left[(x_0,\dots,x_c) = \widehat{x} \wedge V(w,x_0,y_1,x_1,\dots,y_c,x_c) = 1 : (x_0,\dots,x_c) \gets S^A(y_1,\dots,y_c)\right] \\ &\geq \frac{c!}{(q+c+1)^{2c}}\Pr_{H}\left[\begin{array}{l}(x_0,\dots,x_c) = \widehat{x} \wedge V\left(\begin{array}{l}w,x_0,\\y_1 \coloneqq H(w,x_0),x_1,\\y_2 \coloneqq H(y_1,x_1),x_2,\dots,\\y_c \coloneqq H(y_{c-1},x_{c-1}),x_c\end{array}\right) = 1 \end{array}: (x_0,\dots,x_c) \gets A^H \right] - \epsilon_{\widehat{x}},
   \end{align*}
   
   where $\sum_{\widehat{x}}\epsilon_{\widehat{x}} = c!/|Y|$, and $S^A$ is an algorithm that has black-box access to the algorithms of $A$ and for each $i \in [c]$, receives $y_i$ and only after outputting $x_{i-1}$.\footnote{This theorem as stated is actually a special case of \cite[Theorem 7]{C:DonFehMaj20}, where $w$ is fixed. In other words, it corresponds to \cite[Theorem 7]{C:DonFehMaj20} where the class of adversaries considered all produce a fixed $w$ as the first part of their output.}
   
\end{theorem}

\begin{theorem}\label{thm:fiat-shamir}
   If $(P,V)$ is a sound protocol for $\cL$, then $(P_{\mathsf{FS}},V_{\mathsf{FS}})$ is a sound protocol for $\cL$ in the quantum random oracle model.
\end{theorem}

The following corollary then follows immediately from the protocol given in \cref{subsec:public-coin}.

\begin{corollary}
Assuming post-quantum indistinguishability obfuscation, the post-quantum hardness of the learning with errors problem, and post-quantum fully homomorphic encryption, there exists a designated verifier succinct non-interactive argument system (dvSNARG) for QMA in the quantum random oracle model.
\end{corollary}

\begin{proof} (of \cref{thm:fiat-shamir})
Let $H: ((I \times S) \cup S) \times ([c] \times T)\to S$ be the random oracle used in $(P_{\mathsf{FS}},V_{\mathsf{FS}})$. Consider an adversary $A$ in the protocol $(P_{\mathsf{FS}},V_{\mathsf{FS}})$ that makes $q = \poly(\secp)$ queries to $H$ and consider any $x \notin \cL$. For any $r$, let $V_{1,r}$ be the predicate $V_1$ with $r$ hard-coded, and let $A_r$ be the adversary $A$ initialized with $(x,V_0(x,r))$. Define $\epsilon(r)$ to be the success probability of $A_r$ (that is, the probability it makes $V_{\mathsf{FS}}$ output 1). Then for any fixed $r$, 

\begin{align*}
    \epsilon(r) = \Pr_{H}\left[V_{1,r}\left(\begin{array}{l}x,t_0,\\s_1 \coloneqq H((x,V_0(x,r),(0,t_0)),t_1,\dots,\\ s_c \coloneqq H(s_{c-1},(c-1,t_{c-1})),t_c\end{array}\right) = 1 : (t_0,\dots,t_n) \gets A_r^{H}\right].
\end{align*}

Note that the overall success probability $A$ is $\epsilon \coloneqq \E_{r}[\epsilon(r)]$. Now, by setting $W = (I \times S), X = ([c] \times T),Y = S$, and $w = (x,V_0(x,r))$,  \cref{thm:measure-and-reprogram} implies that for any fixed $r$, the success probability $\delta(r)$ of the simulator $S^{A_r}(s_1,\dots,s_c)$ is
\[\delta(r) \geq \frac{1}{\poly(\secp)}\epsilon(r) - \negl(\secp),\] which follows by (i) summing over all $t_0,\dots,t_1$ and noting that $(0,t_0),\dots,(c,t_c)$ contain no duplicates, and (ii) the fact that $q = \poly(\secp)$ and $c$ is a constant. Finally, observe that by the soundness of $(P,V)$, $\E_{r}[\delta(r)] = \negl(\secp)$. Indeed, by definition $S^A$ is a valid cheating prover in the protocol $(P,V)$ since it only receives random $s_i$ after outputting $t_{i-1}$. This establishes that \[\frac{1}{\poly(\secp)}\E_{r}[\epsilon(r)] - \negl(\secp) \leq \negl(\secp),\] which implies that $\epsilon = \negl(\secp)$.

\end{proof}

\subsection{Batch Arguments for QMA}

Now, we show how to obtain batch arguments for QMA from the post-quantum hardness of learning with errors. 
We first describe a verifier-succinct protocol for verifying $n$ QMA instances, where the verifier message size only grows with the time $T$ needed for QMA verification of a \emph{single} instance. Note that here we do \emph{not} use the succinct key generation protocol from \cref{sec:succinct-keygen-construction}, and thus do not rely on indistinguishability obfuscation.

\paragraph{Ingredients:}

\begin{itemize}
\item Let $(P_\mathsf{FHM},V_\mathsf{FHM})$ be the non-interactive protocol described in \cref{lmm:FHM} for language $\cL$ with associated polynomials $k(\secp),\ell(\secp)$. 
\item Let $P_\mathsf{Meas} = (\Commit,\Open)$ and $V_\mathsf{Meas} = (\Gen,\Test,\Out)$ be the prover and verifier algorithms for an $\ell(\secp)$-qubit  commit-and-open measurement protocol, defined in \cref{sec:measurement-protocol-definition} and constructed in \cite{FOCS:Mahadev18a}.
\end{itemize}

\paragraph{The Protocol:}

\begin{itemize}
\item The verifier is initialized with $n$ instances $(x_1, \dots, x_n) \in \{0,1\}^\secp$ and the prover is initialized with $(x_1, \dots, x_n)$ and $k(\secp)$ copies of each witness $\ket{\phi_j} \in \cR_\cL(x_j)$.
\item The verifier samples $h \gets \{0,1\}^{\ell(\secp)}$ and defines $C$ such that $C(i) = h_i$. The verifier computes $(\pk, \sk) \gets \Gen(1^\secp, C)$ and sends $\pk$ to the prover.
\item For each $j \in [n]$, the prover first computes $\ket{\psi_j} \gets P_\mathsf{FHM}\left(x_j,\ket{\phi_j}^{\otimes k(\secp)}\right)$, and then computes $(y_j,\ket{\st_j}) \gets \Commit(\pk,\ket{\psi_j})$. It sends $(y_1, \dots, y_n)$ to the verifier.
\item The verifier samples a random challenge $c \gets \{0,1\}$ and sends $c$ to the prover.
\item For each $j \in [n]$, the prover computes $z_j \gets \Open(\ket{\st_j}, c)$ and sends $(z_1, \dots, z_n)$ to the verifier.
\item If $c = 0$, the verifier checks whether $\Test(\pk,(y_j,z_j)) = \acc$ and rejects if the test fails on any index. If $c = 1$, the verifier computes $m_j \gets \Out(\sk, (y_j, z_j))$ and checks whether $V_\mathsf{FHM}(x_j, m_j) = \acc$. The verifier accepts if and only if all of the verifications are successful.
\end{itemize}

\begin{theorem}
Let $(P_{\mathsf{Batch}},V_{\mathsf{Batch}})$ be the $\secp$-fold parallel repetition of the above protocol. Then, $(P_{\mathsf{Batch}},V_{\mathsf{Batch}})$ satisfies completeness as in \cref{def:interactive-argument} and soundness as in \cref{def:somewhere-soundness}. Moreover, for instances $(x_1,\dots,x_n)$ where $T$ is the maximum QMA verification time for any individual $x_i$, the total size of verifier messages is $\poly(\secp,T)$.
\end{theorem}

\begin{proof}

First, the verifier message guarantee follows immediately from the description of the protocol. Completeness follows via the same argument used to prove completeness in \cref{thm:semi-succinct} (without the additional step involving the PRF). Soundness also follows along the same lines, except that, if $x_i \notin \cL$, we define $V^{(i)}_{\secp,r,0}$ to be the verifier's accept projection \emph{on instance $i$} on a test round, and $V^{(i)}_{\secp,r,1}$ to be the verifier's accept projection \emph{on instance $i$} on a measurement round. \end{proof}

Finally, we observe that \cref{thm:first-compiler} holds when the verifier-succinct protocol is replaced with the batch protocol above, with no change in analysis. This results in the following corollary. 

\begin{corollary}
Assuming the post-quantum hardness of the learning with errors problem, there exists a batch argument for QMA, where the total communication is polynomial in the QMA verification time for a single QMA instance.
\end{corollary}



\subsection{Zero Knowledge}\label{sec:zk}

In this section, we provide sketches for how to obtain the following results. 

\begin{itemize}
    \item A (non-adaptive) \emph{zero-knowledge} succinct argument for QMA (in the plain model).
    \item A (non-adaptive) \emph{zero-knowledge} $\mathsf{dvSNARG}$ for QMA in the quantum random oracle model.
\end{itemize}

In both of our sketches, we will make use of secure two-party computation for \emph{reactive} functionalities, which are interactive functionalities where multiple public circuits may be computed sequentially over private inputs, and where the description of these public circuits may be determined after some of the private inputs are submitted to the functionality and may even depend on the outputs of previously computed circuits.

\paragraph{The plain model.} Here, we can start with the protocol in \cref{subsec:first-compiler}. This protocol as such does not provide any hiding property for the prover's witness. However, we can use secure two-party computation for the following reactive functionality to hide all information about the prover's witness from the verifier, while preserving soundness.

\begin{itemize}
    \item Take as input random coins $r_P, r_V$ from each party, compute the verifier's first message of the protocol using randomness $r \coloneqq r_P \oplus r_V$, and output this message to the prover.
    \item For each subsequent round, take as input the prover's message, and then compute and output the next verifier's message (using random coins $r$) to the prover.
    \item After the final prover's message, compute and output the verifier's verdict (based on all prover messages and $r$) to the verifier.
\end{itemize}

Note that, since the verifier is classical, this functionality can be implemented by a protocol for (post-quantum) secure two-party computation of classical (reactive) functionalities, such as \cite{C:HalSmiSon11}.\footnote{Note that \cite{C:HalSmiSon11} is based on Watrous rewinding, and thus requires polynomially many rounds of interaction. We do not attempt to optimize the round-complexity of our zero-knowledge protocol, but note that non-black-box techniques such as those of \cite{STOC:BitShm20} (with additional assumptions), or a relaxation to $\epsilon$-zero-knowledge \cite{cryptoeprint:2021:1516} could result in a constant-round protocol.} To argue soundness (for any fixed no instance), we can run the two-party computation simulator for a malicious prover in order to extract inputs from the prover and reduce to soundness of the underlying protocol. To argue zero-knowledge (for any fixed yes instance), we can run the two-party computation simulator for a malicious verifier, programming the final output to 1. Note that the verifier only receives this single bit of information from the functionality, which is internally running an honest verifier. Thus, this simulation and output is indistinguishable from the real interaction with an honest prover.  

\paragraph{The QROM.} In order to add zero-knowledge to our two-message succinct argument in the QROM, we have to be careful in order to avoid using the random oracle in a non-black-box manner. We achieve this in two steps: We first construct a
constant-round honest-verifier zero-knowledge argument where, (i) the verifier is public-coin except for the first message, and (ii) the protocol remains zero-knowledge even against a verifier that computes its first message maliciously. Second, we compress this protocol into a two-message protocol in the QROM using the same arguments in \cref{subsec:two-message-QROM}. Since the protocol has constant rounds and is public-coin after the first verifier message,  soundness holds via the same argument. Zero-knowledge holds because we have zero-knowledge against a malicious first message, and honest-verifier zero-knowledge with respect to all subsequent messages, which will be sampled uniformly at random by the random oracle.

We now turn our attention to the construction of the constant-round honest-verifier zero-knowledge argument. We are going to assume the existence of a post-quantum secure two-message two-party computation protocol for reactive (classical) functionalities. One can instantiate this with the two-message secure computation protocol of \cite{C:IshPraSah08} based on (post-quantum) two-message oblivious transfer in the common random string (CRS) model, which we can instantiate from the post-quantum hardness of learning with errors \cite{C:PeiVaiWat08}. Note that when we later compress this protocol in the QROM, the CRS can be sampled by querying the random oracle on a fixed input. We will use such a protocol to implement the following reactive functionality.

\begin{itemize}
    \item Take as input random coins $r$ from the verifier.
    \item The verifier's first message is sampled by the verifier given to the functionality as a \emph{public input}.
    \item Take as input the prover's first message.
    \item The verifier's second message is sampled by the verifier and given to the functionality as a \emph{public input}.
    \item \dots
    \item Take as input the prover's final message.
    \item Check that the verifier's first message is computed honestly from random coins $r$, and if so, compute the verifier's verdict using $r$, the prover messages, and the verifier's messages, and deliver this output to the verifier.
\end{itemize}
Note that all the verifier messages are still sampled publicly and as in the protocol from \cref{subsec:public-coin}, so the prover can still compute its responses given these messages.

Now, we argue that the resulting protocol satisfies the required properties. To argue soundness, we can run the two-party computation simulator for a malicious prover in order to extract inputs from the prover and reduce to soundness of the original protocol. To argue honest-verifier zero-knowledge with a malicious first message, we can run the two-party computation simulator for a malicious verifier, programming the final output to 1. However, since we are allowing the verifier to choose its first message maliciously, we have to argue that for \emph{any} choice of randomness used to generate the verifier's first message, the subsequent interaction between honest prover (on input a valid witness for a true statement) and an honest verifier results in the verifier outputting 1 with overwhelming probability. Recalling the structure of the verifier's first message in the \cref{subsec:public-coin} protocol, we see that this requires perfectly correct FHE and a perfectly correct measurement protocol. Achieving FHE with perfect correctness is standard by truncating the error distribution, and we can obtain a perfectly correct measurement protocol as discussed in \cref{sec:mahadev-rtcf} and \cref{subsec:commit-and-measure}.

\fi

\bibliographystyle{alpha}
\bibliography{crypto,references,abbrev3}

\appendix

\ifsubmission

\else
\fi

\ifsubmission

\else
\fi

\section{Proofs from \cref{section:semi-succinct-delegation}}\label{appendix:cop}

\begin{lemma}
Consider a commit-challenge-response protocol with the following properties.
\begin{enumerate}
    \item $V_{\secp,r,0}$ does not depend on $r$ (that is, it is publicly computable given the transcript).
    \item For any $P^*$, if $\E_{r}\left[\bra{\psi^{P^*}_{\secp,r}}\Pi^{P^*}_{\secp,r,0} \ket{\psi^{P^*}_{\secp,r}}\right] = 1-\negl(\secp),$ then $\E_{r}\left[\bra{\psi^{P^*}_{\secp,r}}\Pi^{P^*}_{\secp,r,1} \ket{\psi^{P^*}_{\secp,r}}\right] = \negl(\secp).$
\end{enumerate}
Then, the protocol has computationally orthogonal projectors.
\end{lemma}

\begin{proof}

Suppose there exists a prover $P^*$ and a polynomial $p(\secp)$ such that for infinitely many $\secp$, \[\E_r\left[\bra{\psi_{\secp,r}^{P^*}}\Pi_{\secp,r,0}^{P^*}\Pi_{\secp,r,1}^{P^*}\Pi_{\secp,r,0}^{P^*}\ket{\psi_{\secp,r}^{P^*}}\right] \geq 1/p(\secp).\]

Define an alternate prover $\widehat{P}^*$ as follows.

\begin{enumerate}
    \item $\widehat{P}^*$ takes as input $p(\secp)^4$ copies of $P^*$'s auxiliary advice, and $\pk$ sampled by the verifier.
    \item Repeat the following at most $p(\secp)^4$ times:
    \begin{enumerate}
        \item Prepare the state $\ket{\psi_{\secp,r}^{P^*}}$ using a copy of $P^*$'s auxiliary advice.
        \item Apply the projective measurement $\left\{\Pi_{\secp,r,0}^{P^*}, \mathbb{I} - \Pi_{\secp,r,0}^{P^*}\right\}$, which is efficient due to property 1 of the commit-challenge-response protocol.
        \item If the first outcome is observed, output the resulting state. Otherwise, repeat.
    \end{enumerate}
    \item If $\widehat{P}^*$ has not terminated, output a dummy state $\ket{\phi}$ such that $\bra{\phi}\Pi_{\secp,r,0}^{P^*}\ket{\phi} = 1$.
\end{enumerate}

Let $\ket{\psi_{\secp,r}^{\widehat{P}^*}}$ be the state that results from the above procedure. Finally, let $\widehat{P}^*$ act identically to $P^*$ after this point.

Next, let $\cR_{\mathsf{term},\secp} \coloneqq \left\{r : \bra{\psi_{\secp,r}^{P^*}}\Pi_{\secp,r,0}^{P^*}\ket{\psi_{\secp,r}^{P^*}}> 1/p(\secp)^2\right\}$, and note that for any $r \in \cR_{\mathsf{term},\secp}$, \begin{align*}
\Pr\left[\ket{\psi_{\secp,r}^{\widehat{P}^*}} \neq \frac{\Pi_{\secp,r,0}^{P^*}\ket{\psi_{\secp,r}^{P^*}}}{\| \Pi_{\secp,r,0}^{P^*}\ket{\psi_{\secp,r}^{P^*}}\|}\right] \leq \left(1-1/p(\secp)^2\right)^{p(\secp)^4} \leq e^{-p(\secp)^2} = \negl(\secp). 
\end{align*}

Now, on the one hand, 

\begin{align*}
    &\frac{1}{|\cR|}\sum_{r \in \cR_{\mathsf{term},\secp}}\bra{\psi_{\secp,r}^{P^*}}\Pi_{\secp,r,0}^{P^*}\Pi_{\secp,r,1}^{P^*}\Pi_{\secp,r,0}^{P^*}\ket{\psi_{\secp,r}^{P^*}} \\
     \leq &\frac{1}{|\cR|}\sum_{r \in \cR_{\mathsf{term},\secp}}\bra{\psi_{\secp,r}^{\widehat{P}^*}}\Pi_{\secp,r,1}^{\widehat{P}^*}\ket{\psi_{\secp,r}^{\widehat{P}^*}} + \negl(\secp)\\ \leq & \E_r \left[\bra{\psi_{\secp,r}^{\widehat{P}^*}}\Pi_{\secp,r,1}^{\widehat{P}^*}\ket{\psi_{\secp,r}^{\widehat{P}^*}}\right] + \negl(\secp) \\ \leq & \negl(\secp),
\end{align*}

where the third inequality follows from property 2 of the commit-challenge-response protocol, since by definition $\E_r \left[\bra{\psi_{\secp,r}^{\widehat{P}^*}}\Pi_{\secp,r,0}^{\widehat{P}^*}\ket{\psi_{\secp,r}^{\widehat{P}^*}}\right] = 1$. On the other hand,

\begin{align*}
    &\frac{1}{|\cR|}\sum_{r \in \cR_{\mathsf{term},\secp}}\bra{\psi_{\secp,r}^{P^*}}\Pi_{\secp,r,0}^{P^*}\Pi_{\secp,r,1}^{P^*}\Pi_{\secp,r,0}^{P^*}\ket{\psi_{\secp,r}^{P^*}} \\ = &\E_r\left[\bra{\psi_{\secp,r}^{P^*}}\Pi_{\secp,r,0}^{P^*}\Pi_{\secp,r,1}^{P^*}\Pi_{\secp,r,0}^{P^*}\ket{\psi_{\secp,r}^{P^*}}\right] - \frac{1}{|\cR|}\sum_{r \notin \cR_{\mathsf{term},\secp}}\bra{\psi_{\secp,r}^{P^*}}\Pi_{\secp,r,0}^{P^*}\Pi_{\secp,r,1}^{P^*}\Pi_{\secp,r,0}^{P^*}\ket{\psi_{\secp,r}^{P^*}} \\ \geq & 1/p(\secp) - 1/p(\secp)^2,
\end{align*}

which is a contradiction, completing the proof.
\end{proof}

\begin{theorem}[\cite{TCC:ACGH20}]
Consider the $\secp$-fold parallel repetition of any commit-challenge-response protocol with computationally orthogonal projectors. The probability that the verifier accepts all $\secp$ parallel repetitions of the protocol is $\negl(\secp)$.
\end{theorem}

\begin{proof}
Let $\cR$ be the randomness space of the single repetition protocol, and $r = (r_1,\dots,r_\secp) \in \cR^{\otimes \secp}$ be verifier randomness for the $\secp$-fold parallel repetition. Now, any non-uniform prover $P^*$ can be described by states $\left\{\ket{\psi^{P^*}_{\secp,r}}\right\}_{\secp,r}$ and families of unitaries $\left\{U_{\secp,c}^{P^*}\right\}_{\secp,c}$, where $c \in \{0,1\}^\secp$ ranges over all of the verifier challenges.

For each $c \in \{0,1\}^\secp$, define \[\Pi^{P^*}_{\secp,r,c} \coloneqq {U^{P^*}_{\secp,c}}^\dagger \left(V_{\secp,r_1,c_1} \otimes \dots \otimes V_{\secp,r_\secp,c_\secp}\right)U^{P^*}_{\secp,c}.\]

\begin{claim}\label{claim:cross-terms}
For any $c_1 \neq c_2 \in \{0,1\}^\secp$, \[\E_r\left[\bra{\psi^{P^*}_{\secp,r}}\Pi^{P^*}_{\secp,r,c_2}\Pi^{P^*}_{\secp,r,c_1} + \Pi^{P^*}_{\secp,r,c_1}\Pi^{P^*}_{\secp,r,c_2}\ket{\psi^{P^*}_{\secp,r}}\right] = \negl(\secp).\]
\end{claim}

\begin{proof}
Suppose there exists $i \in [\secp]$ such that $(c_1)_i = 1$ and $(c_2)_i = 0$ (the other case is symmetric). Since for any quantum state $\ket{\psi}$ and two projectors $\Pi_1,\Pi_2$,  \[\bra{\psi}\Pi_2\Pi_1 + \Pi_1\Pi_2\ket{\psi} \leq 2|\bra{\psi}\Pi_2\Pi_1\ket{\psi}| \leq 2\bra{\psi}\Pi_2\Pi_1\Pi_2\ket{\psi}^{1/2},\] it then suffices (by Jensen's inequality) to show that \[\E_r\left[\bra{\psi^{P^*}_{\secp,r}}\Pi^{P^*}_{\secp,r,c_2}\Pi^{P^*}_{\secp,r,c_1}\Pi^{P^*}_{\secp,r,c_2}\ket{\psi^{P^*}_{\secp,r}}\right] = \negl(\secp).\]

To see this, let

\[V_{\secp,r_i,b}^{(i)} \coloneqq \mathbb{I} \otimes \dots \otimes \mathbb{I} \otimes V_{\secp,r_i,b} \otimes \mathbb{I} \otimes \dots \otimes \mathbb{I},\] for $i \in [\secp], b \in \{0,1\}$, and observe that 

\begin{align*}
    &\E_r\left[\bra{\psi^{P^*}_{\secp,r}}\Pi^{P^*}_{\secp,r,c_2}\Pi^{P^*}_{\secp,r,c_1}\Pi^{P^*}_{\secp,r,c_2}\ket{\psi^{P^*}_{\secp,r}}\right] \\ \leq &\E_r\left[\bra{\psi^{P^*}_{\secp,r}}{{U^{P^*}_{\secp,c_2}}}^\dagger V^{(i)}_{\secp,r_i,0}U^{P^*}_{\secp,c_2}{{U^{P^*}_{\secp,c_1}}}^\dagger V^{(i)}_{\secp,r_i,1}U^{P^*}_{\secp,c_1}{{U^{P^*}_{\secp,c_2}}}^\dagger V^{(i)}_{\secp,r_i,0}U^{P^*}_{\secp,c_2}\ket{\psi^{P^*}_{\secp,r}}\right] \\ = &\E_{r_i}\left[\bra{{\widehat{\psi}}^{P^*}_{\secp,{r_i}}}{{U^{P^*}_{\secp,c_2}}}^\dagger V^{(i)}_{\secp,r_i,0}U^{P^*}_{\secp,c_2}{{U^{P^*}_{\secp,c_1}}}^\dagger V^{(i)}_{\secp,r_i,1}U^{P^*}_{\secp,c_1}{{U^{P^*}_{\secp,c_2}}}^\dagger V^{(i)}_{\secp,r_i,0}U^{P^*}_{\secp,c_2}\ket{{\widehat{\psi}}^{P^*}_{\secp,r_i}}\right]\\ = &\negl(\secp),
\end{align*}

where for each $r_i \in \cR$, $\ket{\widehat{\psi}_{\secp,r_i}^{P^*}}$ is the purification of the mixed state (written in ensemble form) \[\left\{\frac{1}{|\cR|^{\secp - 1}},\ket{\psi_{\secp,(r_1,\dots,r_\secp)}^{P^*}}\right\}_{(r_1,\dots,r_{i-1},r_{i+1},\dots,r_\secp) \in \cR^{\otimes \secp - 1}},\] and the final equality follows from the computational orthogonal projectors property of the commit-challenge-response protocol. Indeed, one can define an efficient prover $P^*_i$ for the $i$'th iteration of the commit-challenge-response protocol by defining $U_{\secp,0}^{P^*_i} \coloneqq U_{\secp,c_2}^{P^*}$ and $U_{\secp,1}^{P^*_i} \coloneqq U_{\secp,c_1}^{P^*}$ and noting that $\ket{\widehat{\psi}_{\secp,r_i}^{P^*}}$ is efficient to prepare while interacting with the $i$'th iteration of $\cV$, by running $P^*$ and $\secp-1$ coherently executed copies of $\cV$. 

\end{proof}

Now observe that the probability the verifier accepts the parallel repeated protocol is

\begin{align*}
    &\frac{1}{2^\secp}\E_{r}\left[\bra{\psi^{P^*}_{\secp,r}}\sum_{c \in \{0,1\}^\secp}\Pi^{P^*}_{\secp,c,r}\ket{\psi^{P^*}_{\secp,r}}\right] \\ \leq &\frac{1}{2^\secp}\E_{r}\left[\left(\bra{\psi^{P^*}_{\secp,r}}\left(\sum_{c \in \{0,1\}^\secp}\Pi^{P^*}_{\secp,c,r}\right)^2\ket{\psi^{P^*}_{\secp,r}}\right)^{1/2}\right] \\ \leq &\frac{1}{2^\secp}\E_{r}\left[\left(\sum_{c \in \{0,1\}^\secp}\bra{\psi^{P^*}_{\secp,r}}\Pi^{P^*}_{\secp,c,r}\ket{\psi^{P^*}_{\secp,r}}\right)^{1/2}\right] \\ &+\frac{1}{2^\secp} \left(\sum_{\{c_1,c_2\} \in (\{0,1\}^{\secp})^2}\E_r\left[\bra{\psi^{P^*}_{\secp,r}}\Pi^{P^*}_{\secp,c_2,r}\Pi^{P^*}_{\secp,c_1,r} + \Pi^{P^*}_{\secp,c_1,r}\Pi^{P^*}_{\secp,c_2,r}\ket{\psi^{P^*}_{\secp,r}}\right]\right)^{1/2} \\ \leq &\frac{1}{2^{\secp/2}} + \frac{1}{2^\secp}\left(\sum_{\{c_1,c_2\} \in (\{0,1\}^{\secp})^2} \negl(\secp)\right)^{1/2} = \negl(\secp),
\end{align*}
where the first inequality holds because $\ket{\psi^{P^*}_{\secp,r}}\bra{\psi^{P^*}_{\secp,r}} \preceq \mathbb{I}$, the second inequality uses Jensen's inequality and the fact that projectors are idempotent, and the third inequality follows from \cref{claim:cross-terms}.

\end{proof}

\section{Proof of~\cref{claim:extracted-state-distribution}}
\label{sec:extracted-dist}

We now prove~\cref{claim:extracted-state-distribution}, which is restated below for convenience.

\begin{claim}
For all $(u,v)\in \{0,1\}^R\times \{0,1\}^S$ it holds that 
 \begin{equation}
\Tr\big(\Pi^{\sigma_x}_{u} \Pi^{\sigma_z}_{v} \btau\big) = \E_{u' \in \{0,1\}^{R}}  \bra{\psi} \Pi^{Z}_{v} Z(u') \Pi^{X}_{u' \oplus u} Z(u') \Pi^{Z}_{v} \ket{\psi}.
  \end{equation}
\end{claim}

\begin{proof}

Using the definition of $\btau$, we get
\begin{align}
    \Tr(\Pi^{\sigma_x}_{u} \Pi^{\sigma_z}_{v} \btau) &= 2^{-2N} \sum_{r', s', r'', s'' \in \{0,1\}^N} \Big( \bra{\psi} Z(s') X(r'\oplus r'') Z(s'') \ket{\psi}_{\RegH} \notag\\ & \bra{\phi^+}^{\tensor N} \left(\sigma_z(s') \sigma_x(r'\oplus r'') \sigma_z(s'')\right)_{\RegA_1} \tensor \left(\Pi^{\sigma_x}_{u} \Pi^{\sigma_z}_{v} \right)_{\RegA_2} \ket{\phi^+}^{\tensor N}\Big)\notag \\
    &= 2^{-2N} \sum_{r', s', r'', s'' \in \{0,1\}^N} (-1)^{(r' \oplus r'') \cdot s''} \Big( \bra{\psi} Z(s') X(r'\oplus r'') Z(s'') \ket{\psi}_{\RegH} \notag \\ & \bra{\phi^+}^{\tensor N} \left(\sigma_z(s' \oplus s'') \sigma_x(r'\oplus r'') \right)_{\RegA_1} \tensor \left(\Pi^{\sigma_x}_{u} \Pi^{\sigma_z}_{v} \right)_{\RegA_2} \ket{\phi^+}^{\tensor N}\Big).\label{eq:probability-expansion}
\end{align}
However, most of the terms in \cref{eq:probability-expansion} are zero: observe that when $(r' \oplus r'')_j \neq 0$ for any $j \in S$, or $(s' \oplus s'')_j \neq 0$ for any $j \in R$, we have
\begin{align*}
    \bra{\phi^+}^{\tensor N} \left(\sigma_z(s' \oplus s'') \sigma_x(r'\oplus r'') \right)_{\RegA_1} \tensor \left(\sigma_x(u) \sigma_z(v) \right)_{\RegA_2} \ket{\phi^+}^{\tensor N} =0.
\end{align*}
We can therefore rewrite \cref{eq:probability-expansion} using the following change of variables:
\begin{itemize}
    \item Since $s' \oplus s''$ must be $0$ on $R$, the restriction of $s'$ and $s''$ to $R$ must be the same vector $u' \in \{0,1\}^R$. Let the restriction of $s'$ and $s''$ to indices in $S$ be $v',v'' \in \{0,1\}^S$ respectively.
    \item Since $r' \oplus r''$ must be $0$ on $S$, let $u'' \in \{0,1\}^R$ denote the restriction of $r' \oplus r''$ to indices in $R$. Note that for each $u''$, there are $2^N$ choices of $(r',r'')$ satisfying $u'' = r' \oplus r''$.
\end{itemize}

By a straightforward calculation, we have for all $u'' \in \{0,1\}^R$ and all $s',s'' \in \{0,1\}^N$ that
\begin{align*}
    \sum_{\substack{r',r'' \in \{0,1\}^N \\ (r' \oplus r'') = u''}} \bra{\phi^+}^{\tensor N} \left(\sigma_z(s' \oplus s'') \sigma_x(r'\oplus r'') \right)_{\RegA_1} \tensor \left(\Pi^{\sigma_x}_{u} \Pi^{\sigma_z}_{v} \right)_{\RegA_2} \ket{\phi^+}^{\tensor N} &= (-1)^{u'' \cdot u + (s'\oplus s'')v}.
\end{align*}
Plugging this into \cref{eq:probability-expansion}, and using the fact that $(-1)^{(s' \oplus s'')v} = (-1)^{(v' \oplus v'')v}$, we obtain
\begin{align*}
    \Tr(\Pi^{\sigma_x}_{u} \Pi^{\sigma_z}_{v} \btau) &=  2^{-2N} \sum_{\substack{u',u'' \in \{0,1\}^R \\ v',v'' \in \{0,1\}^S}} (-1)^{(u \oplus u') \cdot u'' + (v'\oplus v'')v} \Big( \bra{\psi} Z(v') Z(u') X(u'') Z(u') Z(v'') \ket{\psi}_{\RegH}\Big)\\
    &= \E_{u' \in \{0,1\}^{R}}  \bra{\psi} \Pi^{Z}_v Z(u') \Pi^{X}_{u \oplus u'} Z(u') \Pi^{Z}_v \ket{\psi}
\end{align*}
where the second equality follows from plugging in the definitions of $\Pi^Z_{v}$ and $\Pi^{X}_{u \oplus u'}$.
\end{proof}

\section{Proof of~\cref{claim:d1}}
\label{sec:distinguish-marginals}

We now prove~\cref{claim:d1}, which we restate below for convenience.

\begin{claim}
Let $k=k(\lambda)$ be a positive integer-valued function of a security parameter $\lambda$. 
Let $\{D_{0,\lambda}\}_{\lambda\geq 1}$ and $\{D_{1,\lambda}\}_{\lambda \geq 1}$ be families of distributions on $\{0,1\}^{k+1}$ such that the marginal distributions $D_{0,\lambda}'$ and $D_{1,\lambda}'$ of $D_{0,\lambda}$ and $D_{1,\lambda}$ respectively on the first $k$ bits are computationally indistinguishable. Suppose that $D_{0,\lambda}$ and $D_{1,\lambda}$ are computationally \emph{distinguishable}. Then there is an efficiently computable binary-outcome POVM $\{M,\Id-M\}$ acting on $k$ qubits such that
\[\Big| \E_{x\sim D_{0,\lambda}} (-1)^{x_{k+1}} \bra{x_{\leq k}}M \ket{x_{\leq k}} 
- \E_{x\sim D_{1,\lambda}} (-1)^{x_{k+1}} \bra{x_{\leq k}}M\ket{x_{\leq k}} \Big| >  \frac{1}{\poly(\lambda)}.\]
\end{claim}

\begin{proof}
By assumption there exists an efficient distinguisher between $D_0$ and $D_1$ (for simplicity we omit the index $\lambda$ from the notation). Let $A$ be a circuit for the distinguisher: $A$ has $(k+1)$ input qubits as well as $m$ ancilla qubits, and a designated output qubit. Let $\Pi_1$ be the projection on the output qubit being equal to $1$. Suppose without loss of generality that
\begin{equation}\label{eq:sup}
\E_{x\sim D_0} \bra{x,0^m} A^\dagger \Pi_1 A \ket{x,0^m} > \E_{x\sim D_1} \bra{x,0^m} A^\dagger \Pi_1 A \ket{x,0^m}+\frac{1}{q},
\end{equation}
for some polynomial $q=q(\lambda)$. Letting $\ket{b}_{k+1}$ denote the $(k+1)$-st qubit, we can write
\begin{align*}
    \E_{x \sim D_1} \bra{x,0^m}A^\dagger \Pi_1 A\ket{x,0^m} &= \E_{x \sim D_1} x_{k+1} \bra{x_{\leq k},0^m}\bra{1}_{k+1} A^\dagger \Pi_1 A\ket{1}_{k+1}\ket{x_{\leq k},0^m}\\
    & \;\;\;+ \E_{x \sim D_1} (1-x_{k+1}) \bra{x_{\leq k},0^m}\bra{0}_{k+1} A^\dagger \Pi_1 A\ket{0}_{k+1}\ket{x_{\leq k},0^m}.
\end{align*}
Let $M_{b} \coloneqq \bra{b,0^m} A^\dagger \Pi_1 A \ket{b,0^m}$ (where $b$ corresponds to the $(k+1)$-st qubit); note that $M_b$ is a positive semi-define. We can rewrite the right-hand-side as
\begin{align}
    \label{eqn:D1-equation}
    \E_{x \sim D_1} x_{k+1}\bra{x_{\leq k}} (M_1 - M_0) \ket{x_{\leq k}} + \E_{x \sim D_1} \bra{x_{\leq k}} M_0 \ket{x_{\leq k}}.
\end{align}
Using a similar expansion while taking the expectation under $D_0$ yields
\begin{align}
    \label{eqn:D0-equation}
    \E_{x \sim D_0} x_{k+1}\bra{x_{\leq k}} (M_1 - M_0) \ket{x_{\leq k}} + \E_{x \sim D_0} \bra{x_{\leq k}} M_0 \ket{x_{\leq k}}.
\end{align}
Plugging~\cref{eqn:D1-equation,eqn:D0-equation} into~\cref{eq:sup} gives
\begin{align*}
    \E_{x \sim D_0} x_{k+1}&\bra{x_{\leq k}} (M_1 - M_0) \ket{x_{\leq k}}- \E_{x \sim D_1} x_{k+1}\bra{x_{\leq k}} (M_1 - M_0) \ket{x_{\leq k}}\\
    &>\E_{x \sim D_1} \bra{x_{\leq k}} M_0 \ket{x_{\leq k}} - \E_{x \sim D_0} \bra{x_{\leq k}} M_0 \ket{x_{\leq k}} + \frac{1}{q}.
\end{align*}
For $b \in \{0,1\}$, note that $\{M_b,\Id-M_b\}$ is an efficiently computable POVM since it can be performed by initializing the $(k+1)$-st qubit to $\ket{b}$, the ancilla qubits to $\ket{0^m}$, applying $A$, and measuring whether the output qubit is $1$. Since $D_0'$ and $D_1'$ are computationally indistinguishable, we have
\begin{align*}
    \E_{x \sim D_0} x_{k+1}&\bra{x_{\leq k},0^m} (M_1 - M_0) \ket{x_{\leq k},0^m}- \E_{x \sim D_1} x_{k+1}\bra{x_{\leq k},0^m} (M_1 - M_0) \ket{x_{\leq k},0^m}\\
    &> \frac{1}{q} - \negl(\secp).
\end{align*}
We observe that there must exist $b \in \{0,1\}$ such that when $M = M_b$, we have
\begin{align*}
    \Big| \E_{x \sim D_0} x_{k+1}&\bra{x_{\leq k},0^m} M \ket{x_{\leq k},0^m}- \E_{x \sim D_1} x_{k+1}\bra{x_{\leq k},0^m} M \ket{x_{\leq k},0^m} \Big| \\
    &> \frac{1}{\poly(\secp)}.
\end{align*}
Finally, by plugging in the identity $(-1)^b=1-2b$ for $b \in \{0,1\}$ and appealing once again to the indistinguishability of $D_0'$ and $D_1'$, we conclude that 
\begin{align*}
    \Big| \E_{x \sim D_0} (-1)^{x_{k+1}}&\bra{x_{\leq k},0^m} M \ket{x_{\leq k},0^m}- \E_{x \sim D_1} (-1)^{x_{k+1}}\bra{x_{\leq k},0^m} M \ket{x_{\leq k},0^m} \Big| \\
    &> \frac{1}{\poly(\secp)}.
\end{align*}
\end{proof}

\end{document}